\documentclass[onecolumn]{IEEEtran}
\usepackage{multirow}
\usepackage{wrapfig}
\usepackage{float}
\usepackage{tabularx}
\usepackage{longtable}
\usepackage{color}
\definecolor{menucolor}{rgb}{0.1,0.52,0.47}
\definecolor{anchorcolor}{rgb}{0.85,0.37,0.01}
\definecolor{runcolor}{rgb}{0.46,0.44,0.701}
\definecolor{linkcolor}{rgb}{0.3,0.55,0.01}
\definecolor{urlcolor}{rgb}{0.12,0.47,0.70}
\definecolor{citecolor}{rgb}{0.55,0.36,0.01}
\definecolor{filecolor}{rgb}{0.4,0.4,0.4}
\usepackage[colorlinks=true, urlcolor=urlcolor, linkcolor=linkcolor,citecolor=citecolor,filecolor=filecolor, anchorcolor=anchorcolor,menucolor=menucolor]{hyperref}
\usepackage{graphicx}
\usepackage{verbatim}
\usepackage{booktabs}
\usepackage{tikz}
\ifCLASSOPTIONcompsoc
  \usepackage[caption=false,font=normalsize,labelfont=sf,textfont=sf]{subfig}
\else
\usepackage[caption=false,font=footnotesize]{subfig}
\fi
\usepackage{array}

\usepackage{multirow}
\usepackage{wrapfig}
\usepackage{tabularx}
\usepackage{longtable}
\usepackage{hyperref}
\usepackage{amsmath}
\usepackage{amsthm}
\usepackage{graphicx}
\graphicspath{{./image/jpeg/}}
\usepackage{verbatim}
\usepackage[nocompress]{cite}

\usepackage{xcolor}
\usepackage{tikz}

\newcommand{\bLambda}{\boldsymbol{\Lambda}}

\newcommand{\bSigma}{\boldsymbol{\Sigma}}
\newcommand{\bsigma}{\boldsymbol{\sigma}}
\newcommand{\bPsi}{\boldsymbol{\Psi}}

\newcommand{\bmu}{\boldsymbol{\mu}}
\newcommand{\bnu}{\boldsymbol{\nu}}

\newcommand{\bTheta}{\boldsymbol{\Theta}}

\newcommand{\bB}{\boldsymbol{B}}

\newcommand{\bc}{\boldsymbol{c}}
\newcommand{\bD}{\boldsymbol{D}}
\newcommand{\bff}{\boldsymbol{f}}
\newcommand{\bF}{\boldsymbol{F}}
\newcommand{\bV}{\boldsymbol{V}}
\newcommand{\bg}{\boldsymbol{g}}

\newcommand{\bI}{\boldsymbol{I}}

\newcommand{\bS}{\boldsymbol{S}}

\newcommand{\bT}{\boldsymbol{T}}

\newcommand{\be}{\boldsymbol{e}}

\newcommand{\bq}{\boldsymbol{q}}
\newcommand{\bQ}{\boldsymbol{Q}}
\newcommand{\bR}{\boldsymbol{R}}
\newcommand{\bW}{\boldsymbol{W}}
\newcommand{\bu}{\boldsymbol{u}}
\newcommand{\bv}{\boldsymbol{v}}

\newcommand{\bx}{\boldsymbol{x}}
\newcommand{\bX}{\boldsymbol{X}}

\newcommand{\bY}{\boldsymbol{Y}}
\newcommand{\bZ}{\boldsymbol{Z}}

\newcommand{\bzero}{\boldsymbol{0}}

\newcommand{\mL}{\mathcal L}

\newcommand{\mU}{\mathcal U}
\newcommand{\mP}{\mathcal P}
\newcommand{\mN}{\mathcal N}

\newcommand{\mI}{\mathcal I}

\newcommand{\br}{\boldsymbol{r}}

\newcommand{\bb}{\boldsymbol{b}}
\let\muold\mu
\renewcommand{\mu}{{\bm\muold}}
\let\LambdaOLD\Lambda
\renewcommand{\Lambda}{{\bm\LambdaOLD}}
\let\Psiold\Psi
\renewcommand{\Psi}{{\bm\Psiold}}
\let\Usiold\Upsilon
\renewcommand{\Upsilon}{{\bm\Usiold}}
\let\Sigmaold\Sigma
\renewcommand{\Sigma}{{\bm\Sigmaold}}
\let\Gammaold\Gamma
\renewcommand{\Gamma}{{\bm\Gammaold}}
\let\sigmaold\sigma
\renewcommand{\sigma}{{\bm\sigmaold}}

\newcommand{\Tr}{\mathrm{Tr}~}
\def\shalf{\mbox{{\footnotesize$\frac{1}{2}$}}}
\newcommand{\normal}{{\mathcal N}}
\newcommand{\sPN}{{\cal PN}}
\newcommand{\sS}{{\cal S}}
\newcommand{\bbR}{\mathbb{R}}
\newcommand{\bbE}{\mathbb{E}}
\newcommand{\corr}{{\bm{\mathfrak{R}}}}
\newcommand{\diag}{\mathrm{diag}}

\newtheorem{theorem}{Theorem}
\newtheorem{proposition}{Proposition}

\usepackage[many]{tcolorbox}
\newtcolorbox{mybox}[1][]{
    tikznode boxed title,
    enhanced,
    arc=0mm,
    interior style={white},
    attach boxed title to top center= {yshift=-\tcboxedtitleheight/2},
    fonttitle=\bfseries,
    colbacktitle=white,coltitle=black,
    boxed title style={size=normal,colframe=white,boxrule=0pt},
    #1}

\usepackage{amsthm}

\newcommand{\citep}{\cite}
\newcommand{\citet}{\cite}
\newcommand{\citeyear}{\cite}
\newcommand{\citeauthor}{\cite}

\usepackage{algorithm,algorithmic}

\usepackage{array,booktabs,longtable,tabularx}
\newcolumntype{L}{>{\raggedright\arraybackslash}X}
\usepackage{siunitx}
\usepackage{caption}

\usepackage{bm}
\usepackage{amsfonts}

\begin{document}
%
\title{Exploratory Factor Analysis of Data on a Sphere}

\author{{Fan~Dai, Karin~S.~Dorman, Somak~Dutta and
     and Ranjan~Maitra}
\thanks{F. Dai is with the Department of Mathematical Sciences at the Michigan Technological University, Houghton,
  Michigan 49931, USA. e-mail: fand@mtu.edu.}
\thanks{K.S. Dorman is with the Departments of Statistics and Genetics
  Development and Cell Biology, and S. Dutta and R. Maitra are with
  the Department of Statistics,  
  all 
at Iowa State University, Ames, Iowa 50011, USA. e-mail:
\{kdorman,somakd,maitra\}@iastate.edu.}
 }

 \IEEEcompsoctitleabstractindextext{
%
\maketitle
Data on high-dimensional spheres arise frequently in many disciplines either naturally or as a consequence of preliminary processing and can have intricate dependence structure that needs to be understood. We develop exploratory factor analysis of the projected normal distribution to explain the variability in such data using a few easily interpreted latent factors. Our methodology provides maximum likelihood estimates through a novel fast alternating expectation profile conditional maximization algorithm. Results on simulation experiments on a wide range of settings are uniformly excellent. Our methodology provides interpretable and insightful results when applied to tweets with the  {\tt \#MeToo} hashtag in early December 2018, to time-course functional Magnetic Resonance  Images of the average pre-teen brain at rest, to characterize handwritten digits, and to gene expression data from cancerous cells in the Cancer Genome Atlas. 

  \begin{IEEEkeywords}
AECM, $l_2$-normalization, Lanczos algorithm, matrix-free
computations, {\tt \#MeToo},
MNIST, projected normal distribution, resting state fMRI,
term-document matrix, TCGA.
\end{IEEEkeywords}}
\IEEEdisplaynotcompsoctitleabstractindextext
\maketitle

\vspace{-0.2em}
\section{Introduction}
\label{sec:intro-ch3}
Observations that lie on a unit
sphere~\citep{watson83,fisheretal93,jammalamadakaandsengupta01,mardiaandjupp00} arise
naturally in disciplines such as astronomy ({\em e.g.}, the orientation of planetary and cometary orbits),
biology ({\em e.g.}, homing behavior
or the direction of capillaries in tissue),  geology ({\em e.g.}, the
orientation of paleomagnetism in rocks or the axes of sand dunes), 
meteorology ({\em e.g.}, wind direction), oceanography ({\em e.g.},
the direction of ice flows) or physics ({\em e.g.}, the
direction of particle diffusion).  Many
distributions~\citep{jammalamadakaandsengupta01,mardiaandjupp00} have
been proposed to model such observations, but the von Mises (or
circular normal) 
distribution~\citep{vonmises18} is commonly used to model angular
data, while observations on a three-dimensional sphere are often modeled with the
Fisher distribution~\citep{fisher53}. The von Mises and Fisher
distributions are the 2- and 3-dimensional cases of the Langevin~\citep{watson83} distribution
-- also called the Fisher-von Mises~\citep{senguptaandmaitra98} or
von Mises-Fisher~\citep{mardiaandjupp00} distribution -- that
models data on a multi-dimensional sphere. However, this distribution
is unable to adequately characterize relationships between variables, especially in the context of modern applications such as
text processing, genetics or imaging, where observations end up on unit spheres after appropriate processing. We
describe four motivating applications with high-dimensional data on 
unit spheres that can benefit from a distribution that can characterize
dependence between variables. 

\subsection{Motivating Applications}
\label{sec:motivating}
We introduce four scenarios where preprocessing yields
unit-sphered data. In each case, the goal is to
characterize variability and relationships between variables, whether
they are words in tweets containing the hashtag {\tt \#MeToo}, the
time-points in temporal functional Magnetic Resonance Imaging (fMRI)
data of an averaged typically developing child brain at rest, pixels in
bitmap images, or cancerous gene expressions in the Cancer Genome Atlas (TCGA).

\subsubsection{Summarizing  {\tt \#MeToo} tweets}
  \label{sec:intro.me2}
The {\tt \#MeToo} hashtag on Twitter provided an online forum for
women to narrate their harassment or abuse at the hands of
men in positions of authority. The term had been percolating in social
media since 2006, but received a major fillip in 2017 with the outing of high-profile abusers in
entertainment, politics, business and media. Our interest is
in characterizing variability in tweets having the 
hashtag and to countenance its description by a few underlying
factors.
Each tweet is akin to a text document, so
all the tweets  can be processed into a weighted {\em document-term
frequency
matrix}~\citep{saltonandbuckley88,kolda97,dhillonandmodha01}, with the
most common weighting scheme normalizing each tweet to lie on a
high-dimensional unit sphere~\citep{singhaletal96}. Modeling each
tweet  in terms of the Langevin distribution, as commonly done for
text data~\citep{banerjeeetal05}, would
require marginal exchangeability and essentially no correlation
between words and a characterization totally at odds with the reality
of tweet structure and characteristics. 
We return to this dataset in Section~\ref{sec:app-ch3-metoo}.  

\subsubsection{Characterizing the  brain at rest}
  \label{sec:intro.fmri}
Characterizing changes in cerebral blood-flow at
rest~\citep{biswaletal95}, or in the absence of any stimulus or task,
is of interest to  researchers~\citep{biswal12}. A
typical resting state experiment using fMRI~\citep{belliveauetal91,
  kwongetal92,bandettinietal93,lindquist08,lazar08} 
  involves collecting a time series of
  blood-oxygen-level-dependent (BOLD) images of the brain at rest,
  and after preprocessing, fitting a linear model at each voxel with
  motion and other structured noise as covariates. The time series of
  the residuals, or the part of the BOLD response vector unexplained
  by the linear model, is the resting state component of the
  dataset~\citep{foxetal05,coleetal10}. Functionally connected voxels
  may be detected from correlations between the
  residual time series sequences at each voxel, or equivalently by
  calculating the Euclidean distance between the centered and
  standardized time series that lie on the unit
  sphere~\citep{maitraandramler10}. However, it may also be of
  interest to characterize the
  distribution and variability of activity patterns in the resting
  state brain, for which the Langevin distribution is again inadequate. We revisit this application in Section~\ref{sec:app-ch3-fmri}.

  \subsubsection{Assessing variability in handwritten digits}
  \label{sec:intro.digits}
  The Modified National Institute of Standards and Technology (MNIST)
  database~\citep{LeCun2019digits}  has 70,000 normalized and
  anti-aliased $28\!\times 28$ bitmap images of handwritten digits of
  0 through 9, written by 
  500 writers. The standardization of each pixel value to the interval
  [0, 255], as a consequence of the digitization of
  the image, means that the raw photon value of each image pixel is no longer
  available. 
  The loss of a meaningful absolute scale for the image suggests use of correlation as a similarity measure.
  This similarity measure is akin to normalizing each image to
  lie on the unit sphere. Characterizing inter-pixel relationships can
  summarize variability in handwritten digits 
  but again, requires a distribution richer than the Langevin for the same
  reasons as before. We investigate this application further in
  Section~\ref{sec:app-ch3-digits}.

  \subsubsection{Identifying genetic pathways underlying cancer}
\label{sec:tgca.intro}
  The Cancer Genome Atlas (TCGA) dataset provides messenger
  ribonucleic acid (mRNA) sequencing data of tumor samples from more than 11,000 patients collected over a 12-year period. 
  After initial processing, each tumor has gene compositional data,
  where  each observation lies on a simplex, and a square root
  transformation on each feature maps the observation to lie on a unit
  sphere~\citep{daiandmuller18}. It is expected that a few
  groups of genes can explain variability between the tumor samples
  (and patients) and that identification of such gene groups could lead to
  future diagnostic tools and treatments \citep{bose2019chaudhari}. Once again, however, such
  relationships cannot be modeled by the Langevin distribution with its lack of flexible relationships between coordinates. We study this application in detail in Section~\ref{sec:app-ch3-rnaseq}.

Our four showcase applications point to the inadequacy of the Langevin
distribution in characterizing heterogeneity and relationships between
the observed variables that lie on a unit sphere. 
A more general approach to modeling data on a sphere involves taking
the inverse stereographic projection of a general multivariate normal
distribution~\citep{mardiaetal79,dortetbernadetandwicker08}. However,
this approach is not practical in high dimensions: indeed, in their application,
\citet{dortetbernadetandwicker08} use {\em a 
  priori} dimension reduction (into the first three principal
components) on the raw data as a preprocessing step before modeling
and analysis. Other distributions such as the complex
Bingham~\citep{bingham74}, Fisher-Bingham~\citep{kent82} and the
real/complex Watson distributions~\citep{dryden05,kimetal17} also
suffer from intractability with dimensions greater than three. 

A more flexible distributional family arises from the projected normal (PN)
distribution that is obtained by radial projection of the multivariate
normal (MVN) distribution onto the unit sphere~\citep{mardiaandjupp00}. In
general, the PN density  does not have a closed-form
expression. \citet{wangandgelfand13} and \citet{stumpfhauseretal17}
provided Bayesian inference for this model in two and higher
dimensions, respectively, under the (mild) constraint that the
first element of the dispersion matrix parameter of the projected normal density is unity. The structure of this dispersion matrix
parameter is of interest in many applications, so we investigate
if a small number of latent factors can explain the
variability in the (observed) data on a sphere. Our approach posits
parameter estimation from a projected normal
distribution as a missing
data problem, where the complete data is from the generative MVN
distribution that we specify in terms of a factor
model. Two Alternating Expectation Conditional Maximization (AECM)
algorithms~\citep{mengandvandyk97} are developed for 
parameter estimation within our framework. Because our 
interest stems from high-dimensional problems, we also
develop accelerated implementations that incorporate squared 
iterative (SQUAREM) methods~\citep{varadhanandroland08} to achieve superlinear convergence.
 
The remainder of this paper is organized as
follows. Section~\ref{sec:meth-ch3} develops the factor model and
analysis for data on 
a sphere (FADS) and 
proposes two novel AECM algorithms called FADS-D
(Expectation-Maximization, {\em i.e.} EM, with a duet of
latent variables) and FADS-P (EM with one
latent variable and a profile maximization step) for model-fitting. We
use the
extended Bayesian Information Criterion (eBIC) of
\citet{chenandchen08} to decide on the 
 number of factors. 
Section \ref{sec:sim-ch3}
evaluates the performance of FADS-P and FADS-D in simulation
experiments. We next apply our methodology to the motivating
applications of Section~\ref{sec:motivating}. Section \ref{sec:disc} discusses
possible extensions and future work. An online supplement, with
sections, tables, theorems and figures referenced here with the prefix “S”, is
available.

\section{Methodology}
\label{sec:meth-ch3}
\subsection{Background and preliminary development}
\subsubsection{The projected normal  (PN) distribution}
A popular method of constructing probability distributions on the unit
sphere is by radially projecting a multivariate distribution,
typically Gaussian and has been well-studied
in low-dimensional problems by \citet[p. 178]{mardiaandjupp00},
\cite{pukkilaandrao88}, \citet{paineetal18}, and, from a Bayesian perspective,
by \cite{wangandgelfand13} and \cite{stumpfhauseretal17}. The PN distribution $\sPN_p(\bmu,\bSigma)$ on the $p$-unit sphere $\sS_{p-1}$ is defined as the distribution of $\bY/\|\bY\|$ where $\bY$ has an MVN distribution with mean $\bmu$ and covaraince matrix $\Sigma.$
Consequently, the $\sPN_p(\bmu,\bSigma)$ density function is 
\begin{equation}\label{eqn:pdfPN}
 f({\bx};\bmu,\bSigma) =
 (2\pi)^{-\frac{p}{2}}|\bSigma|^{-\frac{1}{2}}\exp{\left\{-\displaystyle{\frac{\bmu^\top\bSigma^{-1}\bmu}{2}} + \displaystyle{\frac{m^2}{2v}}\right\}}
 {\int_{0}^{\infty}{R^{p-1}\exp{\left\{-\displaystyle{\frac{1}{2v}}\Big(R-m\Big)^2\right\}}}\, dR},
\end{equation}
where $m = {{\bx}}^\top\Sigma^{-1}\mu/{{\bx}}^\top\Sigma^{-1}{{\bx}}$,
$v = 1/{{\bx}}^\top\Sigma^{-1}{{\bx}}$ (see
Theorem~\ref{theorem:pndensity}). The integral in
\eqref{eqn:pdfPN} arises from integrating out the unobserved length $R
= \|\bY\|$ and is proportional to the $p$th raw moment of a
univariate normal distribution truncated above zero. However, this
family of distributions is not identifiable \citep{mardiaandjupp00},
so there is need for further constraints, such as fixing a diagonal
entry of $\bSigma$ or, more commonly, setting  
$\bmu\in\sS_{p-1}$ that we adopt in this work. 
The matrix $\bSigma$ flexibly models the dependence among
coordinates, allowing different shapes and orientations. For example,
Fig.~\ref{fig-intro-sig} shows the isodensity lines of a 3D PN 
distribution with the same $\mu$ but three different $\Sigma$s.
\begin{figure}[h]
\vspace{-0.4in}
  \centering
  \mbox{
  \setcounter{subfigure}{-1}
  \subfloat[Spherical $\bSigma$]{
    \hspace{-0.05\textwidth}
    \begin{minipage}[b][][t]{.3\textwidth} 
      \raisebox{.5\height}{      \begin{minipage}[b][][t]{.45\textwidth} 
        {\footnotesize
        $\begin{bmatrix}
          1 & 0 & 0\\
          0 & 1 & 0\\
          0 & 0 & 1
        \end{bmatrix}$}
    \end{minipage}}
    \hspace{-0.165\textwidth}
      \subfloat{\includegraphics[width=.6\textwidth]{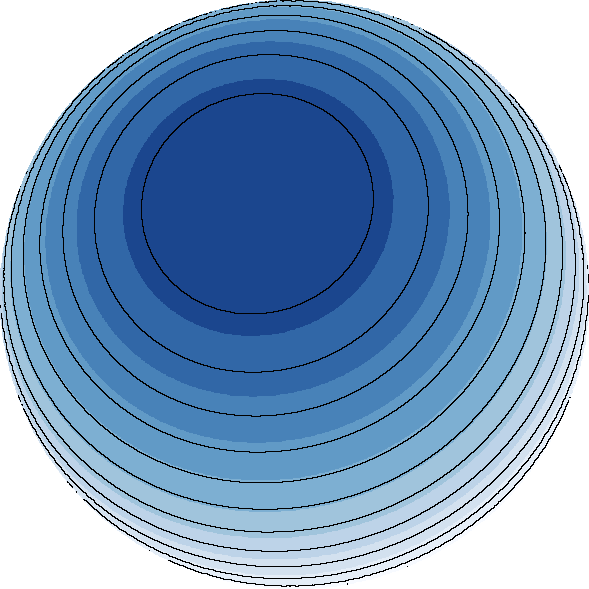}}
    \end{minipage}}%
  \setcounter{subfigure}{0}
  \subfloat[General diagonal $\bSigma$]{
    \hspace{-0.01\textwidth}
    \begin{minipage}[b][][t]{.3\textwidth} 
     \raisebox{.5\height}{ \begin{minipage}[b][][t]{.45\textwidth} 
{\footnotesize
          $ \begin{bmatrix}
            \frac12 & 0 & 0\\
            0 & \frac13 & 0\\
            0 & 0 & \frac14
          \end{bmatrix}$}
      \end{minipage}}
    \hspace{-0.125\textwidth}
    \subfloat{\includegraphics[width=.6\textwidth]{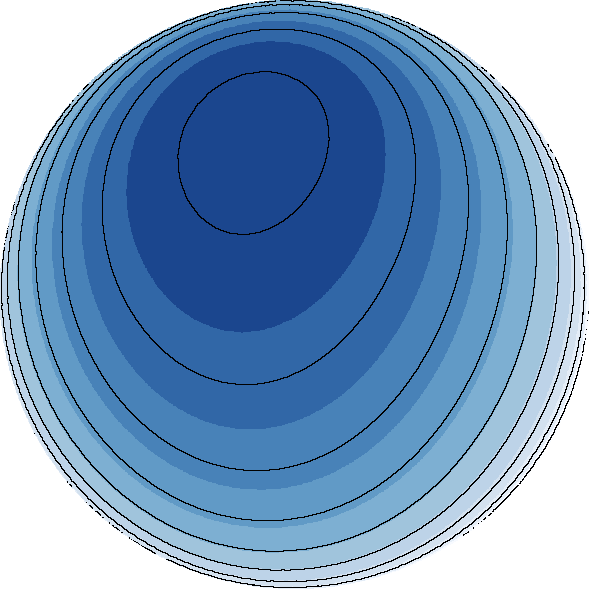}}
    \end{minipage}}%
  \setcounter{subfigure}{1}
  \subfloat[General-structured $\bSigma$]{
    \hspace{-0.01\textwidth}
    \begin{minipage}[b][][t]{.4\textwidth} 
      \centering
      \raisebox{.5\height}{      \begin{minipage}[b][][t]{.5875\textwidth} 
        {\footnotesize
          $ \begin{bmatrix}
            \frac14 &-\frac14  &\frac14\\
            -\frac14  &\frac25 &-\frac13\\
            \frac14 &-\frac13  &\frac25
        \end{bmatrix}$}
    \end{minipage}}
    \hspace{-0.22\textwidth}
  \subfloat{\includegraphics[width=.4625\textwidth]{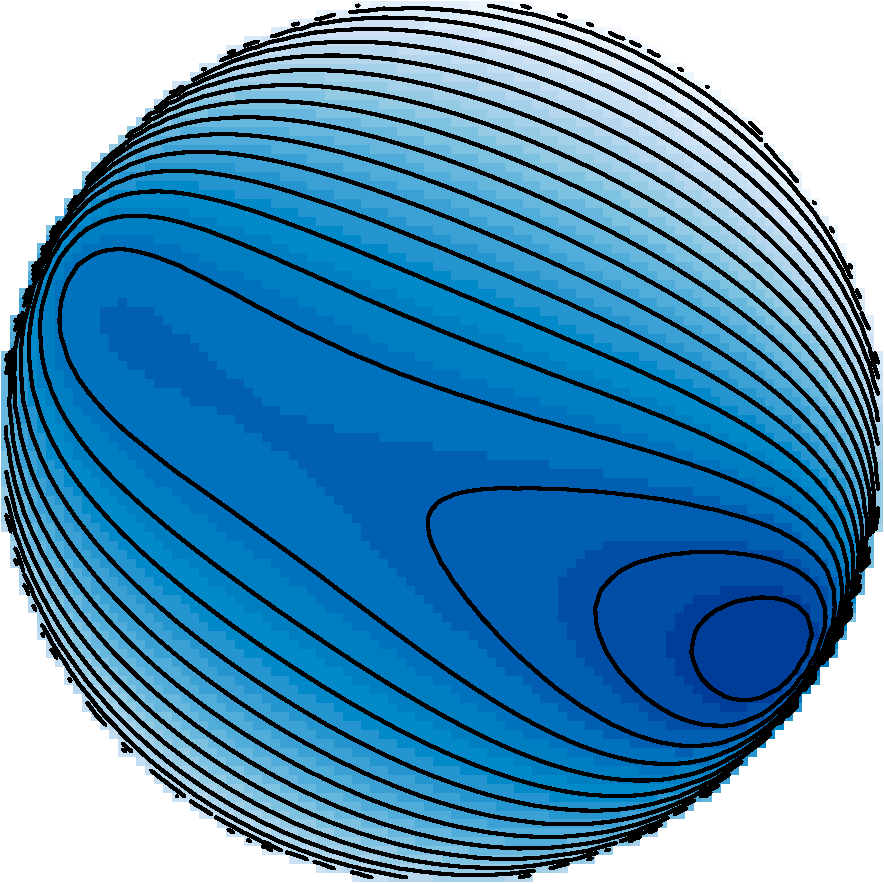}}
    \end{minipage}}%
  }
  \caption{Sample 3D PN densities (with higher values in darker
    shades) for three different $\bSigma$s.}
\label{fig-intro-sig}
\end{figure}
With spherical $\Sigma$, the isodensity lines on the surface of the
sphere are circular, but for the other two heteroscedastic
cases, we obtain ellipsoids of different eccentricities and orientations.
\citet{pukkilaandrao88} provide method-of-moments estimators for the
parameters of the PN distribution but this method is impractical for
high dimensions, especially when the sample size is relatively small.
An alternative approach to specifying general distributions on the sphere is by stereographic projection of the MVN
distribution~\citep{dortetbernadetandwicker08}. Both approaches result in multimodal densities as the diagonal entries of
$\bSigma$ increase. However, the stereographic
projection approach is not easily analyzed in high dimensions, so
\citet{dortetbernadetandwicker08} suggest reducing the dimensionality
of the MVN through principal components analysis before applying
stereographic projection. We provide computationally efficient EM
algorithms for maximum likelihood (ML) estimators for the parameters of the PN distribution that can be used in high dimensions. Our approach also decomposes $\bSigma$ into a factor structure, allowing for interpretation and further understanding of the variability in a dataset. We describe our factor model setup next.                              

\subsubsection{A latent factor model for the PN distribution}
\label{sec:meth-fa}
Let $\bX_1,\bX_2,\ldots,\bX_n$ be independent identically distributed (i.i.d.)
realizations from~\eqref{eqn:pdfPN}. Our latent factor model 
assumes $\bX_i = \bY_i/\|\bY_i\|$, where the generative latent variables $\bY_i$ have 
the decomposition
\begin{equation} \label{mod:fa}
{\bY}_i = \mu+\bLambda{\bZ}_i+\bm{\epsilon}_i, \qquad i=1,2,\ldots,n
\end{equation}
and $\bLambda$ is a $p\times q$ matrix called the factor loading
matrix with rank $q$, ${\bZ}_i$s are latent factors of dimension $q$ and
${\bm \epsilon}_i$s are residual noise vectors. Typically, $q \ll \min(n,p)$
and so \eqref{mod:fa} explains the variability in the $\bY_i$s in
terms of a few ($q$) underlying latent factors. Usually, 
${\bZ}_i$s are assumed independent $\mN_q(\bzero,\bI)$ and
${\bm\epsilon}_i$s independent $\mN_p(\bzero,\bPsi)$ random vectors, where
$\bPsi$ is a diagonal matrix with positive entries called the
uniquenesses. Hence, the covariance
matrix $\Sigma$ admits a lower-dimensional representation: $\Sigma =
\bLambda \bLambda^\top + \bPsi$. Parameter estimation in~\eqref{mod:fa}
has received attention
\citep{lawley40,joreskog66,lawleyetal71,mardiaetal79,anderson03,owenandwang15}
over several decades, with recent work~\citep{daietal20} providing near-instantaneous
solutions even for very high dimensions.


The Gaussian factor model \eqref{mod:fa} is not invariant to
orthogonal transformation~\citep{anderson03,mardiaetal79} unless all
uniquenesses are equal (i.e. $\bPsi=\sigmaold^2 \bI_p$ for some
$\sigmaold^2>0$). 
Also, \eqref{mod:fa} is not identifiable because both
$\bLambda$ and $\bLambda {\bm Q}$ give rise to the same model for any
orthogonal matrix $\bm Q$ and so additional constraints
are needed for identifiability. For computational simplicity, 
we assume that $\bLambda^\top\bPsi^{-1}\bLambda$ is diagonal with
decreasing diagonal entries. However, for better interpretability, and
as is standard in traditional factor analysis, we present  estimates of
$\bLambda$ after applying some appropriate orthogonal~({\em
  e.g.} varimax, quartimax or oblimin) rotation~\citep{costelloandosborne05}. 

\subsection{Estimation of parameters}
The PN factor model of Section~\ref{sec:meth-fa} has two
generative latent components -- the lengths $R_i=\|\bY_i\|$ and the
factors $\bZ_i$, $i=1,2,\ldots,n$. Given $q$ and these inherent
latencies, the EM algorithm provides a natural way to obtain
the ML estimators of $\mu,\bLambda$ and $\bPsi$. But vanilla EM is 
slow to converge, so we develop an 
AECM  algorithm~\citep{mengandvandyk97}
for parameter estimation using the  latent components
$(R_i,\bZ_i)$s. However, \citet{daietal20} recently showed a
matrix-free profile likelihood approach to be more computationally
efficient than EM~\citep{owenandwang15} for the Gaussian factor model.
Building on that approach, we develop a AECM algorithm,
where the expecation step (E-step) only integrates over the latent $R_i$s, while the conditional maximization step (CM-step)
maximizes over $\mu$ with the latent factor model parameters estimated after profiling the expected complete log-likelihood. Our
algorithms are further accelerated by SQUAREM~\citep{varadhanandroland08} to achieve superlinear convergence.

\subsubsection{An AECM algorithm with a duet of latent variables}
\label{sec:aecm-d}
The EM algorithm for ML estimation in a Gaussian factor
model~\citep{rubinandthayer82} can be naturally extended to the
PN factor model by considering both $\bR =
(R_1,R_2, \ldots, R_n)$ and $\bZ = [\bZ_1,\bZ_2,\ldots,\bZ_n]$ as
missing data. Then the $Q$-function, that is, the conditional
expectation of the complete data log-likelihood given the data $\bX$,
is
\begin{equation} \label{eq:exp-llk1}
        Q_D(\bmu,\bLambda,\bPsi;\bmu_t,\bLambda_{t},\bPsi_{t})
        = c - \shalf n[\log\det\bPsi
         + \Tr\{\bPsi^{-1} \bbE_{\bZ,\bR|\bX}({\bS}_{\bZ,\bR})\}],
\end{equation}
where $c$ is a constant free of $\bmu,\bLambda$ and $\bPsi$, and
$ \bbE_{\bZ,\bR|\bX}({\bS}_{\bZ,\bR})$ is the conditional expectation
of ${\bS}_{\bZ,\bR} = n^{-1}\sum^{n}_{i=1} (R_i{\bX}_i- \bmu -
\bLambda\bZ_i)(R_i{\bX}_i- \bmu-
\bLambda{\bZ}_i)^\top$ 
given $\bX$, at the current parameter values $(\mu_t,\bLambda_t,\bPsi_t)$.
The E-step needs this conditional expectation,
but it is straightforward to show 
(Section~\ref{sec:supp-fads-d}) that the
conditional distribution of $\bZ_i$ given $\bR$ and $\bX$ at the
current parameter values $(\bmu_t,\bLambda_t,\bPsi_t)$ is
$\mN_q(\bLambda_t^\top(\bLambda_t\bLambda_t^\top + \bPsi_t)^{-1}(R_i{\bX}_i-\bmu_t),
({\bI}+\bLambda_t^\top\bPsi_t^{-1}\bLambda_t)^{-1})$. Further, the conditional
expectations of $R_i$ and $R_i^2$ given $\bX$ are
$\bbE_{R_i|\bX_i}(R_i) = \mI_p(m_i,v_i)/\mI_{p-1}(m_i,v_i)$ and $\bbE_{R_i|\bX_i}(R_i^2) = \mI_{p+1}(m_i,v_i)/\mI_{p-1}(m_i,v_i)$ where
$m_i = \bX_i^\top (\bLambda_t\bLambda_t^\top+\bPsi_t)^{-1}\bmu_t/\bX_i^\top
(\bLambda_t\bLambda_t^\top+\bPsi_t)^{-1}\bX_i$, $v_i = 1/\bX_i^\top
(\bLambda_t\bLambda_t^\top \allowbreak +\bPsi_t)^{-1}\bX_i$ at the current
parameter values $(\bmu_t,\bLambda_t,\bPsi_t)$, and
\begin{equation}\label{eqn:integral}
 \mI_k(m,v) := \int_{0}^\infty R^k\exp\left\{-\frac{1}{2v}(R-m)^2\right\}dR,
\end{equation}
for $k\geq 0$. We discuss fast methods for numerically computing these quantities in Section \ref{sec:comp}.

Unlike for Gaussian factor models, the estimate of $\bmu$
cannot be profiled out but is part of the M-step. Further, there is
no closed-form solution for maximizing $Q_D(\cdot)$ with respect to
(w.r.t.) $\bmu,$
$\bLambda$ and $\bPsi$ simultaneously, nor are there easy numerical schemes to maximize $Q_D(\cdot)$.
So we partition the parameter space naturally in
terms of $\{\bmu\}$, $\{\bLambda\}$, $\{\bPsi\}$ and conditionally maximize
$Q_D(\cdot)$ w.r.t. each set of parameters, in turn,
while fixing the other sets of parameters at their current
values. We alternate an E-step with each CM-step that we now discuss
individually.
It is important to note that all CM-steps in our AECM algorithm have closed-form solutions. The E-step has two sets (or a duet) of latent variables, so we use FADS-D to specify this algorithm.

The CM-step update for $\bmu$ at the $(t+1)$th iteration is
$\bmu_{t+1}\!\!=\!\!\arg\!\!\max_{\bmu} Q_D(\bmu, \bLambda_t,\bPsi_t;\bmu_t,
\bLambda_t,\bPsi_t)$ subject to the constraint $\|\bmu\|=1$. This
CM-step has the closed-form solution
\begin{equation} \label{eq:muEM}
         \mu_{t+1} = ({\bI}+\lambda\bPsi_t)^{-1}{\bb}_t,
\end{equation}
where 
${\bb}_t=\frac1n\sum^{n}_{i=1}\left\{\bbE_{R_i\mid\bX_i}(R_i)\bX_i-\bLambda_t\bbE_{\bZ_i|\bX_i}({\bZ}_i)\right\}$
and $\lambda$ is the Lagrange multiplier. The following proposition shows
that $\lambda$ can be computed easily, and so $\bmu_{t+1}$ has a unique closed-form solution for $\bmu_{t+1} = \arg\max_{\bmu} Q_D(\bmu, \bLambda_t,\bPsi_t;\bmu_t,
\bLambda_t,\bPsi_t)$ subject to the constraint $\|\bmu\|=1$.

\begin{proposition}\label{proposition:maxmu}
The global maximizer of $\bmu\mapsto Q_D(\bmu,\bLambda_t,\bPsi_t)$ is given by \eqref{eq:muEM}, where $\lambda$ is a root of $f(u) = {\bb}_t^\top ({\bI} + u\bPsi_{t})^{-2}{\bb}_t - 1.$
Moreover, if $\|{\bb}_t\| > 1,$ then $\lambda$ is the unique root of $f(u)$ between $0$ and ${\bb}_t^\top \bPsi_{t}^{-1}{\bb}_t/2$, and otherwise the unique negative root between $(|b^*|/2 - 1)/\psi^*$ and 0, where $\psi^*$ is the largest diagonal entry of $\bPsi_t$ and $b^*$ is the corresponding entry in ${\bb}_t$.
\end{proposition}
\begin{proof}
The proposition is a corollary of Theorem~\ref{theorem:maxmu} in Section
\ref{sec:supp-proposition1&2} with $\bB=\bPsi$.
\end{proof}

For the CM-step for $\bLambda$, we maximize $Q_D(\bmu_{t+1},\bLambda,\bPsi_t;\bmu_{t+1},\bLambda_t,\bPsi_t)$ w.r.t. $\bLambda$ at the current values of the other parameters,
that is at $\bmu_{t+1}$ and $\bPsi_t$ to get the update
\begin{equation} \label{eq:dem-lambda}
{\bLambda}_{t+1} =
\left[\sum^{n}_{i=1}\bbE_{\bZ_i,R_i|\bX_i} \{(R_i\bX_i  -
\bmu_{t+1}){\bZ}_i^\top\}\right]\left\{\sum^{n}_{i=1}{\bbE_{\bZ_i | \bX_i} ({\bZ}_i{\bZ}_i^\top)}\right\}^{-1},
\end{equation}
where the conditional expectations are 
at the current parameter values $\{\bmu_{t+1},\bLambda_{t},\bPsi_{t}\}$.

Finally, the CM-step for $\bPsi$ maximizes
$Q_D(\bmu_{t+1},\bLambda_{t+1},\bPsi;\bmu_{t+1},\bLambda_{t+1},\bPsi_{t})$
w.r.t. $\bPsi$ while keeping fixed the values of the other
parameters at $\bmu_{t+1}$ and $\bLambda_{t+1}$. The CM-step update is then
\begin{equation} \label{eq:dem-psi}
    \bPsi_{t+1} =
    \mathrm{diag}\left[\frac1n\sum_{i=1}^n\bbE_{R_i,\bZ_i | \bX_i}\left\{(R_i{\bX}_i- \bmu_{t+1} - \bLambda_{t+1}{\bZ}_i)(R_i{\bX}_i- \bmu_{t+1} - \bLambda_{t+1}{\bZ}_i)^\top \right\}\right],
 \end{equation}
where the conditional expectation is at the current parameter values $\{\bmu_{t+1},\bLambda_{t+1},\bPsi_t\}$.


\subsubsection{A faster AECM algorithm with profiling}
\label{sec:meth-profile-aecm}
The methodology of Section~\ref{sec:aecm-d} has two sets of latent
variables in the lengths $R_i$s and the factors $\bZ_i$s. The latter
also arise in the case of EM methods for Gaussian factor 
models, but can be eliminated by employing matrix-free profile likelihood
maximization~\citep{daietal20}. We propose to speed up our AECM algorithm by 
incorporating these methods. 
We drop $\bZ_i$s from our erstwhile duet and only consider $R_i$s as our latent variables. 
Then the $Q$-function is
\begin{equation} \label{eq:exp-llk2}
        Q_P(\bmu,\bLambda,\bPsi;\bmu_t,\bLambda_t,\bPsi_t) = c -\frac
        n2\{\log\det (\bLambda\bLambda^\top+\bPsi) +\Tr
        (\bLambda\bLambda^\top+\bPsi)^{-1}\bbE_{\bR | \bX} (\bS_{\bR}) \},
\end{equation}
where $c$ is a constant free of $\bmu,$ $\bLambda$ and $\bPsi,$ and $\bbE_{\bR | \bX}
(\bS_{\bR})$ is the conditional expectation of ${\bS}_{\bR} = n^{-1}\sum^{n}_{i=1}
(R_i{\bX}_i-\bmu)(R_i{\bX}_i-\bmu)^\top$ 
at the current parameter values $(\bmu_t,\bLambda_t,\bPsi_t)$. The calculation
of the conditional expectation $\bbE_{\bR |\bX}(\bS_{\bR})$ involves the same
conditional expectations of $R_i$ and $R_i^2$ given $\bX$ 
as in FADS-D and are computed similarly. We now discuss the
CM-steps of the AECM algorithm, where we partition
the parameter space into $\{\bmu\}$ and $\{\bLambda,\bPsi\}$.

The CM-step for $\bmu$ (at the $(t+1)$th cycle) involves maximizing
$Q_P(\bmu,\bLambda_t,\bPsi_t; \bmu_t,\bLambda_t,\bPsi_t)$ as
a function of $\bmu$ subject to the constraint $\|\bmu\|=1$. 
 Analogous to Proposition
\ref{proposition:maxmu}, a unique closed-form solution for the CM-step
for $\bmu$ can be obtained through the following 
\begin{proposition}\label{lem:maxmu2}
 The global maximizer
 of $\bmu\mapsto Q_P(\bmu,\bLambda_t,\bPsi_t)$ is $\bmu_{t+1} = ({\bI}+\lambda\Sigma_t)^{-1}{\bc}_t$,
 where ${\bc_t} = n^{-1}\sum_{i=1}^n
 \bX_i\bbE(R_i|\bX_i,\bmu_t,\bLambda_t,\bPsi_t)$,
 $\Sigma_t=\bLambda_t\bLambda_t^\top + \bPsi_t$ and $\lambda$ is a
 root of $g(u) = {\bc}_t^\top ({\bI} + u\Sigma_{t})^{-2}{\bc}_t - 1.$
 Moreover, if $\|{\bc}_t\| > 1,$ then $\lambda$ is the unique root of
 $g(u)$ between $0$ and ${\bc}_t^\top \Sigma_{t}^{-1}{\bc}_t/2$,
 otherwise it is the unique negative root between $(|{\bv}_*^\top{\bc}_t|/2 - 1)/\sigmaold^*$ and 0, where $\sigmaold^*$ is the largest eigenvalue of $\Sigma_t$ with eigenvector ${\bv}_*.$
\end{proposition}
\begin{proof}
  The proposition follows as a corollary to
  Theorem~\ref{theorem:maxmu} in Section~\ref{sec:supp-proposition1&2} with $\bB=\bSigma$.
\end{proof}
\noindent The root $\lambda$  in Proposition~\ref{lem:maxmu2} can be
computed by bisection while ${\bv}_*$ can be computed using a
restarted Lanczos algorithm, as  seen shortly in Section \ref{sec:comp}. Also,  the
inverses of $\Sigma_t$ or ${\bI}+\lambda\Sigma_t$ are computed using
the Sherman-Morrison-Woodbury identity \citep{hendersonandsearle81}.

Given $\bmu_{t+1}$ and $\bX_i$, $i=1,2,\ldots,n$, we obtain
$\bbE(R_i\mid\bX_i)$ and $\bbE(R^2_i\mid\bX_i)$ and let
\[\widetilde\bS_t = \bbE_{\bR | \bX}
(\bS_{\bR})\bigg\vert_{\bmu=\bmu_{t+1}} \equiv \frac1n\sum_{i=1}^n\bbE_{R_i|\bX_i} \left\{(R_i\bX_i-\bmu_{t+1})(R_i\bX_i-\bmu_{t+1})^\top\right\}\]
at the current parameter values.
With $\bmu$ held fixed at $\bmu_{t+1}$, maximizing~\eqref{eq:exp-llk2} w.r.t. $\bLambda$ and $\bPsi$ is equivalent to maximizing the
log-likelihood of a  Gaussian factor model where the mean is profiled
out and the ``sample covariance'' matrix is $\widetilde\bS_t$.
Therefore, we jointly update $\bLambda$ and $\bPsi$ by following~\citet{daietal20} and profiling out $\bLambda$, using 
the common (in Gaussian factor analysis) and computationally useful
identifiability constraint on $\bLambda$ that the signal
matrix $\Gamma = \bLambda^\top\bPsi^{-1}\bLambda$ is diagonal with
non-increasing diagonal entries. This scale-invariant constraint is
determined 
up to sign in the columns of $\bLambda$. Under this constraint,
$\bLambda$ is profiled out from
$Q_P(\bmu_{t+1},\bLambda,\bPsi;\bmu_{t+1},\bLambda_t,\bPsi_t)$ for a
given $\bPsi$ as per the following
\begin{proposition}\label{proposition:profileout-ch3}
  For a positive-definite diagonal matrix $\bPsi$, let $\theta_1 \geq
  \theta_2 \geq \cdots \geq \theta_q$, be the $q$ largest eigenvalues
  of ${\bW} = \bPsi^{-1/2}\widetilde\bS_t\bPsi^{-1/2}$. Let the columns of
  $\bV_q$ store the eigenvectors corresponding to these eigenvalues. Then the function $Q_P(\bmu_{t+1},\bLambda,\bPsi;\bmu_{t+1},\bLambda_t,\bPsi_t)$ is maximized w.r.t. $\bLambda$ at ${\hat\bLambda} = \bPsi^{1/2}{\bV}_q\bm\Delta,$ where $\bm\Delta$ is a $q\times q$ diagonal matrix with $i$th diagonal entry as $[\max(\theta_i-1,0)]^{1/2}.$ The profile Q-function is 
\begin{equation}\label{eqn:profilelikelihood-ch3}
Q_p(\bPsi) \doteq c' - \frac n2 \{\log\det\bPsi + \Tr \bPsi^{-1}\widetilde\bS_t
 + \sum_{i=1}^q(\log\theta_i - \theta_i + 1)\},
\end{equation}
where $c'$ is a constant 
free of $\bPsi.$ Further, 
$\nabla Q_p(\bPsi) = -\shalf n~\mathrm{diag}({\hat{\bLambda}}{\hat{\bLambda}}^\top + \bPsi - \widetilde\bS_t).$
\end{proposition}
\begin{proof}
See Section \ref{sec:supp-proposition3-ch3}.
\end{proof}
\noindent 
Our algorithm here involves profiling in the second  CM-step, so we
call it FADS-P.  
We use the limited-memory Broyden-Fletcher-Goldfarb-Shanno
quasi-Newton algorithm \citep{byrdetal95} with box-constraints
(L-BFGS-B) to maximize \eqref{eqn:profilelikelihood-ch3}. Instead of
the exact Hessian matrix, this quasi-Newton algorithm uses values of
$Q_p(\Psi)$ and $\nabla Q_p(\Psi)$ from the last few (typically five) iterations to
find an effective approximation to the exact Newton-step, thereby
reducing storage costs from $O(p^2)$ to $O(p).$  Having obtained
$\bPsi_{t+1} = \arg\max Q_p(\bPsi)$, we use Proposition
\ref{proposition:profileout-ch3} to compute the eigenvalues and eigenvectors of 
$\bPsi_{t+1}^{-1/2}\widetilde\bS_t\bPsi_{t+1}^{-1/2}$ and 
update $\bLambda_{t+1} = \bPsi_{t+1}^{1/2}{\bV}_q\bm\Delta.$

We conclude with a few remarks on the CM-update of $\Psi$.
The function $Q_p(\bPsi)$ in
\eqref{eqn:profilelikelihood-ch3} depends on $\widetilde\bS_t$ only through the
$q$ largest eigenvalues of ${\bW}$, so 
for $Q_p(\bPsi)$ and $\nabla Q_p(\bPsi)$, we need to compute only the
$q$ largest eigenvalues and corresponding eigenvectors of ${\bW}$. 
For $q\ll\min(n,p)$, the largest eigenvalues and eigenvectors are
speedily computed using the restarted Lanczos algorithm, with constraints on ${\bPsi}$
({\em e.g.}, ${\bPsi}=\sigmaold^2{\bI}_p,$ $\sigmaold^2 >0$) easily
incorporated. Also, $\nabla Q_p(\bPsi)$ is available in closed-form, enabling
a check for first-order optimality. 
%
Finally, $Q_p(\bPsi)$ is expressed in terms of $\widetilde\bS_t$, but ML estimators are scale-equivariant, so 
each iteration of FADS-P estimates $\bLambda$ and $\bPsi$ using the
correlation matrix 
 and  scales the estimates back to that of $\widetilde\bS_t$.
\subsubsection{Implementation via fast statistical computations}
\label{sec:comp}
We want our FADs algorithms to also be used for datasets with large
$n$ or $p$, as in our applications, so we now discuss and investigate
practical ways to speed up computations.
\paragraph{Computation of log-likelihood and conditional expectations}
Both the log-likelihood \eqref{eqn:pdfPN} and the conditional expectations require computing  \eqref{eqn:integral} for $k\in\{p-1,p,p+1\}$. 
To reduce clutter, we write $\mI_k\equiv\mI_k(m,v).$ Integration by parts yields the recurrence relation \citep{pukkilaandrao88}
\begin{equation}\label{eqn:recI}
 \mI_{k+2} = m \mI_{k+1} + (k+1)v \mI_k.
\end{equation}
However, this recurrence relation is vulnerable to numerical overflow
for large $p$, so  we derive (Section~\ref{sec:supp-int}) a
recurrence relation for the ratio
\begin{eqnarray}\label{eqn:recRatio}
\dfrac{\mI_{k+2}}{\mI_{k+1}} =  m + (k+1)v \dfrac{\mI_{k}}{\mI_{k+1}}.
\end{eqnarray}
Together, \eqref{eqn:recI}  and \eqref{eqn:recRatio} also yield
$\mI_{k+2}/\mI_{k} = (k+1)v + 
m\mI_{k+1}/\mI_{k}.$ When $m>0,$ the recurrence relation
\eqref{eqn:recRatio} is numerically stable and computed very quickly,
even for large $k.$ However, for $m<0,$ both 
\eqref{eqn:recI} and \eqref{eqn:recRatio} are unstable and often yield
negative values for $\mI_k$ with $k >10.$ So, 
for $m < 0$, we use quadrature to compute $\mI_k$, for
$k\in\{p-1,p,p+1\}$, using the fact that the integrand in
\eqref{eqn:integral} is strictly
log-concave and maximized at $\varrho = (m + \sqrt{m^2 + 4vk})/2$. We write
\begin{equation}
\label{gk.int}
  \log \mI_k = k\log \varrho - \frac{1}{2v}(\varrho-m)^2 + \log
  \int_{0}^\infty \left(\frac{R}{\varrho}\right)^k
  \exp\left\{-\frac{1}{2v}(R-\varrho)(R+\varrho-2m)\right\}dR.
\end{equation}
The integrand is bounded above by unity so we can easily
find ({\em e.g.}, using bisection) two quantities $0 <
\underline{\varrho} < \overline{\varrho}$ where the integrand is less than some small value ({\em e.g.}, $10^{-15}$). The integration in~\eqref{gk.int} is by non-adaptive
Gauss-Kronrod quadrature over 
$[\underline{\varrho},\overline{\varrho}]$ with 10, 21, 43 and 87
points \citep{piessensetal83}: in our experiments, 43 points gave high numerical accuracy.
\paragraph{Computation of partial eigenvalues and eigenvectors}
We can obtain the $q$
largest eigenvalues and eigenvectors of $\mathbf{W}$ 
via the implicitly restarted Lanczos algorithm~\citep{sorensen92}. Suppose that $m = \max\{2q + 1, 20\}$ and that $\mathbf{f}_1 \in \bbR^p$ is any vector with $\|\mathbf{f}_1\| = 1$ and initialize $\mathbf{F}_1 = \mathbf{f}_1.$
We then employ the Lanczos iterations \citep{duttaandmondal15} as follows.
For $k=1,2,\ldots,m$,
\begin{itemize}
\item Compute ${\bu}_k = {\bW\bff}_k$ and $\alpha_k = {\bff}_k^\top{\bu}_k$.
\item Compute ${\br}_k = {\bu}_k - \alpha_k{\bff}_k - \beta_{k-1}{\bff}_{k-1}$ (assuming $\beta_0=0$ and ${\bff}_0 = {\bzero}$).
\item Let $\beta_k = \|{\br}_k\|$ and if $k<q$ and $\beta_k\ne 0,$ compute ${\bff}_{k+1} = \frac{{\br}_k}{\beta_k}$ and set ${\bF}_{k+1}=[{\bF}_k,{\bff}_{k+1}]$.
\end{itemize}
Suppose that ${\bT}_m$ is the $m\times m$ symmetric tridiagonal
matrix with diagonal entries $\alpha_1,\alpha_2,\ldots,\alpha_m$ and $j$th off-diagonal 
entries $\beta_j$ for $j=1,\ldots,m-1$.
We compute the eigenvalues $e_1 > e_2 > \cdots > e_m$ of ${\bT}_m$ with eigenvectors ${\bg}_1,{\bg}_2,\ldots,{\bg}_m$ via a Sturm sequencing algorithm \citep{wilkinson1958}. Also let ${\bv}_j = {\bF}_m{\bg}_j,$ for $1\leq j \leq m.$ The $e_j$'s and ${\bv}_j$'s are called Ritz values and Ritz vectors of ${\bW}$. It can be shown that $\|{\bW\bv}_j - {\bv}_je_j\| = \beta_{m}|g_{j,m}|,$ for $g_{j,m}$ the $m$th entry of vector ${\bg}_j,$ for $j=1,2,\ldots,m.$ The algorithm stops if
\begin{equation}\label{eqn:lanczosStop}
 \beta_{m}\max_{1\leq j \leq m}|g_{j,m}| < \delta
\end{equation}
for some prespecified tolerance $\delta$ and $e_1,\ldots,e_q$ and ${\bv}_1,\bv_2,\ldots,{\bv}_q$ are accurate approximations of the $q$ largest eigenvalues and corresponding eigenvectors of ${\bW}$.

However, in practice, more iterations are needed for the Ritz vectors
and Ritz values to converge to the eigenvalues and eigenvectors of ${\bW}$, so~\citet{sorensen92} suggests implicitly restarting the
Lanczos algorithm and also shifting the spectrum of the symmetric tridiagonal matrices iteratively to force
the new residuals ${\br}_m$ to zero, thereby accelerating the
convergence rate. So we 
compute the QR-decompositions: ${\bT}_{m} - e_{j}{\bI}_m = {\bQ}_j{\tilde{\bR}}_j,$ for $j=q+1,\ldots,m,$ let $\tilde{{\bQ}} = {\bQ}_{q+1}{\bQ}_{q+2}\cdots{\bQ}_{m}$ and reset
${\bF}_{m} = {\bF}_m\tilde{{\bQ}}$ and ${\bT}_{m} = \tilde{{\bQ}}^\top
{\bT}_m \tilde{{\bQ}}.$ Then 
\begin{equation}\label{eqn:restart}
 {\bW\bF}_{q} = {\bF}_q{\bT}_q + \beta^*{\bff}_{q+1}{\be}_q^\top
\end{equation}
where $\beta^*$ is the $(q+1,q)$th entry of ${\bT}_m,$ ${\be}_q$ is
the $q$th canonical basis vector in $\bbR^q,$ and ${\bT}_q$ is the
$q\times q$ principal sub-matrix of ${\bT}_m$~\citep{sorensen92}. Therefore,
\eqref{eqn:restart} is itself a $q$th-order Lanczos factorization of
${\bW}.$ Next, 
we ``restart'' the Lanczos iterations from $k=q+1,\ldots,m$ instead of
$1$ through $m$, terminating if 
\eqref{eqn:lanczosStop} is satisfied, and restarting the algorithm
otherwise. 

The only way ${\bW}$ enters our algorithm is through matrix-vector products ${\bW\bff}_k$ that  can be computed without storing ${\bW}$ or 
$\widetilde\bS_t$. Additionally,  $\bW$ is symmetric, so  we do not
need separate computations for  the left 
 and right  singular vectors in partial singular value decomposition using the
Lanczos algorithm (as needed in \citet{daietal20}), yielding substantial
savings in terms of both compute time and storage.
Overall, our algorithm calculates the $q$ largest eigenvalues and
eigenvectors with $O(qnp)$ computational cost and $O(qp)$
additional memory.
\subsubsection{Initialization}
\label{sec:init}
As with most iterative algorithms, initial values can significantly
impact performance of our EM algorithms. 
We devise methods for initializing $(\bmu,\bLambda,\bPsi)$, borrowing
ideas from~\citet{maitra09, maitra13}. 
Our  $\bmu$ was simulated from $\mN_p(\bzero,\bI)$ and
$l_2$-normalized, while the diagonal entries of $\bPsi$ were
each i.i.d. $\mU(0.2,0.8)$ draws and the elements of 
$\bLambda$ were  i.i.d. $\mN(0,1)$ realizations.
Starting with $M$ 
such initial values, the AECM
algorithm was run for $J$ ``short'' iterations  after which only the $L\leq M$ streams with the highest
log-likelihood values were iterated all the way to convergence. The final
estimates are those obtained from the run achieving the highest final log-likelihood values.
In this paper, we used $(M,J,L) = (1000,10,10)$.

\subsection{Choosing $q$}
The factor model explains variability in a large number of variables
through a small number ($q$) of latent factors, so the choice of $q$ is
important. There is little theoretical work on the selection of $q$
for Gaussian factor models. The  Bayesian information criterion (BIC)
of~\citet{schwarz78} is one possibility, however it does not always
perform well in  high dimensions~\citep{chenandchen08}. We use
eBIC that minimizes
$\mathrm{eBIC}(q) = -2\log\hat{\ell}_q +  pq \{\log(n) + 2\gamma 
\log(p)\},$
 with $\gamma = \max\{1-1/(2\log_n(p)) , 0\},$ 
which performs well in our experiments.
\subsection{Estimating the factor scores}
\label{sec:est.scores}
Once $q$ has been chosen and the ML estimators of the parameters obtained, we
may compute factor scores for use in subsequent
analyses. For example, these scores can be used to rank each
observational unit on the factors, for post-hoc clustering
\citep{argelaguet2018multi,meng2016mocluster,mo2018fully,shen2009integrative},
in scRNA-seq analysis \citep{grun2015single,saelens2019comparison}, or
they can be used as covariates in generalized linear regression
models \citep{distefano2009understanding}, for example, to account for
population structure in genome-wide association studies
\citep{patterson2006population,price2006principal,novembre2008interpreting}. 
We extend the approach of \citet[p 226--231]{thurstone1935vectors} to
estimate the factor score for the $i$th observation by minimizing the
expected weighted squared error loss $\bbE\{ \|\bPsi^{-1/2}(\bY_i -
\mu - \bLambda{\bZ}_i)\|^2 | \bX_i\}.$ The resulting solution is the conditional expectation
of Bartlett score with complete data: 
$  \widehat{\bZ_i} =
  (\bLambda^\top\bPsi^{-1}\bLambda)^{-1}\bLambda^\top
  \bPsi^{-1}\{\bbE_{R_i|\bX_i}(R_i)\bX_i - \bmu\}.$
We use the unobserved $\bY_i$s instead of $\bX_i$s in the
minimization problem because $\Psi$ and $\Lambda$ are estimated on the
latent scale. As a result, the factor scores are also estimated on the
latent scale. In practice, the ML estimates of the parameters are
plugged in, and the value of the conditional expectation $\bbE_{R_i|\bX_i}(R_i)$ at the ML estimates is available as a byproduct of our FADS algorithms.

\subsection{Standard errors of the ML estimates}
The standard errors of the ML estimates can be obtained by using the missing information principle \citep{louis1982} when $n \ge (q+2)p.$ To that end, suppose that $\bTheta = \{\bmu,\bLambda,\bPsi\}$. The observed information at the ML estimates is given by,
    $\bI_{\bX}(\bTheta) = (1/n)\sum_{i=1}^{n}\nabla\bq_i\nabla\bq_{i}^\top$,
where for $i=1,2,\ldots,n$ $\nabla\bq_i = \bbE_{R_i\mid\bX_i}\big({\partial \ell(\bY_i; \bTheta)}/{\partial \bTheta}\big)$ is the complete data score statistic and
\begin{equation*}
\begin{split}
   \frac{\partial \ell(\bY_i; \bTheta)}{\partial \bmu} 
    &= (\bLambda\bLambda^\top+\bPsi)^{-1}(R_i\bX_i-\bmu)\\
    \frac{\partial \ell(\bY_i; \bTheta)}{\partial \bLambda} 
    &= -\frac{1}{2}\left\{
    (\bLambda\bLambda^\top+\bPsi)^{-1}\bLambda - (\bLambda\bLambda^\top+\bPsi)^{-1}(R_i\bX_i-\bmu)(R_i\bX_i-\bmu)^\top(\bLambda\bLambda^\top+\bPsi)^{-1}\bLambda
    \right\}\\
     \frac{\partial \ell(\bY_i; \bTheta)}{\partial \bPsi}
    &= -\frac{1}{2}\mathrm{diag}\left\{
    (\bLambda\bLambda^\top+\bPsi)^{-1}-(\bLambda\bLambda^\top+\bPsi)^{-1}(R_i\bX_i-\bmu)(R_i\bX_i-\bmu)^\top(\bLambda\bLambda^\top+\bPsi)^{-1}
    \right\}.
\end{split}
\end{equation*}
In practice, the ML estimates of the parameters $\bmu,\bLambda$ and
$\bPsi$ are plugged into the above. Further, 
$\bbE(R_i\mid\bX_i)$ and $\bbE(R^2_i\mid\bX_i)$ at the ML
estimates are available as byproducts of the AECM algorithm so these
calculations are very easily obtained in the course of the FADS
calculations. 

\section{Performance evaluations} 
\label{sec:sim-ch3}

We evaluated our FADS algorithms through  simulation 
experiments that  assessed estimation accuracy and consistency as well
as computing speed for a range of $(n,p,q)$
 settings. 

\subsection{Experimental setup}
\label{sec:sim-setup-ch3}
We simulated 100 datasets for each $(n,p)\in\{(300,30),(1000,100)\}$ and $q\in\{3,5\}.$
Our strategy for each
dataset was to set the diagonal elements of  $\bPsi$ as
i.i.d. $\mathcal{U}(0.2,0.8)$ and entries in
$\Lambda$ as i.i.d. $\mathcal{N}(0,1)$ pseudo-random deviates. Entries of $\mu$ were
simulated independently from $\mathcal{N}(0,1)$ and 
$l_2$-normalized. A separate evaluation -- to mimic the $p\gg n$ setting of the
application in Section~\ref{sec:tgca.intro} -- used 100 simulated
datasets with $(p,q)=(5123,12)$ and $n\in\{380,500\}$,
with the true $\Psi, \Lambda,\bm\mu$ set to the ML estimates obtained upon fitting FADS-P to the TCGA dataset. In all cases, we
stopped FADS-P and FADS-D when the improvement in the observed
log-likelihood did not exceed $10^{-4}$ and $\|\nabla
Q_P\|_\infty < \sqrt{\epsilon_0}$, where $\epsilon_0\approx$
$2.2\times10^{-16}$ is the machine tolerance, or if the number of
iterations exceeded $10^{4}$.  For each simulated dataset, we fitted
models with $k=1,2,\cdots,2q$ factors and chose the number of factors
using eBIC.  All experiments 
were run on a workstation with Intel E5-2640 v3 CPU clocked
@2.60 GHz and 64GB RAM.

\begin{figure}[t]
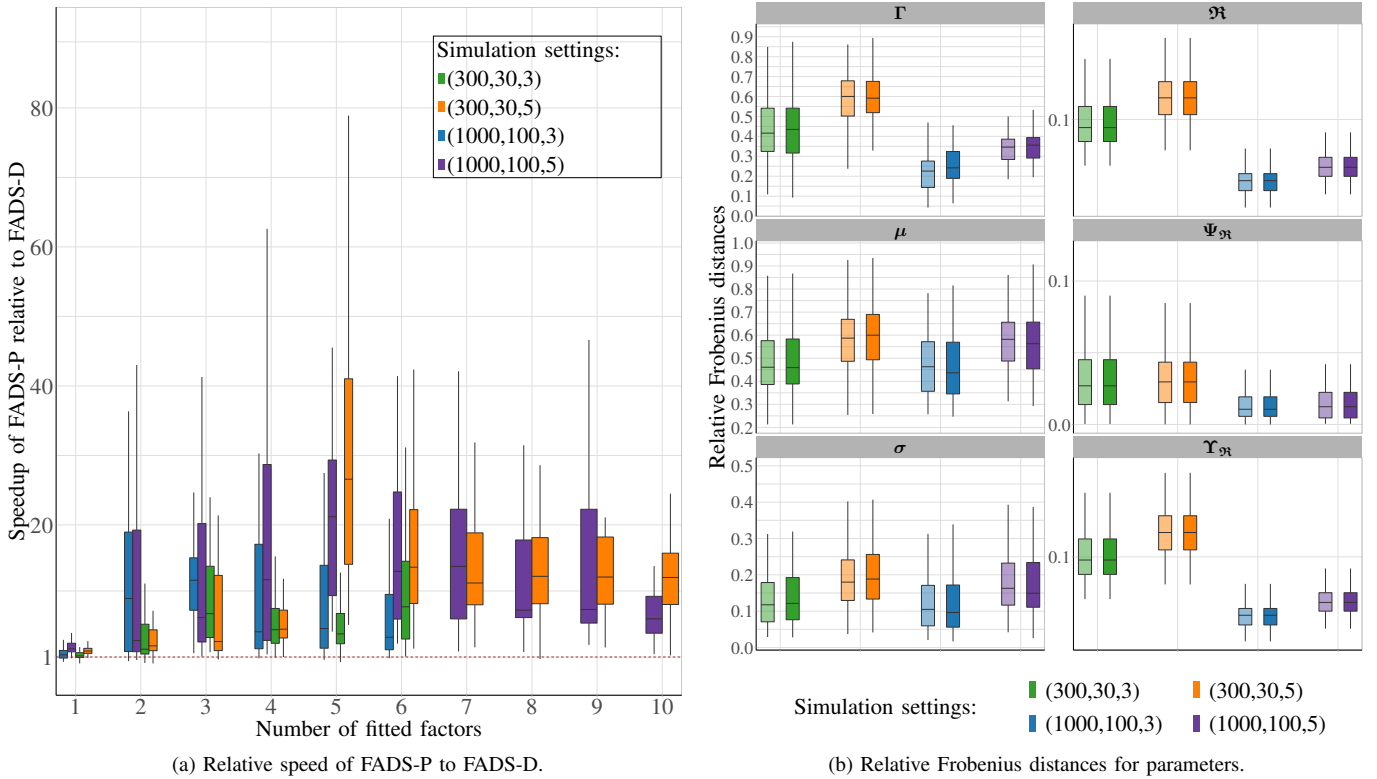

  \centering
\vspace{-0.25in}
  \mbox{
    \subfloat[Relative speed of FADS-P to FADS-D.]{\label{fig-time}
      \hspace{-0.17in}
      \resizebox{.5\textwidth}{!}{\input{figures/main/plot_time_c.tex}}
    }
    \subfloat[Relative Frobenius distances for parameters.]{\label{fig-est}
      \resizebox{.485\textwidth}{!}{\input{figures/main/plot_est.tex}}
    }
}
\caption{(a) Computational speed of FADS-P relative to FADS-D. (b)
  Estimation accuracy of  FADS-D (lighter shades) and FADS-P (darker
  shades). The experimental settings $(n,p,q)$ are in parentheses.
  }
\vspace{-0.1in}
\label{fig-sim}
\end{figure}

\subsection{Results} 
\label{sec:sim-res-ch3}
In our experiments, eBIC always correctly picked $q$. Since factor loadings and uniquenesses are better interpreted on the
correlation scale~\citep{mardiaetal79}, we considered the parameters
$\Lambda_\corr\! =\! \diag(\bsigma)^{-1/2}\Lambda,$
$\Psi_\corr\!=\!\diag(\bsigma)^{-1}\Psi,$ $\Upsilon_\corr\! =\!
\Lambda_\corr\Lambda_\corr^\top$ and $\corr\! =\! \Upsilon_\corr +
\Psi_\corr,$  where $\bsigma$ is the vector of marginal standard
deviations in $\Sigma.$ We compared the FADS-D and
FADS-P estimates using the relative Frobenius distances: $d_{\widehat{\corr}} =
\|\widehat{\corr}-\corr\|_F/\|\corr\|_F$, $d_{\widehat{\Upsilon}_{\corr}} = \|\widehat{\Upsilon}_{\corr}-\Upsilon_{\corr}\|_F/\|\Upsilon_{\corr}\|_F,$ $d_{\widehat{\Psi}_{\corr}} = \|\widehat{\Psi}_{\corr}-\Psi_{\corr}\|_F/\|\Psi_{\corr}\|_F,$ $d_{\widehat{\mu}} =
\|\widehat{\mu}-\mu\|_F/\|\mu\|_F = \|\widehat{\mu}-\mu\|_F,$  $d_{\widehat{\bsigma}} =
\|\widehat{\bsigma}-\bsigma\|_F/\|\bsigma\|_F$, and 
$d_{\widehat{\bm
    \Gamma}} = \|\widehat{\bm\Gamma}-{\bm\Gamma}\|_F/\|{\bm\Gamma}\|_F$,
 where $\|{\bf A}\|_F = \Tr {\bf A}^\top{\bf A}$ denotes the Frobenius norm of any matrix $\bf A$, $\widehat{\mathbf{\Gamma} }
=
\widehat{\bm{\Lambda}}^{\top}\widehat{\bm{\Psi}}^{-1}\widehat{\bm{\Lambda}}$,
$\widehat{\corr} =
\widehat{\Upsilon}_{\corr}+\widehat{\bm{\Psi}}_{\corr}$ with
$\widehat{\Upsilon}_{\corr} =
\widehat{\Lambda}_{\corr}\widehat{\Lambda}_{\corr}^{\top}$ are the ML estimates.


Fig.~\ref{fig-time} presents the relative speed of FADS-P to FADS-D for
the randomly simulated cases -- see  Section
\ref{sec:supp-sim-avgtime-ch3} for average CPU time. (Our reported compute
times include the common  initialization time for both
algorithms.) For  $(n,p)\in\{(300,30),(1000,100)\}$, 
FADS-P was faster than FADS-D, with maximum speedup occurring at true
$q$. FADS-P also generally needed far fewer iterations than FADS-D to converge. In particular, FADS-D did not converge even within
$10^4$ iterations for four  $(n,p,q) = (1000, 100, 5)$ cases. In
contrast, FADS-P always 
converged even though it required more iterations in the over-fitted
models than when the fitted $q$ did not exceed the true $q$ (Section~\ref{sec:supp-sim-avgiter}). 
Therefore, the speedup of FADS-P relative to FADS-D
was slightly underestimated because of the censored reports from FADS-D.

For the data-driven cases where $n\ll p$, FADS-D never converged within
$10^4$ iterations, but FADS-P always converged successfully. 
So we only report FADS-P estimates for $n\ll p$, and calculate
and compare estimation accuracy of FADS-D and FADS-P for the 
$n>p$ cases where both of them converged. (The four $n>p$ cases
where FADS-D did not converge had terminating results that were not
significantly different in estimation accuracy from those from
FADS-P.) 
FADS-P and FADS-D yielded identical values of $\widehat{\corr}$,
$\widehat{\Gamma}$, $\widehat{\Upsilon}_{\corr}$ and
$\widehat{\Psi}_{\corr}$ under the best-fitted models, so the relative estimation errors 
(Figure~\ref{fig-est}) 
were also identical. The relative errors for $\mu$ and $\sigma$ were also similar
for FADS-P and FADS-D. With $n\ll p$, the relative errors of FADS-P estimates 
decreased with increasing sample size (Table~\ref{tab-sim-tcga} and
Figs.~\ref{fig-est} and ~\ref{fig-est-tcga}). 
\begin{table}
\caption{\label{tab-sim-tcga} Average relative Frobenius distances of FADS-P for the $p \gg n$ experiments.}
{\centering%
\fbox{%
\begin{tabular}{c|c|c|c|c|c|c}
$(n,p,q)$ & $\corr$ & $\bm\Gamma$ & ${\Upsilon}_{\corr}$ & ${\Psi}_{\corr}$ & $\mu$ & $\bm\sigma$\\ 
    \hline
(380,5123,12) & 0.020 & 0.330 & 0.020 & 0.186 & 1.191 & 0.693\\
    \hline
(500,5123,12) & 0.020 & 0.326 & 0.020 & 0.184 & 1.188 & 0.689
\end{tabular}}}
\end{table}

The results of our experiments show that FADS-P and FADS-D (when it converges) can, along with the use of eBIC, accurately recover factor
structures underlying our PN model. Overall, FADS-P provides faster
and more reliable estimates, so we use it in our applications.

\section{Real Data Applications}
\label{sec:app-ch3}
We applied FADS-P to the four datasets of
Section~\ref{sec:motivating}. For each dataset, we first discuss how
its processing places the observations on the unit sphere
and then fit, analyze and interpret our factor models. We also
provide diagnostic checks to verify the preference for our PN factor
model over the existing (Langevin/Fisher-von Mises) model. 
The Langevin distribution has hitherto been the only one 
routinely used to directly model observations on unit spheres of as
high dimensions as our applications,  
so our diagnostic is in the form of a hypothesis test, where the null 
hypothesis is that $\bX_1,\bX_2,\ldots,\bX_n$ are i.i.d. realizations from the
$p$-variate Langevin distribution  $\mL_p(\kappa,\bnu)$ with 
mean direction vector $\bnu$ and concentration parameter $\kappa$. The
alternative hypothesis specifies each
$\bX_i\sim\mP\mN_p(\bmu,\bLambda\bLambda^\top+\bPsi)$, where
$\bmu^\top\bmu=1$, 
$\bLambda$ is a $p\times q$ matrix and $\bPsi$ is a diagonal
matrix. Our test statistic is the difference in maximized
log-likelihood ($V = \ell_a- \ell_0$) of the data using the
$\mL_p(\kappa,\bnu)$ and $\mP\mN_p(\bmu,\bLambda\bLambda^\top+\bPsi)$
models the optimal $q$ selected from the data using eBIC. We
use~\citet{barnard63}'s exact Monte Carlo method to calculate the
$p$-value of our test statistic in each application. 
Under this framework, the  alternative model is favored by a large
test statistic $V_0$ computed from the observed data compared with
the values $V_1,V_2,\ldots,V_{M}$, 
test statistics obtained from the $M$ Monte Carlo samples generated under
the (null) $\mL_p(\kappa,\bnu)$ distribution with parameters estimated from
fitting the $\mL_p(\kappa,\bnu)$ model to the original data.  We fit
both the  $\mL_p(\kappa,\bnu)$ and
$\mP\mN_p(\bmu,\bLambda\bLambda^\top+\bPsi)$ models  
to the $i$th Monte Carlo sample, obtaining the maximized likelihood values,
$\ell_0^{(i)}$ and $\ell_a^{(i)}$, and $V_i =\ell_a^{(i)}-\ell_0^{(i)}$. 
The exact Monte Carlo $p$-value was the rank of $V_0$ in
$\{V_0,V_1,V_2,\ldots,V_{M}\}$ divided by $M+1$. The $p$-value
calculation may be accelerated~\citep{besagandclifford91}, but
this speedup was not needed here.

\subsection{Characterizing major themes underlying the {\tt \#MeToo} tweets}
\label{sec:app-ch3-metoo}
\subsubsection{Data collection and preprocessing}
\label{sec:tweets.preprocess}
We collected all 43,247 original tweets (excluding so-called re-tweets) in
English that contained the {\tt \#MeToo} hashtag 
between December 03, 2018 at 02:06:15 (UTC) and  December 13, 2018 at
01:36:20 (UTC). These tweets originally had 1,165,762 unique
words. Standard text preprocessing put the dataset through (1)
unitization and tokenization, (2) standardization and cleansing, (3)
stop word removal, and (4) stemming and lemmatization
\citep{anandarajanetal18}. In addition, infrequently-occurring words
(appearing in less than 0.25\% of tweets) were removed from the
lexicon, following~\citet{dhillonandmodha01}. We also
removed one-word tweets because they contribute no
information on relationships between terms. The resulting collection
had $n=31,385$ unique tweets and $p=721$ unique words (see
Section~\ref{sec:supp-wordlist}).  
(Henceforth, we use tweets to refer to this reduced dataset.)
We applied the commonly used term frequency-inverse document
frequency (\textit{tf-idf}) weighting~\citep{anandarajanetal18}  to modulate the relative frequency of each word by tweet length.  
This procedure weighted $d_{ij}$, or the raw frequency of
the $j$th word, in the $i$th tweet by  $\log(n/d_{\cdot j})$ where $d_{\cdot
  j}
=\sum_{i=1}^nd_{ij}$ is the number of tweets with the $j$th word. 
We therefore have a \textit{tf-idf} weighted document-term matrix
(DTM). Each tweet vector was then  $l_2$-normalized, placing it on the unit
sphere. The $l_2$-normalization mitigates the effect of
differing tweet lengths~\citep{singhaletal96} and also allows for the use
of cosine similarity that performs better than the Euclidean distance metric 
on sparse text data~\citep{lebanon06}. 
The $l_2$-normalized data $\bX_1,\bX_2,\ldots,\bX_n$ is our dataset
for analysis. 
\subsubsection{Results and analysis}
\label{sec:tweets.results}
We applied FADS-P, with $q=0,1,2,\ldots,15$, to the processed dataset of
Section~\ref{sec:tweets.preprocess}. The eBIC chose the eight-factor model, which was significantly preferred over the Langevin model
($p\mbox{-value}<0.001$). Fig.~\ref{fig-mt-l}a displays the
estimated uniquenesses for the fitted model,
\begin{figure}
\vspace{-.4in}
  \centering
\begin{minipage}{.95\textwidth}
\vspace{-.6in}
\mbox{
 \subfloat[Uniquenesses]{\label{fig-mt-d}\includegraphics[width =
   0.33\textwidth]{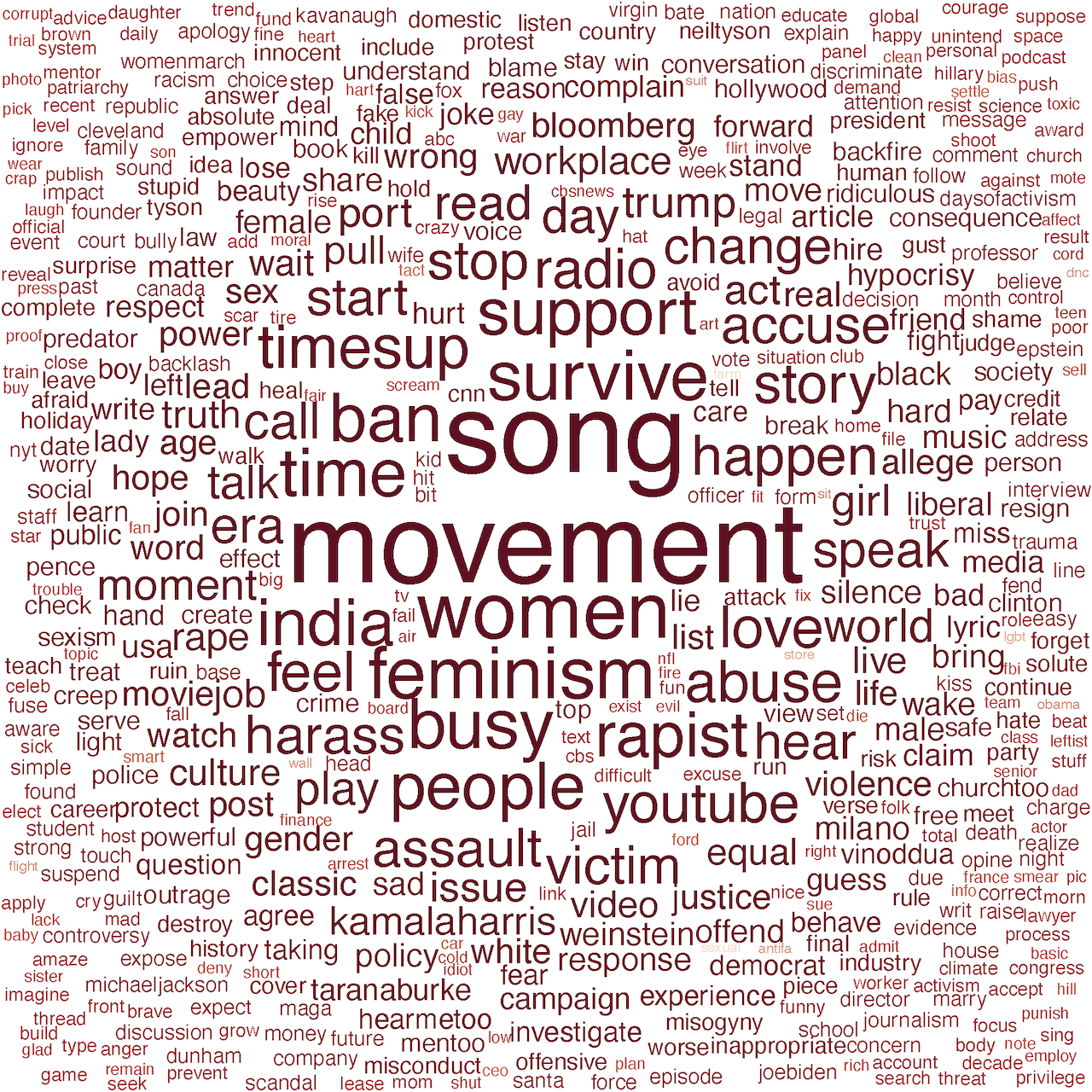}}
 \subfloat[Dec. 3 Bloomberg article.]{\label{fig-mt-l1}\includegraphics[width = 0.33\textwidth]{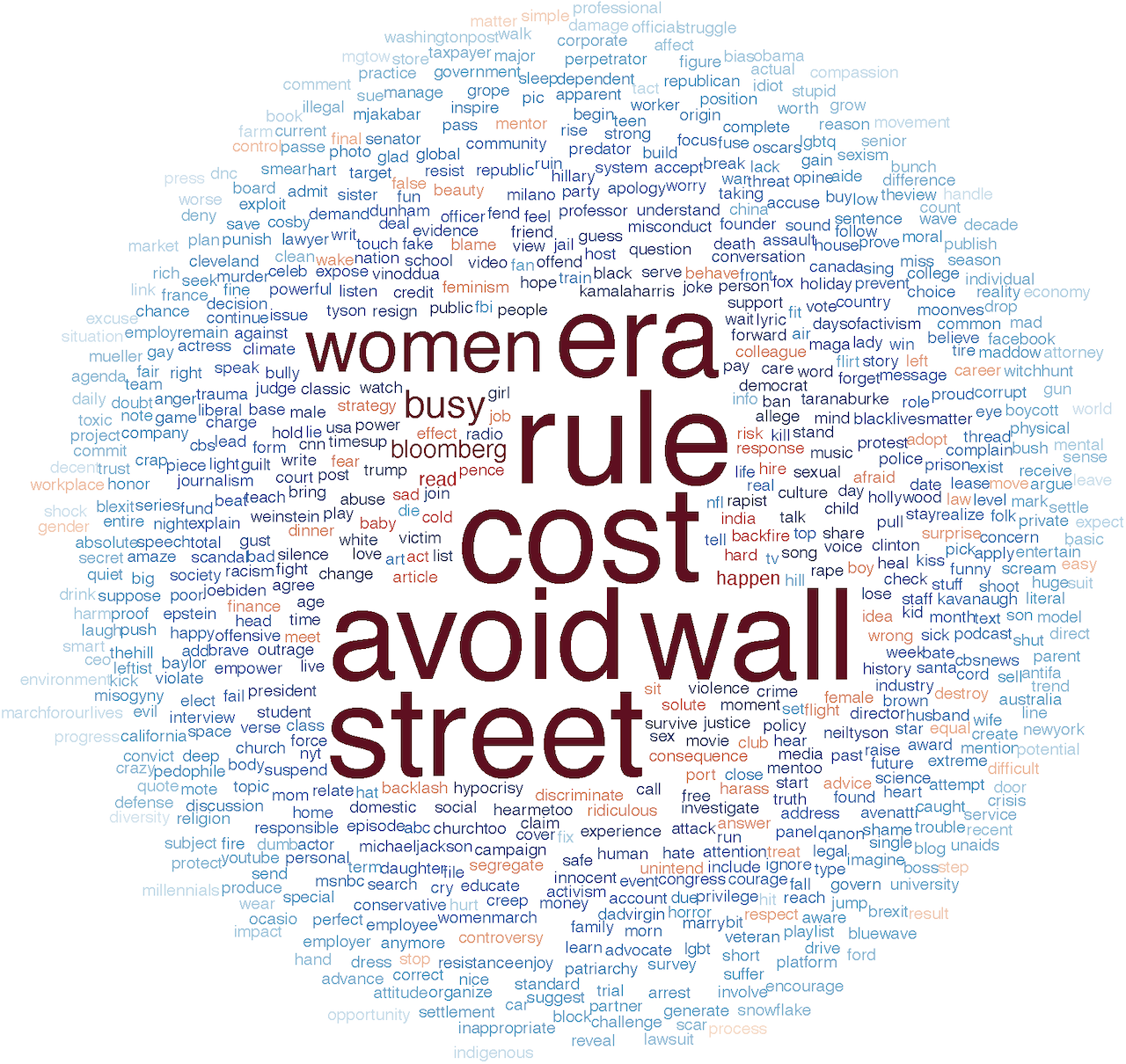}}
  \subfloat[Banned Christmas song.
  ]{\label{fig-mt-l2}\includegraphics[width = 0.33\textwidth]{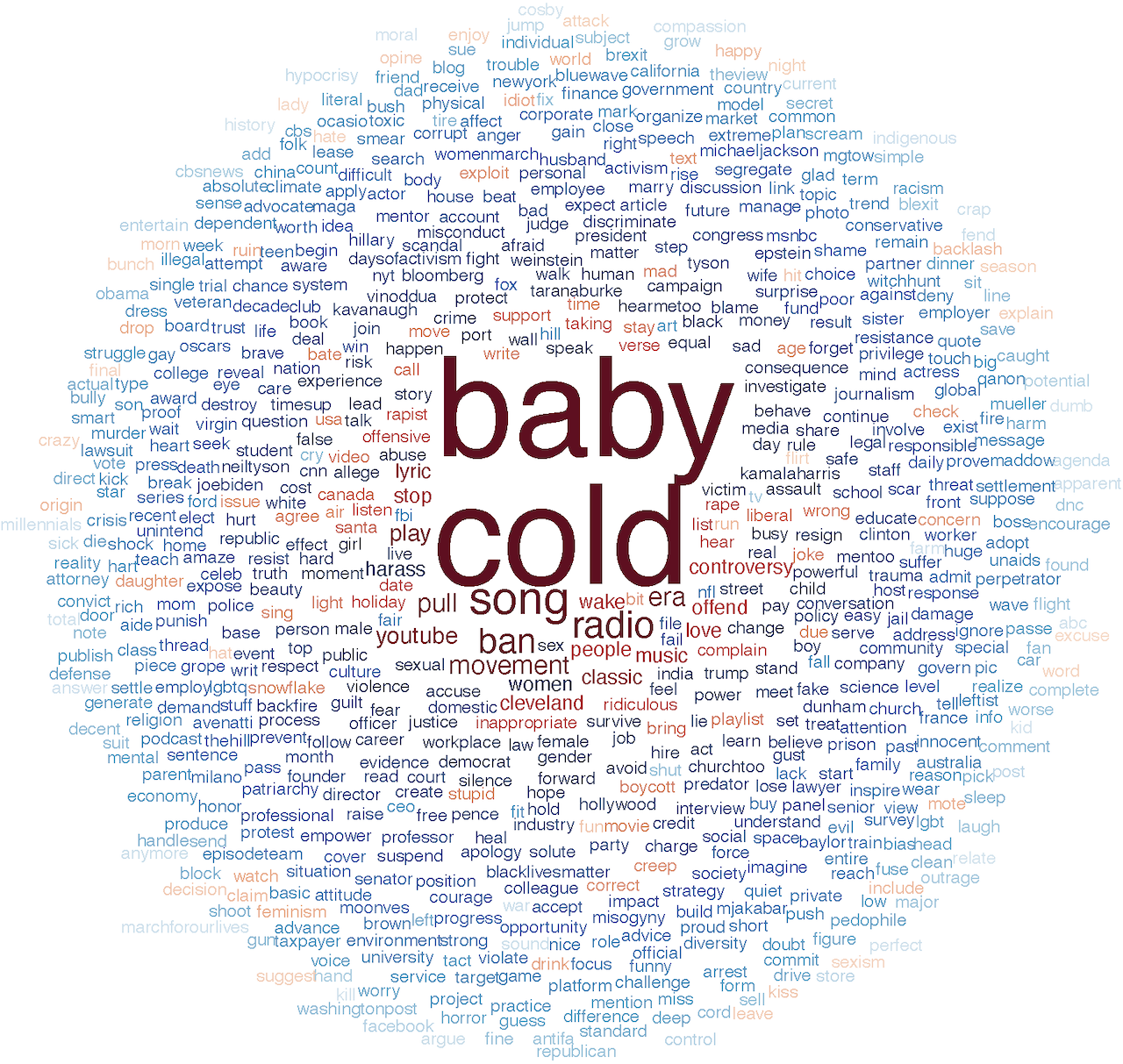}}}
\mbox{
  \subfloat[Sexual harassment/assault.]{\label{fig-mt-l3}\includegraphics[width = 0.33\textwidth]{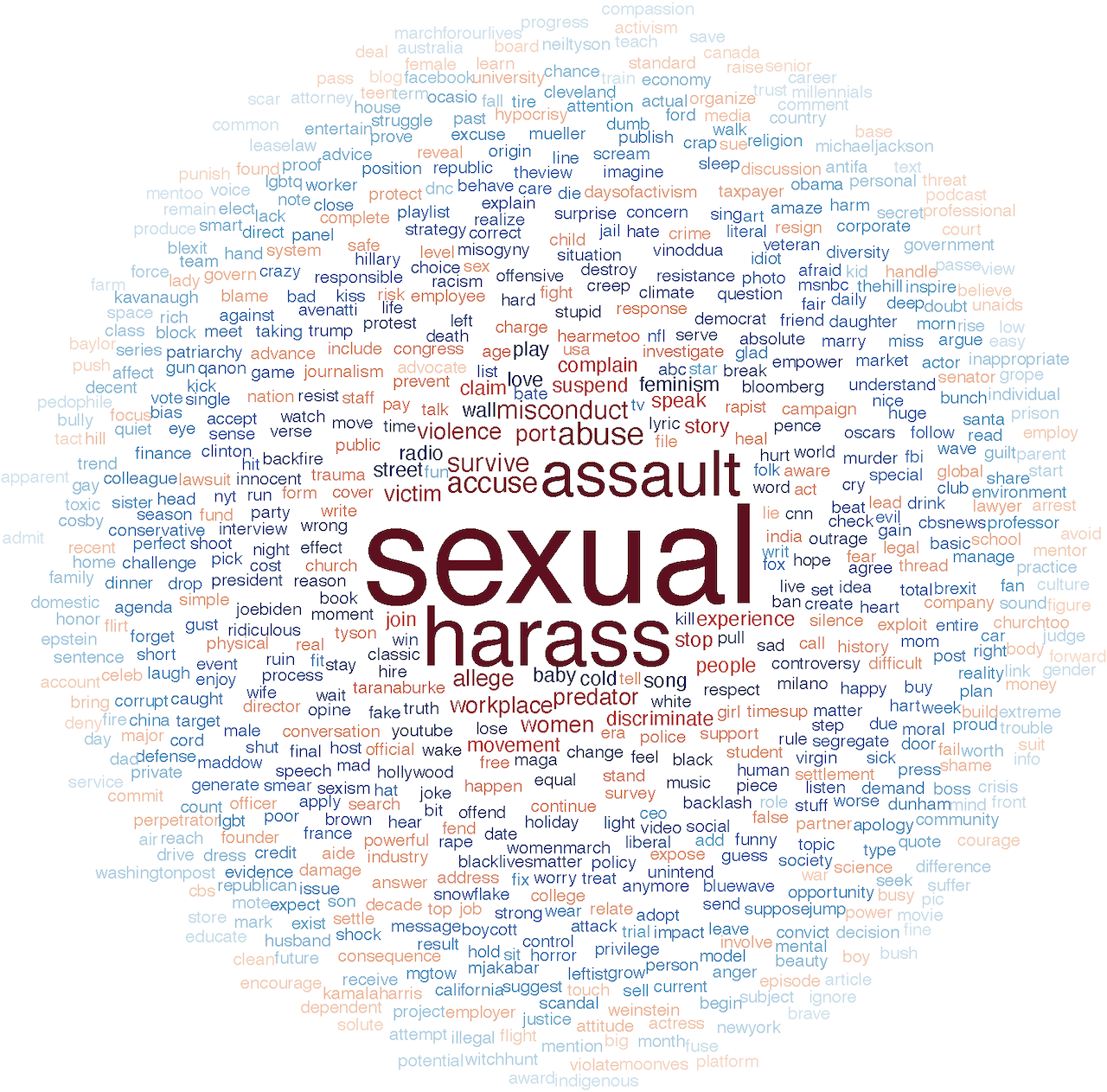}}
  \subfloat[Dec. 5 Bloomberg article.]{\label{fig-mt-l4}\includegraphics[width = 0.33\textwidth]{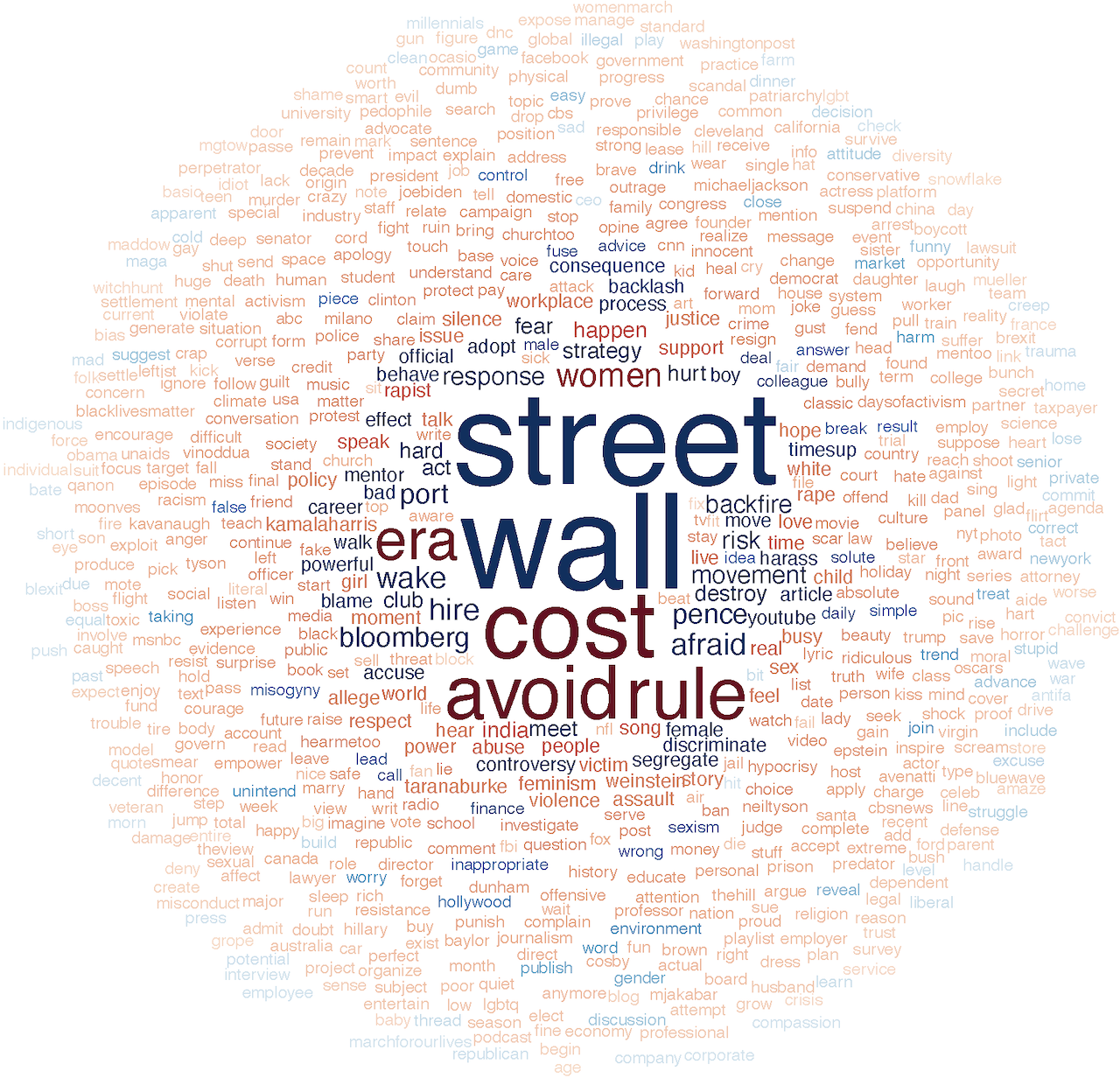}}
  \subfloat[Politics, Policy and {\tt \#MeToo}.]{\label{fig-mt-l5}\includegraphics[width =
    0.33\textwidth]{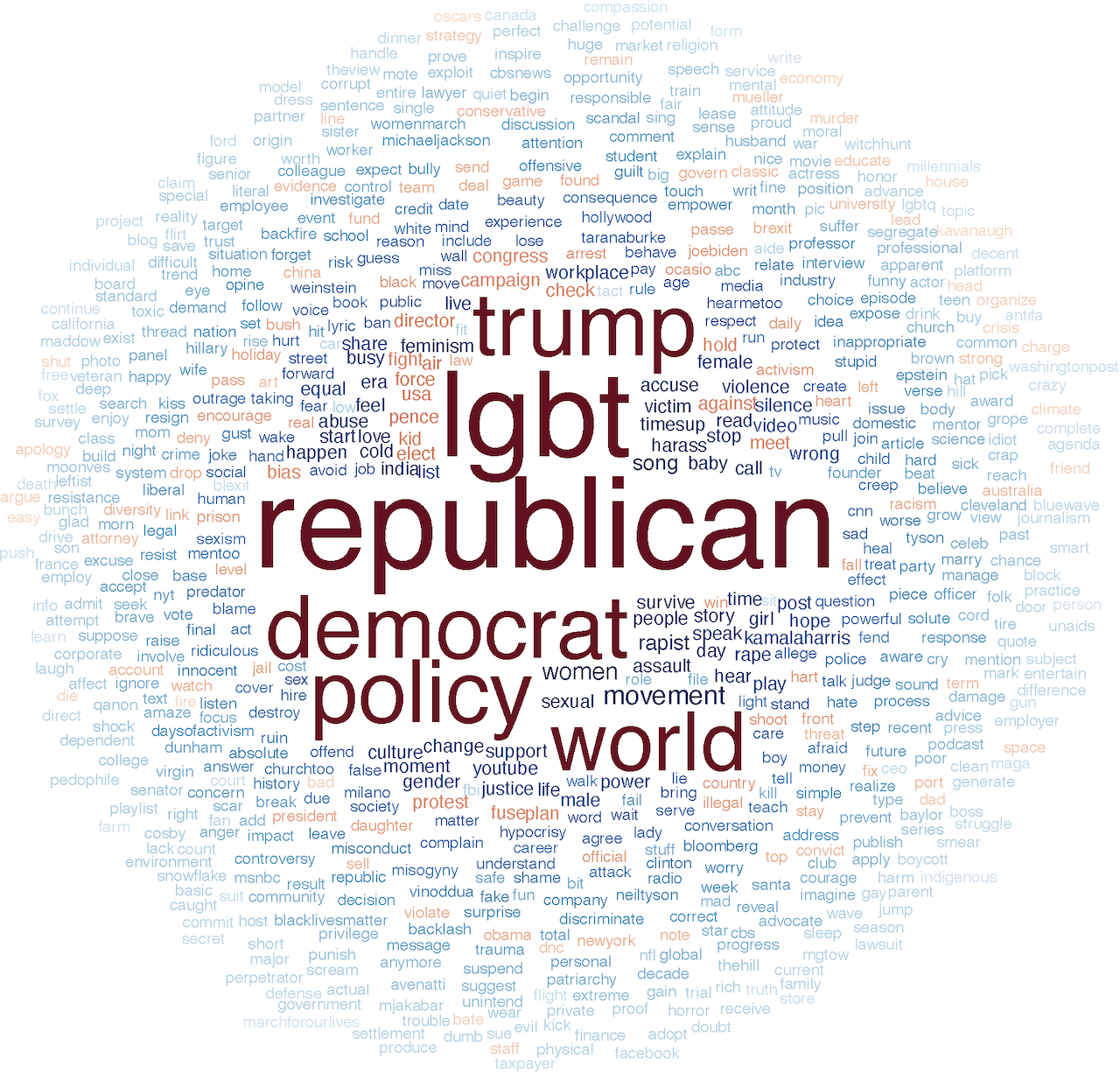}}}
\mbox{
\subfloat[Female workplace exclusion.]{\label{fig-mt-l6}\includegraphics[width =
    0.33\textwidth]{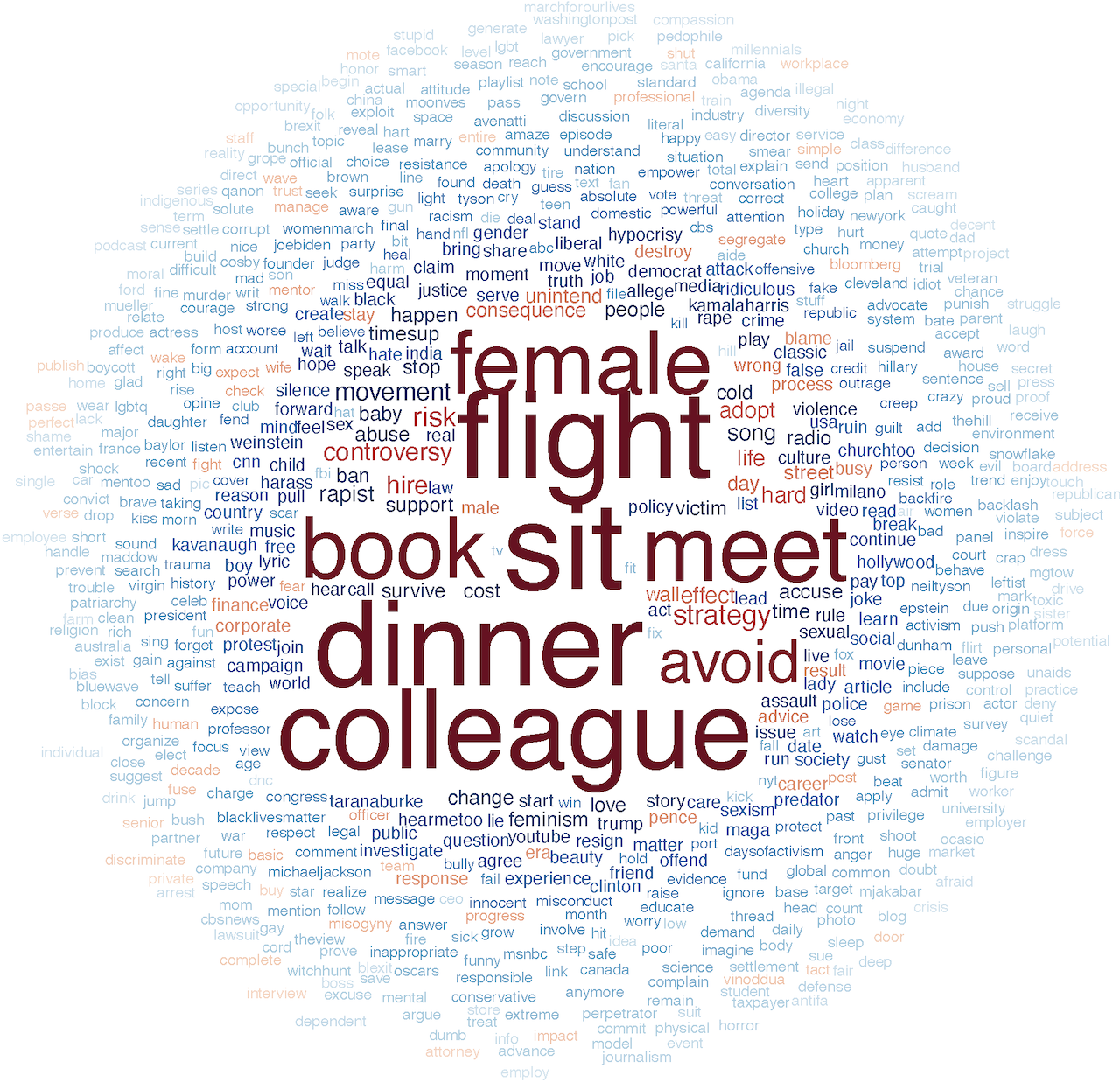}}
  \subfloat[{\tt \#MeToo} across society.]{\label{fig-mt-l7}\includegraphics[width = 0.33\textwidth]{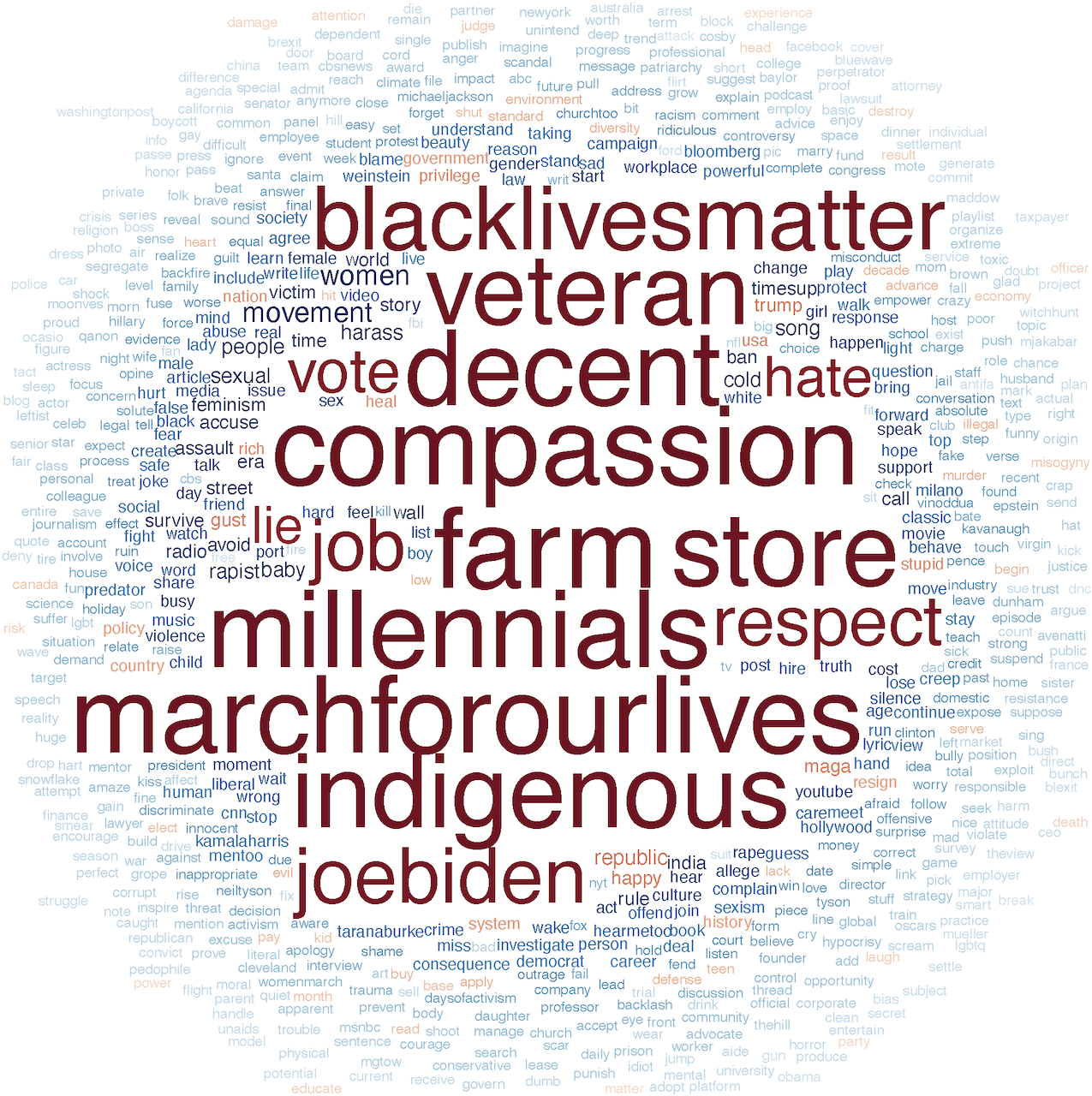}}
 \subfloat[Far-right take on {\tt \#MeToo.}]{\label{fig-mt-l8}\includegraphics[width = 0.33\textwidth]{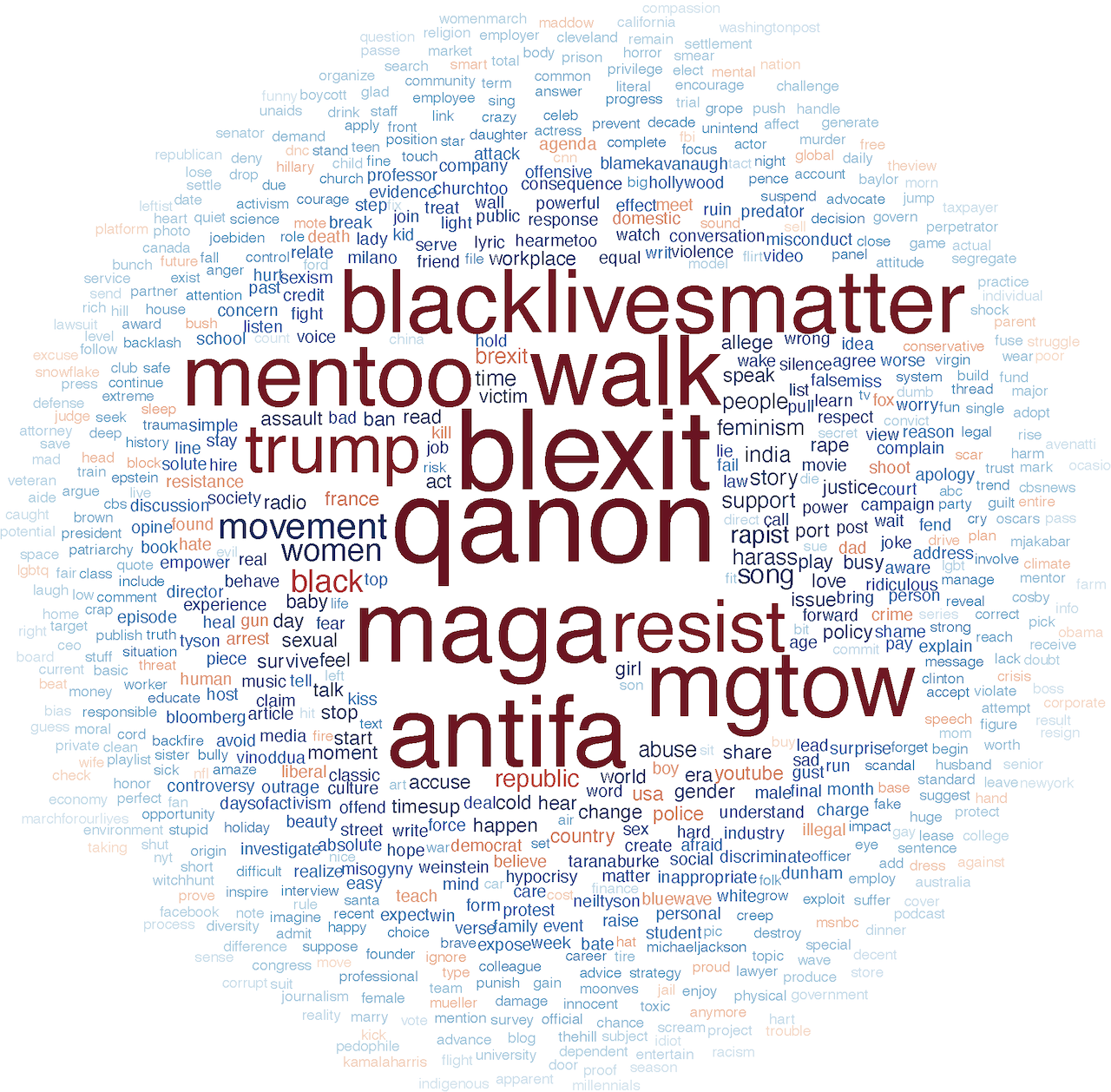}}}
\end{minipage}
\centering
\begin{minipage}{.04\textwidth}
  \raisebox{.09\height}{\includegraphics[width=\textwidth]{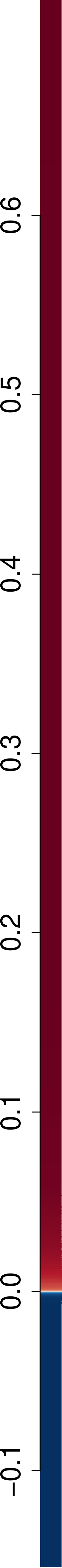}
  }
\end{minipage}
\vspace{-0.2in}
\caption{(a) Estimated uniquenesses (with magnitude 
  proportional to shading and size of word font) from the eight-factor
  PN model fitted to the {\tt \#MeToo} 
  dataset and (b)--(i) the fitted factor loadings after  quartimax
  rotation (here, word color corresponds to the sign and font size and
  shading correspond to magnitude of 
  factor loadings). We characterize each factor loading through a
  descriptive sub-caption.}   
\label{fig-mt-l}
\vspace{-0.15in}
\end{figure}
and is similar to the relative frequency of
words contributing to the main discussion topics in the {\tt
  \#MeToo} movement during the given time period. 
We analyzed the eight factor loadings from the fitted model using
quartimax rotation that simplifies interpretation by minimizing
both the number of heavily loaded features on each factor and the number
of factors needed to explain a feature~\citep{costelloandosborne05}. 
Figs.~\ref{fig-mt-l}b--i provide
{\em word clouds} of the eight factor loadings, ordered by the
proportion of variance in $\bSigma$ explained. 
The word clouds provide insight into the dominant words that
characterize each factor and help describe the 
variability in the {\tt \#MeToo} tweets. We now discuss and interpret the eight
factors.

The first, fourth and sixth factors pertain to words appearing in
tweets discussing different aspects of sexual harassment in the
workplace. For instance, the first factor  is dominated by words that
appear in tweets discussing the December 3, 2018 Bloomberg
article~\citep{tanandporzecanski18a} titled ``Wall Street Rule
For The {\tt \#MeToo} Era: Avoid Women At All Cost" that outlines 
exclusionary strategies reportedly adopted by male Wall Street
employees to avoid allegations in the {\tt \#MeToo} era, but that
simultaneously raise obstacles for professional women 
and exacerbate gender segregation. We notice similar trends in the
sixth factor that heavily weights words appearing in tweets discussing
dilemmas faced by women in the workplace as their male colleagues
limit professional contact ({\em e.g.}, by excluding them from dinner
outings, solo meetings, sitting together while flying) to avoid
accusations of professional sexual misconduct. So, though the words and
topics are similar in these two factors, the first factor is
focused on the Bloomberg article and avoidance strategies
adopted by men on Wall Street while the sixth factor relates to
avoidance issues in the general workplace. The related fourth factor
is a contrast between the dominant terms ``wall'' 
and ``street'' on one hand and ``cost'', ``avoid'', ``era'', ``rule''
and ``women'' on the other. Tweets with these words are generally about the December 5, 2018 Bloomberg article ``NYC officials blast
Wall Street’s ice-out of women in wake of \#MeToo'' \citep{tanandporzecanski18b} in response to the  article
associated with the first factor. These tweets mainly criticise Wall
Street men's reported strategy of avoiding women in the 
workplace, and instead school them on appropriate behaviours in a
professional setting. 

The second factor prominently loads  words that are in tweets reacting
to the cancellation by several radio stations, in the last week of November 2018, of the 1940s classic Christmas song
``Baby It's Cold Outside", in response to listener complaints
about its lyrics describing inappropriate and manipulative male
behaviors toward women. Tweets with these words generally
support the ban, though some others argue that the song was written in a different era. 

The third factor primarily weights words conveying sexual assault,
harassment, abuse and violence, and, to a lesser extent, ``women", ``allege", ``survive" and ``workplace". We interpret this 
factor to  reflect the main objectives of the {\tt \#MeToo}
movement to expose sexual assault and harassment by men in power, especially in the workplace, and to provide support to survivors. 

The remaining factors relate to politics and movements. The fifth factor has terms related to politics or policies,
mostly appearing in tweets discussing the just concluded 2018 midterm 
elections in the United States, where female voters, female candidates and
the {\tt \#MeToo} movement were generally viewed to have played a major role. 
The seventh factor is dominated by words that are in tweets related
to movements for social justice and reform. Examples are tweets
related to gender discrimination and harassment issues across
generations and professions (farm workers, veterans, indigenous
people, etc) or discussions on socio-political matters or allied
movements like ``Black Lives Matter'' or ``March for our Lives'', where
millenials have been quite vocal and active.
Conversely, the eighth factor is overwhelmingly weighted by words
related to right-wing, 
misogynistic, white supremacist and conspiracy theory-based groups
(\#MAGA, QAnon, \#Blexit, \#MenToo, Men Going Their Own Way, etc)
and their views on \#MeToo, ``Black Lives Matter'' and other
anti-fascist and social justice movements.  


Z. Zhu has asked why we do not simply fit a
(generative) Gaussian factor model to the tweets dataset before
$l_2$-normalization. The reason for not doing so is that the
unnormalized tweet vectors cannot be assumed to all come from the same
normal distribution. The processed tweets, after routine
$l_2$-normalization, are assumed to be 
i.i.d. from a PN factor model. The only other currently
realistic model for such data is the Langevin
distribution, but the hypothesis test reported on earlier fit our projected PN 
factor model significantly better than the Langevin model for this
dataset. Our fitted model had eight factors that, upon quartimax
rotation, provided insight into the main
contributors of variability in the tweets.

\subsection{Characterizing the pre-teen brain at rest}
\label{sec:app-ch3-fmri}
\subsubsection{Data preprocessing}
\label{sec:fmri.preprocess}
The  data are from resting-state
measurements~\citep{nebeletal14} of 33 typically developing children
between 8 and 12 years of age. Here, participants were instructed
to relax and fixate on a crosshair while fMRI scans
were acquired 
over seven minutes, at intervals of 2.5s. The first 30s of data were
discarded to let the scanner magnetization achieve a steady state,
yielding 156 time points.  For each registered, motion-corrected
dataset, the AFNI software 
package~\citep{cox96,cox12} used standard preprocessing techniques to regress out nuisance parameters. The residual time series for each
subject at each 
of 75,589 in-brain voxels was normalized by its respective temporal
mean and standard deviation~\citep{chenandwang19}, resulting in
normalized time series on the surface of a 156-dimensional sphere. The
normalized time series at each voxel were then summarized over the 33
pre-teens as the mean direction time series vector at each voxel. At this stage, we have a dataset $\bD$ of $n=75,589$ 156-dimensional 
observations on the unit sphere. However, regression (in this case, of
the nuisance parameters) yields a singular matrix of residuals, so
$\bD$ is projected to a 69-dimensional space by means of a matrix of
orthogonal rows spanning the null space of the regression hat
matrix. (Any lower-dimensional orthogonal projection of a unit norm
vector is also a unit norm vector in lower dimensional space if the
projection vector is orthogonal to the null space of the vector being
projected.) We thus have a dataset on the 69-dimensional unit sphere
of $n=75,589$ observations. 



\subsubsection{Results and analysis}
\label{sec:results.fmri}
FADS-P with eBIC chose the 2-factor model
from among $q=0,1,2,\ldots,10$. The PN factor model fit the
data significantly better ($p\mbox{-value}<0.001$) than its Langevin cousin. We backprojected the factor loadings to the 
original 156-dimensional space and used oblimin rotation for
interpretable results that we now discuss.

Fig.~\ref{fig-td-scores}a displays the time series of the
\begin{figure}[h]
  \centering
\mbox{
  \subfloat[Time series of factor loadings.]{\label{fig-td-trace}
  \hspace{-0.18in}
    \resizebox{.445\textwidth}{!}{\input{figures/main/fmri-td-tobL-trace}}
  }
  \subfloat[ACF of factor loadings.]{\label{fig-td-acf}
  \hspace{-0.1in}
    \resizebox{.275\textwidth}{!}{\input{figures/main/fmri-td-tobL-acf}}
  }
  \subfloat[PACF of factor loadings.]{\label{fig-td-pacf}
    \hspace{-0.1in}
    \resizebox{.275\textwidth}{!}{\input{figures/main/fmri-td-tobL-pacf}}
  }
}
\mbox{
\hspace{-0.03\textwidth}
    \subfloat{
    \begin{minipage}[b][][t]{0.99\textwidth}
      \mbox{
  \setcounter{subfigure}{1}
\vspace{-0.1in}
  \subfloat[First factor scores.]{
    \begin{minipage}[b][][t]{\textwidth} 
      \mbox{\subfloat{\includegraphics[width=.49\textwidth]{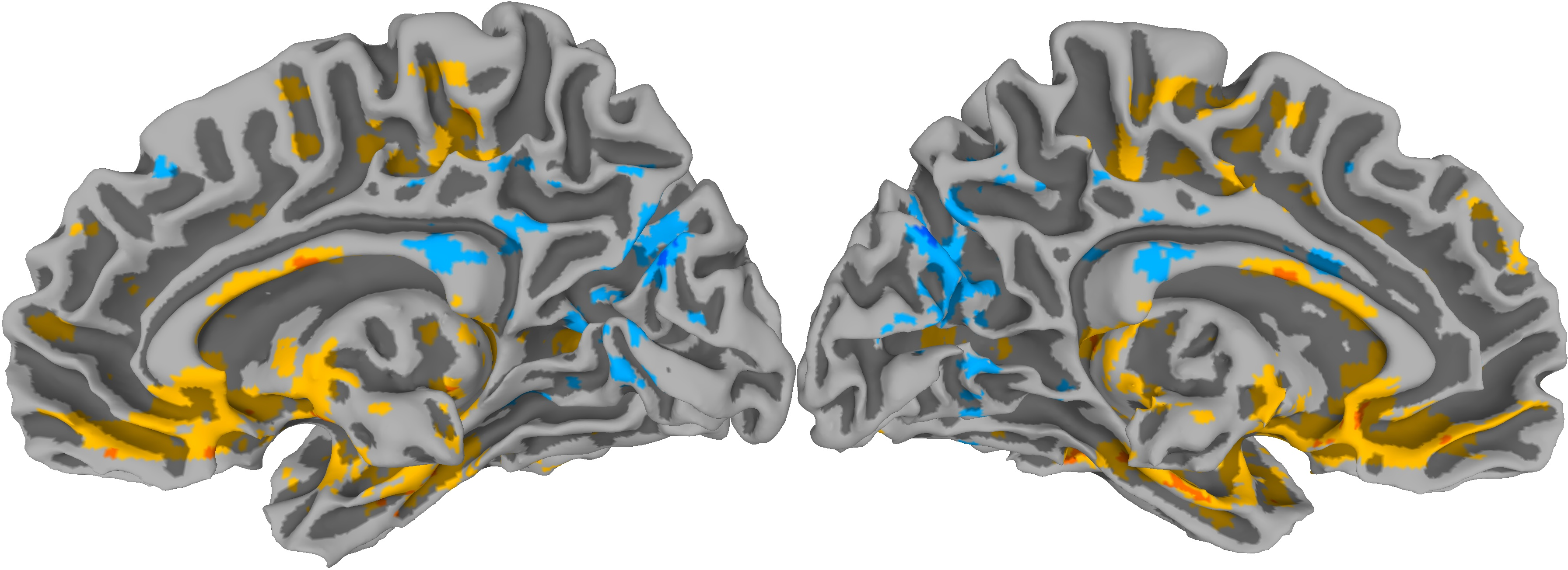}}
        \subfloat{\includegraphics[width=.45\textwidth]{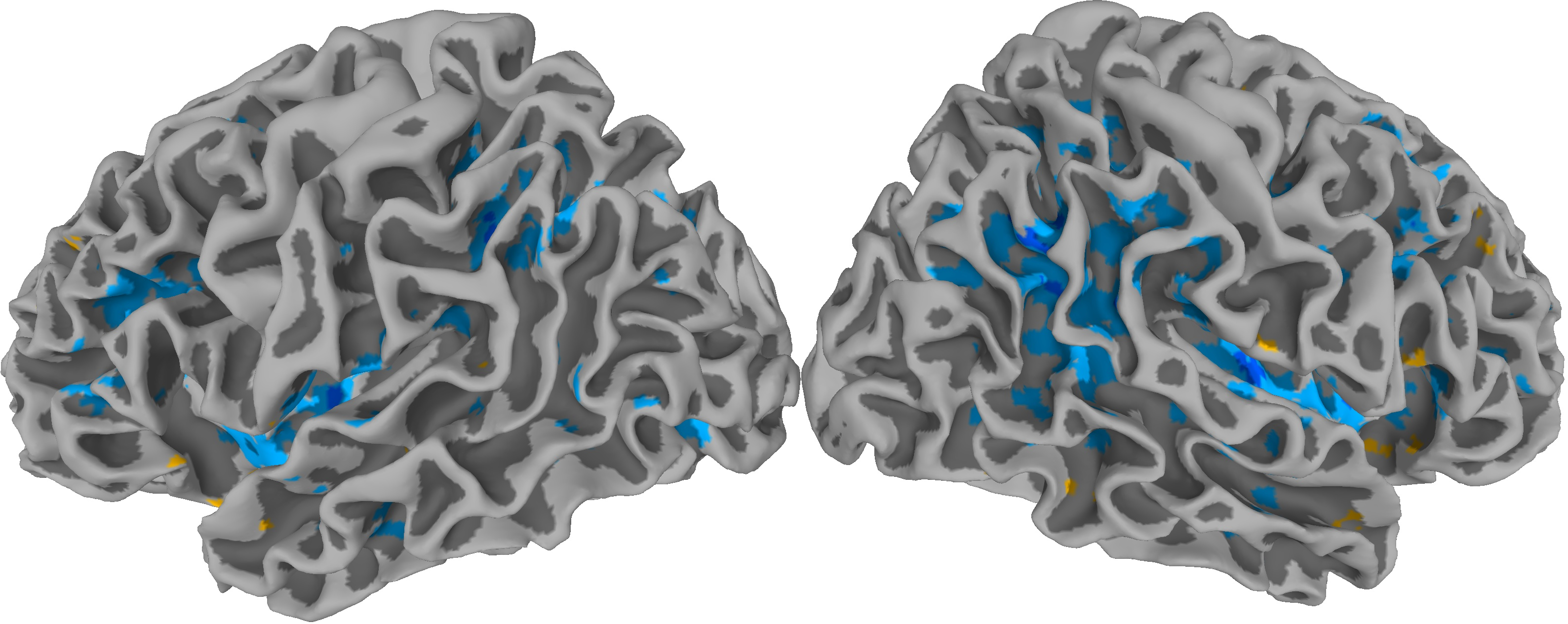}}
      }
      \end{minipage}}%
  }
\mbox{
\vspace{-0.1in}
  \setcounter{subfigure}{2}
  \subfloat[Second factor scores.]{
    \begin{minipage}[b][][t]{\textwidth} 
      \mbox{\subfloat{\includegraphics[width=.49\textwidth]{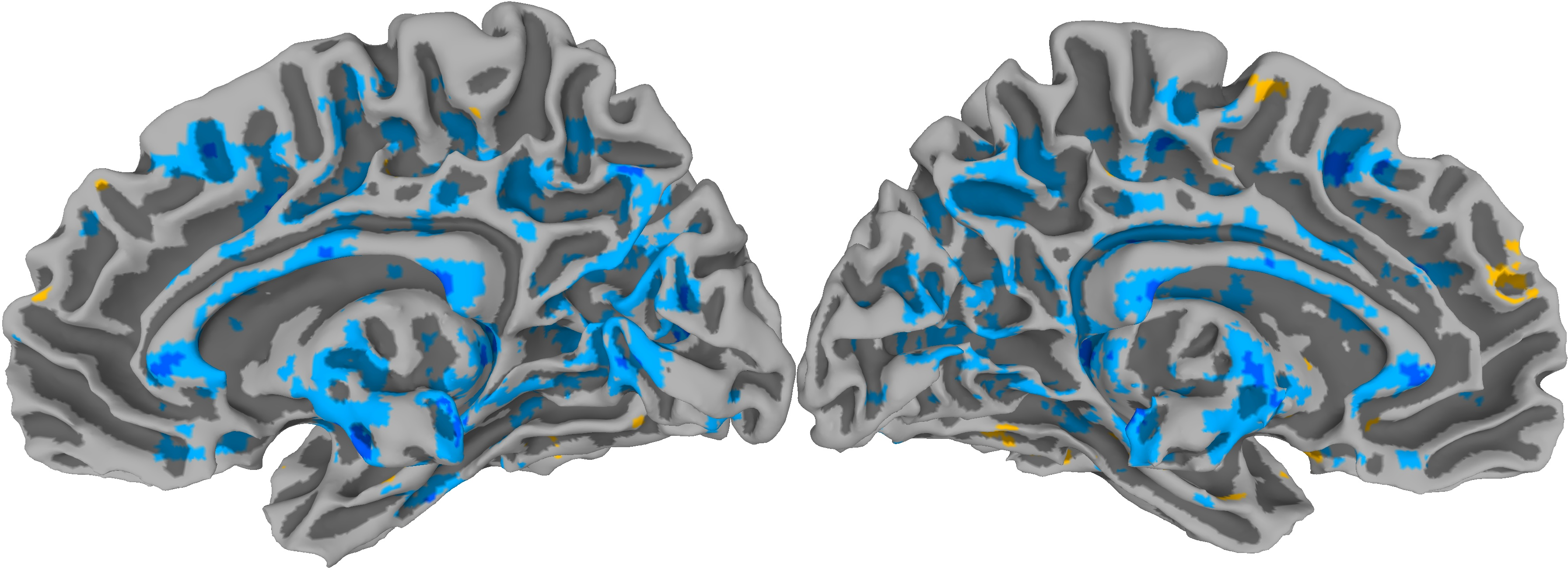}}
        \subfloat{\includegraphics[width=.45\textwidth]{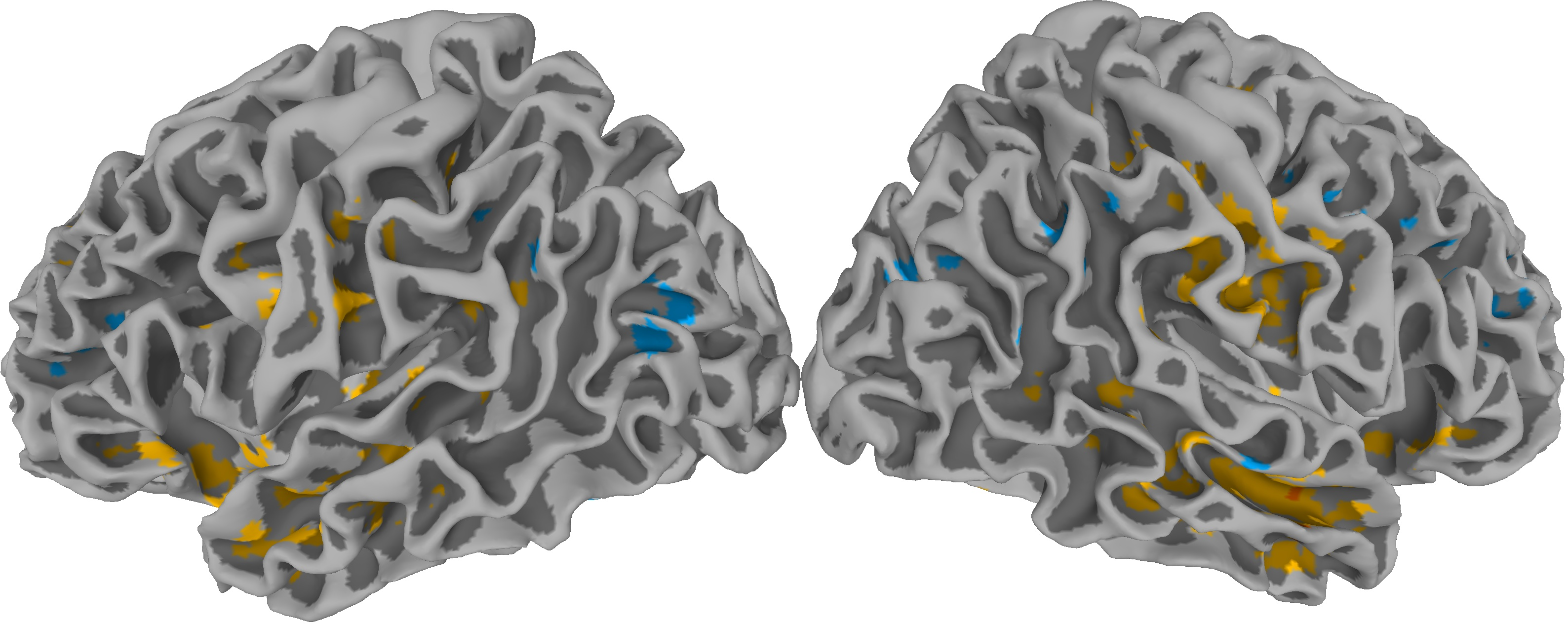}}
      }
      \end{minipage}}%
\vspace{-0.1in}
}
\end{minipage}
}
\hspace{-0.03\textwidth}
 \raisebox{.03\height}{\subfloat{\includegraphics[width=.042\textwidth]{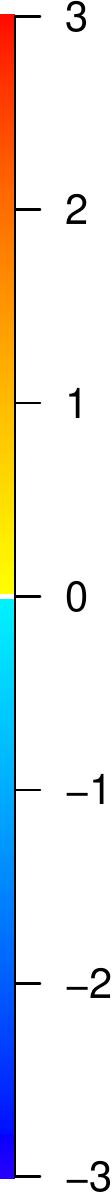}}}
}
\vspace{-0.1in}
  \caption{Fitted factor loadings (a) time series, (b) ACF and (c)
    PACF of the pre-teen brain at rest show distinct characteristics. (d,e) The two predicted factor scores distribute
    to interpretable regions.}
  \label{fig-td-scores}
\end{figure}
oblimin-rotated factor loadings in 156-dimensional space, while
Figs.~\ref{fig-td-scores}b,c plot the
auto-correlation function (ACF) and partial ACF
(PACF) along with $95\%$ confidence 
intervals (broken lines). 
The first factor loadings show larger amplitude fluctuations
between the tenth and fortieth time points than elsewhere. The second
factor loadings fluctuate more frequently than the first, with a
clearer seasonal pattern spanning approximately 15 time-points (37.5s)
in both the time series (Fig.~\ref{fig-td-scores}a) and
the ACF (Fig.~\ref{fig-td-acf}). The large spikes in the first
few lags of the ACF  for both factor loadings, followed by a
cycle of decreasing (in magnitude) alternating negative and positive
correlations, indicate a higher-order autoregression
structure and a strong long-term seasonal component for the second
factor loading. The PACF  (Fig.~\ref{fig-td-pacf}) 
also suggests presence of moving average terms, long-term dependencies and
non-stationary components in both  factor loadings. 
In summary, it appears that the variability over time of the in-brain
voxels of the pre-teen brain at rest can be summarized in terms of two
temporal components, both nonstationary and periodic, but one having 
longer-range dependence than the other. We now discuss the factor
scores. 

Figs.~\ref{fig-td-scores}d,e provide voxel-wise displays of the two
factor scores 
estimated as per Section~\ref{sec:est.scores}. The
first factor scores are high in magnitude in the 
visual and sensory cortex, in areas of emotional and higher mental function (concentration, planning, judgment,
emotional expression, creativity and inhibition) and in the motor
function areas of the cerebellum that
coordinate movement and balance. The second factor scores have
non-negligible values in areas of the cerebral cortex involving
visual, higher 
mental function, sensory association, olfactory and motor functions
that initiate 
and orient the voluntary movement of muscles and the
eye~\citep{sukel19}. We see that the factor scores represent fairly functionally distinctive brain regions.

\subsection{Handwritten digits}
\label{sec:app-ch3-digits}
\subsubsection{Preprocessing}
\label{sec:digits.preprocess}
As per Section~\ref{sec:intro.digits}, the dataset is of 70,000
784-pixel bitmap images. Fig.~\ref{fig-digits-pic} displays 20
random specimens of each digit and shows  variation in
handwriting between (and within) digits. It is this variability that
we want to characterize by  our PN factor model. We removed
pixels with the same value across all coordinates because they contribute no
information on variability. Further, because of the
discretized (bitmapped) values in the dataset, we jittered each of
them by adding i.i.d. 
$\mN(0,0.1^2)$ pseudo-random deviates.
The jittered images are essentially indistinguishable from the originals 
(Fig.~\ref{fig:jittered-digits}) and were $l_2$-normalized for use in
our analysis.

\subsubsection{Results and Analysis}
\begin{figure}
  \vspace{-0.15in}
  \mbox{
    \subfloat[Handwritten digit specimens.]{\includegraphics[width=0.3265\textwidth]{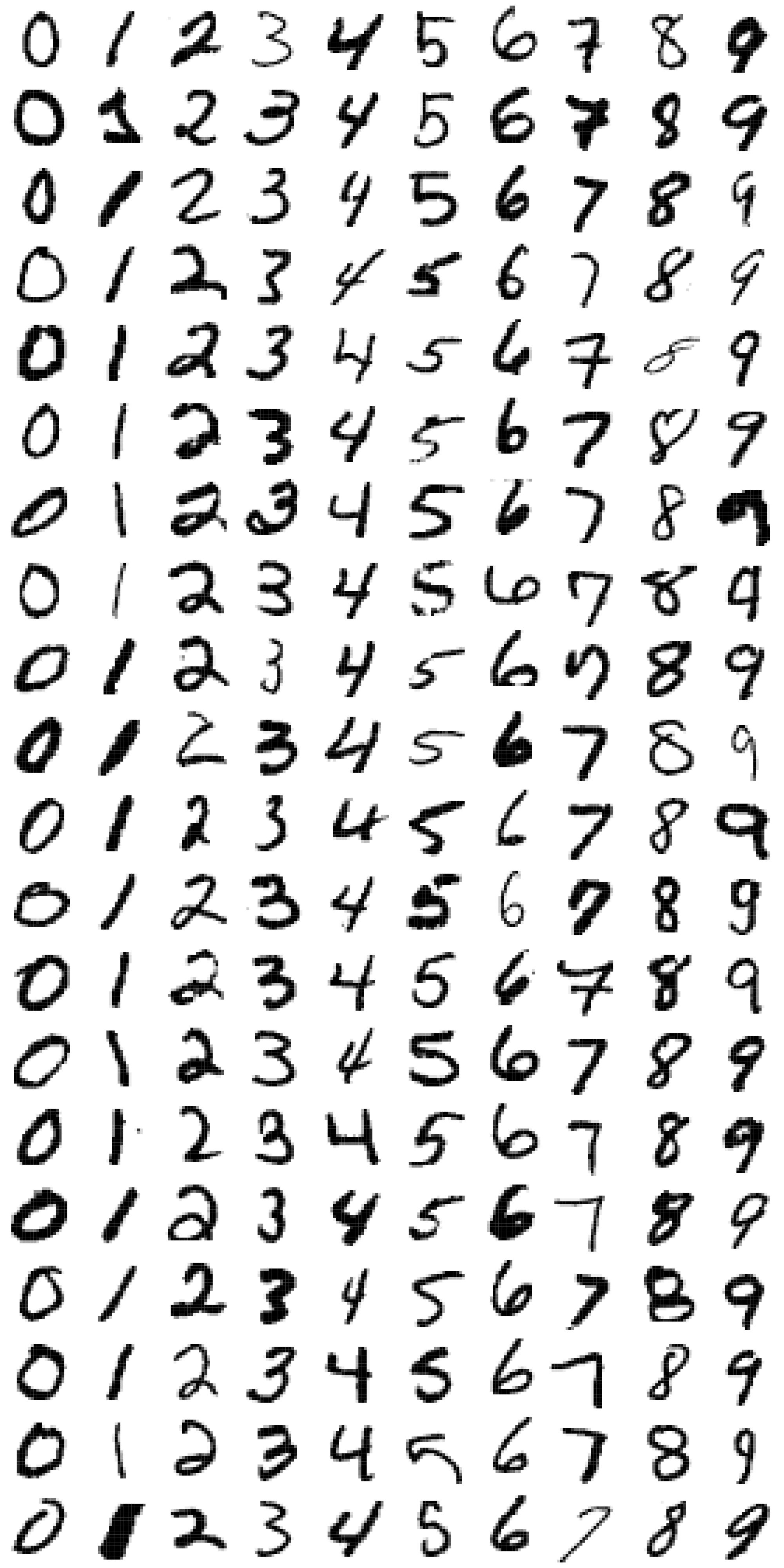}
      \label{fig-digits-pic}}}
  \mbox{
    \setcounter{subfigure}{0}
    \subfloat{
  \resizebox{.0523\textwidth}{!}{\input{figures/main/col-messy-red}}}
    \subfloat[Factor loadings and correlated digits.]{\includegraphics[width = 0.583\textwidth]{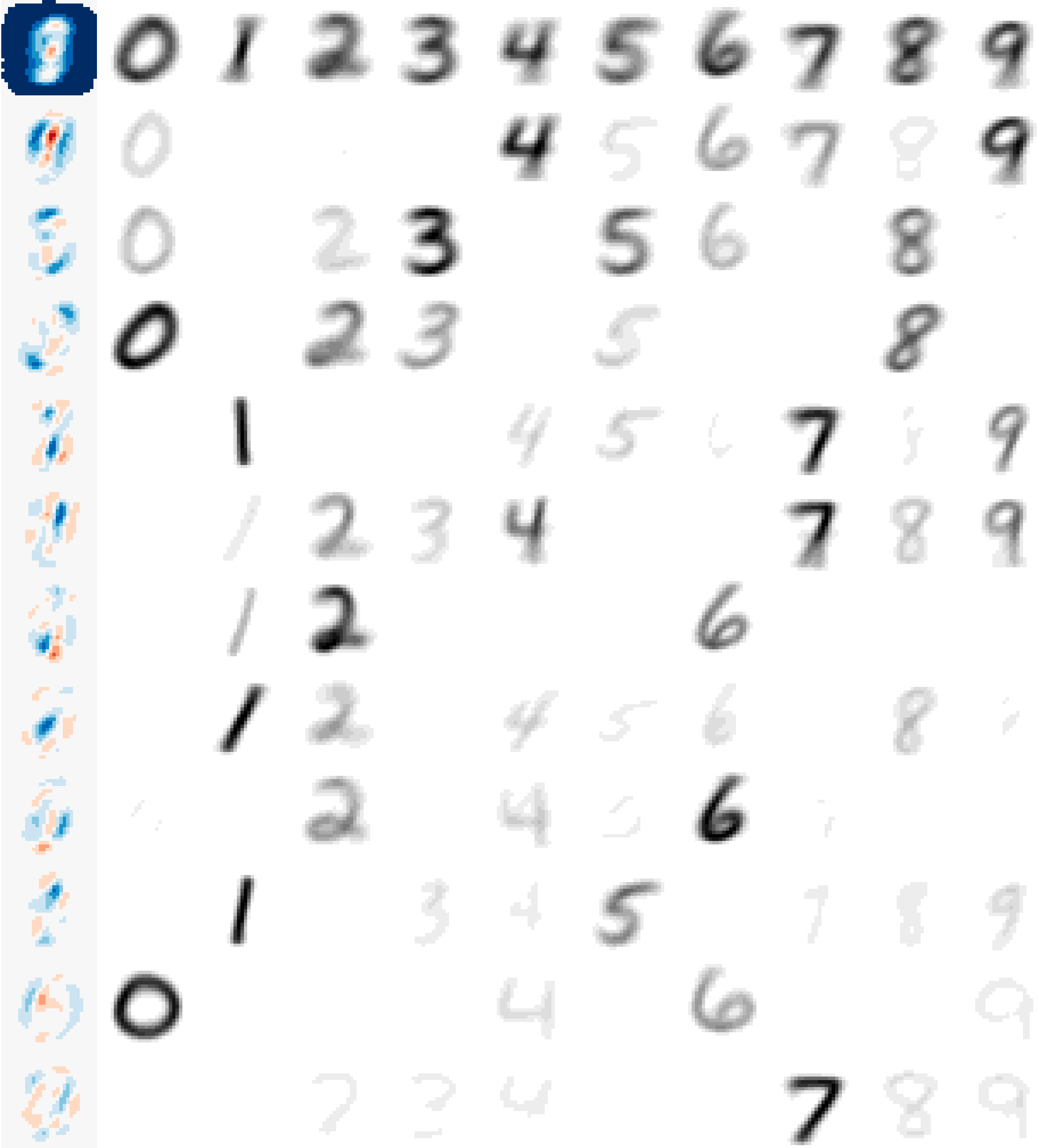}
      \label{fig-digits-m}}
  }
\caption{%
MNIST digits: (a) sample and (b) results. In (b), the first column has the
factor loadings, while the rest have the weighted mean of the
digits most correlated with the loadings of each factor. 
These diplays show   ``0'' as primarily related to the 4th and 11th factor loadings,
    ``1'' with the 5th, 8th, and 10th factors,
    ``2'' with the 4th and 6th--9th factors,
    ``3'' with the 3rd and 4th factors,
    ``4'' with the 2nd and 6th factors,
    ``5'' with the 3rd and 10th factors,
    ``6'' with the 2nd, 3rd, 7th, 9th and 11th factors,
    ``7'' with the 2nd, 5th, 6th, and 12th factors,
    ``8'' with the 3rd and 4th factors, and
    ``9'' with the 2nd, 5th, and 6th factors.
    }
\vspace{-0.02\textheight}
\label{fig-digits}
\end{figure}
\label{sec:digits.results}
The PN factor model was fit using FADS-P for  $q=0,1,\ldots,15$: 
eBIC chose the $12$-factor model that was found to 
 fit the  data significantly
better ($p\mbox{-value}\allowbreak <0.001$) than the Langevin. The
factor loadings after quartimax rotation are displayed, in decreasing
order of their contribution to the variance, in
Fig.~\ref{fig-digits-m} (first 
column) -- dropped pixels were set to zero for the factor loadings. 
For interpretability, we also used absolute correlation to assign each
handwritten digit to its closest factor loading.
Fig.~\ref{fig-digits-m} displays the relative-frequency-weighted mean of
the assigned images for each factor-digit combination. We now discuss
and interpret the factors.

The first factor in Fig.~\ref{fig-digits-m}~(upper left)
is largely an overall contrast between  foreground and background
pixels over all the digits, and consequently is strongly correlated
with samples from each digit.  
The remaining 11 factor loadings are contrasts between different sets
of pixels in the imaging regions. Indeed, each of them are correlated
with different kinds of digits (see rows, and caption of
Fig.~\ref{fig-digits-m}). For example, the second factor 
loading has high association primarily with ``4'' and ``9'' and to a lesser extent ``6'', ``7'' and ``0''. Images
having high correlations with this factor loading have high values at
many of the same pixel locations. Also, the 12th factor loading is primarily
associated with large handwritten ``7''s spanning the entire
image. 
Some digits, such as ``2'' and ``6'', are strongly associated with more factor loadings than others, probably because there are more handwriting styles (in terms of pixels used) for these digits.
D. M. Ommen has pointed out that the third and
eleventh factors are, strikingly, largely associated
with upright-written digits, potentially indicating left-handed
writers. (The MNIST dataset does  not have information on handedness,
so we are unable to verify this conjecture.)

D. Nettleton has asked about  using traditional Gaussian
factor analysis on this dataset, that is, essentially ignoring the 
constraints imposed by the digitization of a bitmap image. Doing so
yielded poor results, unlike FADS that not only accounts for the
imaging data acquisition 
process but also identified interpretable factors that capture the
variability in handwriting styles both between and within digits. We
remark that our methodology has not accounted for 
spatial context or pixel neighborhood. Incorporating these
aspects may improve our results further.


\subsection{TCGA dataset}
\label{sec:app-ch3-rnaseq}
\subsubsection{Data collection and preprocessing}
We use an mRNA sequence (RNA-seq) dataset comprised of 775 primary breast cancer (BC) tumors from the TCGA effort~\citep{weinsteinetal13}, prepared and released in the Bioconductor package \texttt{curatedTCGAData}~\citep{ramos20}.
RNA-seq leverages high-throughput next-generation sequencing
technology~\citep{metzker10} to examine the presence and quantity of
mRNA from each gene in a biological sample by sequencing
random deoxyribonucleic acid fragments derived from the sampled
mRNA~\citep{wang2008rnaseq}. 
During data preparation, the sampled sequence fragments were mapped to
20,502 genes, producing a count of fragments for each gene. Since the
total number of fragments sequenced (coverage) is not 
biologically meaningful, we scaled the counts in each experiment to
sum to unity, and then screened out low expression genes, by
discarding genes with mean scaled expressions below the third
quartile, producing our dataset on 5,123 genes.
After rescaling the 5,123 surviving gene proportions to
sum to unity in each of the 775 experiments,
we applied a square root transformation to place the compositional data on the unit sphere, as in~\citet{daiandmuller18}.

\subsubsection{Results and analysis}
FADS with eBIC selected a 12-factor model when fitting up to 15 factors.
The first eight of these factors explained 95\% of the variation in the data.
This model also significantly ($p\mbox{-value}<0.001$) out-performed the Langevin model in fitting the processed data. 
We used quartimax rotation for further analysis.
To determine if any factor could be prognostic, we fitted the Cox proportional-hazards model \citep{cox1972} to the survival time of the 775 observations using each factor score in turn as the predictor.
A likelihood ratio test indicated significant association of the sixth
factor with survival (Holm's corrected $p$-value of 0.005). 
Fig.~\ref{fig:survival} displays survival curves with the $95\%$
confidence intervals, which are estimated using the Kaplan–Meier
method \citep{kaplanmeier58}, and obtained after thresholding on the
median value (1.09925) of the sixth factor scores. The two survival
curves are significantly different ($p$-value of 0.00069 via a logrank test~\citep{Mantel1966}). Further, high values of the sixth factor correspond to increased survival, especially between years 2--8. 
\begin{figure}
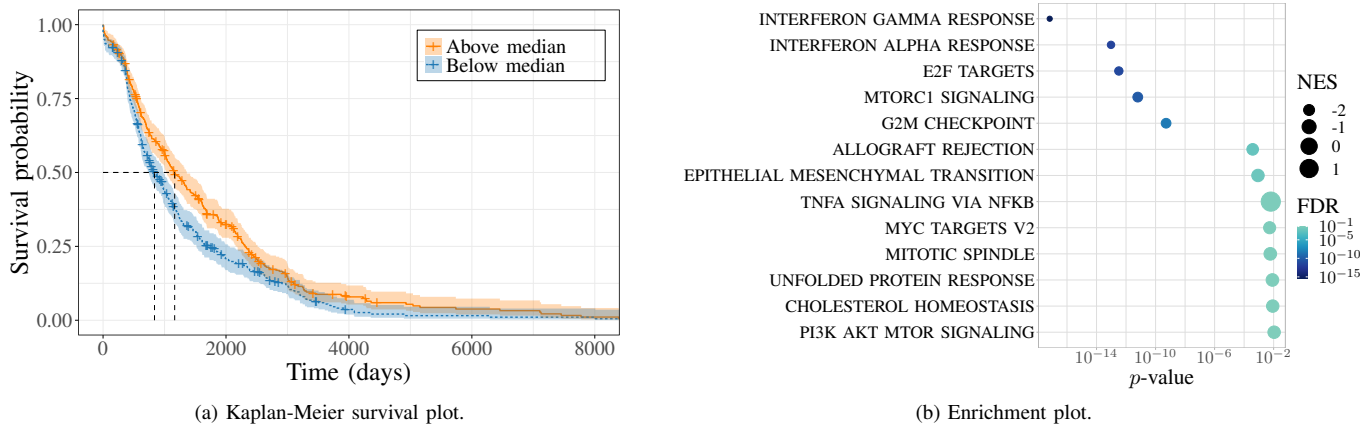

  \centering
  \mbox{
    \subfloat[Kaplan-Meier survival plot.]{\label{fig:survival}
    \hspace{-0.18in}
      \resizebox{.46\textwidth}{!}{\input{figures/main/tcga_surv}}
    }
    \hspace{-0.12in}
    \subfloat[Enrichment plot.]{\label{fig:gsea}
      \resizebox{.54\textwidth}{!}{\input{figures/main/tcga_gsea}}
    }
}
\caption{(a) Kaplan-Meier survival plot split by median value of the
  sixth factor. Orange cases are above and blue are below the
  median. The horizontal and vertical dashed lines indicate the median
  survival times for the two groups. (b) Enrichment plot. Gene sets are labeled on y-axis.}
\label{fig-tcga}
\end{figure}

The sixth factor can be associated with genes to identify possible prognostic biomarkers or biological processes to hypothesize a mechanism for altered survival. 
In particular, it is common to identify the biological processes
associated with gene-level statistics using gene set enrichment
analysis (GSEA)~\citep{subramanian2005}.
We performed GSEA using the sixth factor loadings, on the correlation scale, as the gene-level statistic and an enrichment $p$-value was calculated for each of the 50 hallmark gene sets in the MSigDB collections, version 5~\citep{Liberzon2015}. 
GSEA found the sixth factor enriched for 13 of the 50 hallmark gene
sets with false discovery rate (FDR) adjusted p-values $<0.05$. 
Figure~\ref{fig:gsea} presents the normalized enrichment scores (NES)
and FDR-adjusted $p$-values of these thirteen selected gene sets. 

We examine the five gene sets with FDR-adjusted $p$-values below $10^{-8}$, comparing our findings to the literature.
The \verb#INTERFERON_GAMMA_RESPONSE# and
\verb#INTERFERON_ALPHA_RESPONSE# gene sets are immunity-related and
substantially overlapped hallmark gene sets (Jaccard index $0.33$).
Though the role of the immune response in BC is far from clear~\citep{Gatti-Mays2019}, the constellation of genes activated in response to the interferon-$\gamma$ cell-signaling cytokine have recently been associated with poor prognosis in another cohort of BC patients~\citep{Heimes2020}. 
The top three loaded genes in \textit{both} gene sets are IFI27, ISG15, and RSAD2.
	IFI27 
		physically associates with estrogen receptor $\alpha$, an important mediator of proliferation in BC~\citep{Cervantes-Badillo2020},
	ISG15 modifies other proteins through ISGylation and correlates with lymphovascular invasion and poor prognosis~\citep{Bektas2008,Kariri2021},
	and RSAD2, known to inhibit virus replication~\citep{Rivera-Serrano2020}, associates with tumor stage and lymph node metasteses~\citep{Tang2020}.
The \verb#E2F_TARGETS# and \verb#G2M_CHECKPOINT# gene sets are involved in cell proliferation, a commonly upregulated function in tumor samples~\citep{Chen2021}.
E2F transcription factors participate in cell cycle control, and several have recently been associated with BC prognosis~\citep{Sun2019}, 
while the entire \verb#G2M_CHECKPOINT# hallmark gene set has been deemed a prognostic biomarker of metastasis in one type of BC~\citep{Oshi2020}.	
Transcription factor MYBL2 is the most (negatively) loaded gene on the sixth factor in \textit{both} gene sets, 
	while TOP2A is the second most (negatively) loaded gene among E2F targets and the third most (negatively) loaded G2M checkpoint gene.
MYBL2 directly targets several kinesin motor proteins that are correlated with poor BC prognosis~\citep{Wolter2017}.
TOP2A is a DNA topoisomerase whose expression has been correlated with poor BC prognosis~\citep{Ogino2020,Chen2021}. 
In fact, it is a direct target of some chemotherapies, including a trial where its expression was successfully used to guide treatment~\citep{Li2021}.
The gene set \verb#MTORC1_SIGNALING# consists of genes upregulated
through activation of MTORC1, a protein complex responsible for protein synthesis 
	that has also been targeted by chemotherapeutic agents for BC~\citep{Seto2012}.
MTORC1 appears frequently in the BC literature as several prognostic markers are in MTORC1-related gene sets~\citep{Qiu2018, Takeshita2019, Wang2020, Cai2021},
Here, SQLE is the most (negatively) loaded gene on the MTORC1 gene set and the second most (negatively) loaded gene on the G2M gene set.
It is an enzyme involved in cholesterol synthesis, with an impact on proliferation, that has been associated with poor prognosis~\citep{Brown2016,Kim2021b}.
Thus, FADS finds 
	a single, highly interpretable, apparently proliferation-related, axis of variation significantly associated with BC survival.

\section{Discussion}
\label{sec:disc}
This paper develops methodology for exploratory factor
analysis of data on a sphere, with a view to explaining the
variability in the coordinates in terms of a few underlying (latent)
factors. Such datasets arise in applications involving text documents,
time series image volumes of the brain, digitized
images, or compositional gene expression data, that we investigate in this paper. 
Our approach uses the PN distribution~\citep{mardiaandjupp00}
which is the distribution of a normalized $\mN_p(\bmu,\bSigma)$ random
vector, with the restriction that $\bmu\in{\mathcal S}_{p-1}$ for
identifiability. Our factor model is introduced through a
decomposition of the generative MVN model. We develop
computationally fast estimation methods for the parameters using two
versions of the AECM algorithm. The first, called FADS-D, uses the
latent factors and lengths of the random vectors of our generative model
as the missing variables, while the second approach, called FADS-P,
profiles out the latent factors
in the generative model. We also provided uncertainty estimators for
our models. Our methods are implemented in the developmental version
of  our \texttt{fad} package in  R~\citep{R} that is available at
\url{https://github.com/somakd/fad}.
Simulation results show fast and excellent
performance, especially of FADS-P with the number of factors
well-estimated using eBIC. Analysis of datasets from four
motivating applications provided interpretable insights into the
underlying mechanisms governing variability in each case. 

There are a number of points that could benefit from further
attention. For instance, our PN factor model (and indeed the MVN
factor model itself) could benefit from
results on consistency of the eBIC. Further, our model could also be used in regression and 
classification settings, replacing, for instance, the Fisher
distribution, in the development of the spherical regression model
of~\citet{downs03}, with appropriate 
modifications. Such an approach would allow us to characterize resting
state brains in autistic and typically developing children, for
example. Other extensions to be investigated include Bayesian
methods for estimation, as well as including sparsity structures on
$\bLambda$. It may also be of interest to incorporate the PN model in
settings that involve a combination of $p$ linear variables and $q$
variables on $\mathcal S_{q-1}$. From an applications perspective, we
may also be interested in using factor scores in genome-wide
association studies to account for population structure. Therefore, we
see that there are many directions for further investigation. 
 \section*{Acknowledgments}
 The  research was
supported in  part by the United States Department
  of Agriculture (USDA)/National Institute of Food and
  Agriculture (NIFA), Hatch projects IOW03617 and IOW03717. The
  content of this paper however is 
  solely the responsibility of the  authors and does not represent the
  official views of the USDA. 
\ifCLASSOPTIONcaptionsoff
  \newpage
\fi
\bibliographystyle{IEEEtran}
\bibliography{references,acm}

\begin{thebibliography}{100}
\providecommand{\url}[1]{#1}
\csname url@samestyle\endcsname
\providecommand{\newblock}{\relax}
\providecommand{\bibinfo}[2]{#2}
\providecommand{\BIBentrySTDinterwordspacing}{\spaceskip=0pt\relax}
\providecommand{\BIBentryALTinterwordstretchfactor}{4}
\providecommand{\BIBentryALTinterwordspacing}{\spaceskip=\fontdimen2\font plus
\BIBentryALTinterwordstretchfactor\fontdimen3\font minus
  \fontdimen4\font\relax}
\providecommand{\BIBforeignlanguage}[2]{{%
\expandafter\ifx\csname l@#1\endcsname\relax
\typeout{** WARNING: IEEEtran.bst: No hyphenation pattern has been}%
\typeout{** loaded for the language `#1'. Using the pattern for}%
\typeout{** the default language instead.}%
\else
\language=\csname l@#1\endcsname
\fi
#2}}
\providecommand{\BIBdecl}{\relax}
\BIBdecl

\bibitem{watson83}
G.~S. Watson, \emph{Statistics on Spheres}, ser. University of Arkansas Lecture
  Notes in the Mathematical Sciences.\hskip 1em plus 0.5em minus 0.4em\relax
  Wiley Interscience, 1983.

\bibitem{fisheretal93}
N.~Fisher, T.~Lewis, and B.~Embleton, \emph{Statistical Analysis of Spherical
  Data}.\hskip 1em plus 0.5em minus 0.4em\relax Cambridge University Press,
  1993.

\bibitem{jammalamadakaandsengupta01}
S.~R. Jammalamadaka and A.~SenGupta, \emph{Topics in Circular
  Statistics}.\hskip 1em plus 0.5em minus 0.4em\relax World Scientific, 2001.

\bibitem{mardiaandjupp00}
K.~V. Mardia and P.~E. Jupp, \emph{Directional Statistics}.\hskip 1em plus
  0.5em minus 0.4em\relax New York: Wiley, 2000.

\bibitem{vonmises18}
R.~{von Mises}, ``\"uber die ``ganzzahligkeit'' der atomgewichte und verwandte
  fragen,'' \emph{Physikalische Zeitung}, vol.~19, pp. 490--500, 1918.

\bibitem{fisher53}
R.~A. Fisher, ``Dispersion on a sphere,'' \emph{Proc. R. Soc. Lond. A}, vol.
  217, pp. 295--305, 1953.

\bibitem{senguptaandmaitra98}
A.~SenGupta and R.~Maitra, ``On best equivariance and admissibility of
  simultaneous {MLE} for mean direction vectors of several {L}angevin
  distributions,'' \emph{Annals of the Institute of Statistical Mathematics},
  vol.~50, no.~4, pp. 715--27, 1998.

\bibitem{saltonandbuckley88}
G.~Salton and C.~Buckley, ``Term-weigting approaches in automatic text
  retrieval,'' \emph{Information Processing and Management}, vol.~27, no.~5,
  pp. 513--523, 1988.

\bibitem{kolda97}
T.~G. Kolda, ``Limited-memory matrix methods with applications,'' Ph.D.
  dissertation, The Applied Mathematics Program, University of Maryland,
  College Park, Maryland, 1997.

\bibitem{dhillonandmodha01}
I.~S. Dhillon and D.~S. Modha, ``Concept decompositions for large sparse text
  data using clustering,'' \emph{Machine Learning}, vol.~42, pp. 143--175,
  2001.

\bibitem{singhaletal96}
A.~Singhal, C.~Buckley, M.~Mitra, and G.~Salton, ``Pivoted document length
  normalization,'' in \emph{{ACM SIGIR'96}}, 1996, pp. 21--29.

\bibitem{banerjeeetal05}
A.~Banerjee, I.~S. Dhillon, J.~Ghosh, and S.~Sra, ``Clustering on the unit
  hypesphere using {von Mises-Fisher} distributions,'' \emph{Journal of Machine
  Learning Research}, vol.~6, pp. 1345--1382, 2005.

\bibitem{biswaletal95}
B.~Biswal, F.~Z. Yetkin, V.~M. Haughton, and J.~S. Hyde, ``Functional
  connectivity in the motor cortex of resting human brain using echo-planar
  {MRI},'' \emph{Magnetic Resonance in Medicine}, vol.~34, no.~4, pp. 537--541,
  1995.

\bibitem{biswal12}
B.~B. Biswal, ``Resting state {fMRI}: A personal history,'' \emph{NeuroImage},
  vol.~62, no.~2, pp. 938 -- 944, 2012.

\bibitem{belliveauetal91}
J.~W. Belliveau, D.~N. Kennedy, R.~C. McKinstry, B.~R. Buchbinder, R.~M.
  Weisskoff, M.~S. Cohen, J.~M. Vevea, T.~J. Brady, and B.~R. Rosen,
  ``Functional mapping of the human visual cortex by magnetic resonance
  imaging,'' \emph{Science}, vol. 254, pp. 716--719, 1991.

\bibitem{kwongetal92}
K.~K. Kwong, J.~W. Belliveau, D.~A. Chesler, I.~E. Goldberg, R.~M. Weisskoff,
  B.~P. Poncelet, D.~N. Kennedy, B.~E. Hoppel, M.~S. Cohen, R.~Turner, H.-M.
  Cheng, T.~J. Brady, and B.~R. Rosen, ``Dynamic magnetic resonance imaging of
  human brain activity during primary sensory stimulation,'' \emph{Proceedings
  of the National Academy of Sciences of the United States of America},
  vol.~89, pp. 5675--5679, 1992.

\bibitem{bandettinietal93}
P.~A. Bandettini, A.~Jesmanowicz, E.~C. Wong, and J.~S. Hyde, ``Processing
  strategies for time-course data sets in functional mri of the human brain,''
  \emph{Magnetic Resonance in Medicine}, vol.~30, pp. 161--173, 1993.

\bibitem{lindquist08}
M.~A. Lindquist, ``The statistical analysis of {fMRI} data,'' \emph{Statistical
  Science}, vol.~23, no.~4, pp. 439--464, 2008.

\bibitem{lazar08}
N.~A. Lazar, \emph{The Statistical Analysis of Functional MRI Data}.\hskip 1em
  plus 0.5em minus 0.4em\relax Springer, 2008.

\bibitem{foxetal05}
M.~D. Fox, A.~Z. Snyder, J.~L. Vincent, M.~Corbetta, D.~C. Van~Essen, and M.~E.
  Raichle, ``The human brain is intrinsically organized into dynamic,
  anticorrelated functional networks,'' \emph{Proceedings of the National
  Academy of Sciences}, vol. 102, no.~27, pp. 9673--9678, 2005.

\bibitem{coleetal10}
D.~Cole, S.~Smith, and C.~Beckmann, ``Advances and pitfalls in the analysis and
  interpretation of resting-state {fMRI} data,'' \emph{Frontiers in Systems
  Neuroscience}, vol.~4, p.~8, 2010.

\bibitem{maitraandramler10}
R.~Maitra and I.~P. Ramler, ``A $k$-mean-directions algorithm for efficient
  clustering of data on a sphere,'' \emph{Journal of Computational and
  Graphical Statistics}, vol.~19, no.~2, pp. 377--396, 2010.

\bibitem{LeCun2019digits}
\BIBentryALTinterwordspacing
Y.~LeCun, C.~Cortes, and C.~J.~C. Burges, ``The mnist database,'' 2019.
  [Online]. Available: \url{http://yann.lecun.com/exdb/mnist/}
\BIBentrySTDinterwordspacing

\bibitem{daiandmuller18}
X.~Dai and H.-G. M\"uller, ``Principal component analysis for functional data
  on {R}iemannian manifolds and spheres,'' \emph{Annals of Statistics},
  vol.~46, no.~6B, pp. 3334--3361, 12 2018.

\bibitem{bose2019chaudhari}
K.~Bose and P.~Chaudhari, \emph{Unravelling Cancer Signaling Pathways: A
  Multidisciplinary Approach}.\hskip 1em plus 0.5em minus 0.4em\relax Springer,
  2019.

\bibitem{mardiaetal79}
K.~V. Mardia, J.~T. Kent, and J.~M. Bibby, \emph{Multivariate analysis}.\hskip
  1em plus 0.5em minus 0.4em\relax New York: Academic Press, 1979.

\bibitem{dortetbernadetandwicker08}
J.~Dortet-Bernadet and N.~Wicker, ``Model-based clustering on the unit sphere
  with an illustration using gene expression profiles,'' \emph{Biostatistics},
  vol.~9, no.~1, pp. 66--80, 2008.

\bibitem{bingham74}
C.~Bingham, ``An antipodally symmetric distribution on the sphere,''
  \emph{Annals of Statistics}, vol.~2, pp. 1201--1225, 1974.

\bibitem{kent82}
J.~T. Kent, ``The {F}isher–{B}ingham distribution on the sphere,''
  \emph{Journal of the Royal Statistical Society}, vol.~44, pp. 71--80, 1982.

\bibitem{dryden05}
I.~Dryden, ``Statistical analysis on high-dimensional spheres and shape
  spaces,'' \emph{Annals of Statistics}, vol.~33, pp. 1643--1665, 09 2005.

\bibitem{kimetal17}
B.~Kim, S.~Huckemann, J.~Schulz, and S.~Jung, ``Small sphere distributions for
  directional data with application to medical imaging,'' \emph{Scandinavian
  Journal of Statistics}, pp. 432--444, 05 2017.

\bibitem{wangandgelfand13}
F.~Wang and A.~E. Gelfand, ``Directional data analysis under the general
  projected normal distribution,'' \emph{Statistical Methodology}, vol.~10, pp.
  113--127, 07 2013.

\bibitem{stumpfhauseretal17}
D.~Hernandez-Stumpfhauser, F.~Breidt, and M.~Woerd, ``The general projected
  normal distribution of arbitrary dimension: Modeling and {B}ayesian
  inference,'' \emph{Bayesian Analysis}, vol.~12, pp. 113--133, 01 2016.

\bibitem{mengandvandyk97}
X.~Meng and D.~van Dyk, ``The {EM} algorithm --- an old folk-song sung to a
  fast new tune (with discussion),'' \emph{Journal of the Royal Statistical
  Society B}, vol.~59, pp. 511--567, 1997.

\bibitem{varadhanandroland08}
R.~Varadhan and C.~Roland, ``Simple and globally convergent methods for
  accelerating the convergence of any em algorithm,'' \emph{Scandinavian
  Journal of Statistics}, vol.~35, pp. 335--353, 06 2008.

\bibitem{chenandchen08}
J.~Chen and Z.~Chen, ``Extended bayesian information criteria for model
  selection with large model spaces,'' \emph{Biometrika}, vol.~95, no.~3, pp.
  759--771, 2008.

\bibitem{pukkilaandrao88}
T.~M. Pukkila and C.~R. Rao, ``Pattern recognition based on scale invariant
  discriminant functions,'' \emph{Information sciences}, vol.~45, no.~3, pp.
  379--389, 1988.

\bibitem{paineetal18}
P.~J. Paine, S.~Preston, M.~Tsagris, and A.~Wood, ``An elliptically symmetric
  angular gaussian distribution,'' \emph{Statistics and Computing}, vol.~28,
  pp. 689--697, 05 2018.

\bibitem{lawley40}
D.~N. Lawley, ``The estimation of factor loadings by the method of maximum
  likelihood,'' \emph{Proceedings of the Royal Society of Edinburgh}, vol.~60,
  no.~1, p. 64–82, 1940.

\bibitem{joreskog66}
K.~G. J\"oreskog, ``Some contributions to maximum likelihood factor analysis,''
  \emph{ETS Research Bulletin Series}, no.~2, pp. 1--54, 1966.

\bibitem{lawleyetal71}
D.~N. Lawley, A.~E. Maxwell, and A.~Ernest,
  \emph{\BIBforeignlanguage{English}{Factor analysis as a statistical method}},
  2nd~ed.\hskip 1em plus 0.5em minus 0.4em\relax New York : American Elsevier
  Publishing Company, 1971.

\bibitem{anderson03}
T.~W. Anderson, \emph{An Introduction to multivariate statistical analysis},
  ser. Wiley Series in Probability and Statistics.\hskip 1em plus 0.5em minus
  0.4em\relax Wiley, 2003.

\bibitem{owenandwang15}
A.~B. Owen and J.~Wang, ``Bi-cross-validation for factor analysis,''
  \emph{Statistical Science}, vol.~31, pp. 119--139, 03 2015.

\bibitem{daietal20}
F.~Dai, S.~Dutta, and R.~Maitra, ``A matrix-free likelihood method for
  exploratory factor analysis of high-dimensional gaussian data,''
  \emph{Journal of Computational and Graphical Statistics}, vol.~29, no.~3, pp.
  675--680, 2020.

\bibitem{costelloandosborne05}
A.~B. Costello and J.~Osborne, ``Best practices in exploratory factor analysis:
  Four recommendations for getting the most from your analysis,''
  \emph{Practical Assessment, Research and Evaluation}, vol.~10, pp. 1--9, 01
  2005.

\bibitem{rubinandthayer82}
D.~Rubin and D.~Thayer, ``Em algorithms for ml factor analysis,''
  \emph{Psychometrika}, vol.~47, pp. 69--76, 02 1982.

\bibitem{hendersonandsearle81}
\BIBentryALTinterwordspacing
H.~V. Henderson and S.~R. Searle, ``On deriving the inverse of a sum of
  matrices,'' \emph{SIAM Review}, vol.~23, no.~1, pp. 53--60, 1981. [Online].
  Available: \url{http://www.jstor.org/stable/2029838}
\BIBentrySTDinterwordspacing

\bibitem{byrdetal95}
R.~H. Byrd, J.~N. P.~Lu, and C.~Zhu, ``A limited memory algorithm for bound
  constrained optimization,'' \emph{SIAM Journal on Scientific Computing},
  vol.~16, pp. 1190--1208, 1995.

\bibitem{piessensetal83}
R.~Piessens, E.~de~Doncker-Kapenga, C.~W. {\"U}berhuber, and D.~K. Kahaner,
  \emph{Quadpack: a subroutine package for automatic integration}.\hskip 1em
  plus 0.5em minus 0.4em\relax Springer-Verlag, 1983, vol.~1.

\bibitem{sorensen92}
D.~C. Sorensen, ``Implicit application of polynomial filters in a $k$-step
  arnoldi method,'' \emph{{SIAM} Journal on Matrix Analysis and Applications},
  vol.~13, no.~1, pp. 357--385, 1992.

\bibitem{duttaandmondal15}
S.~Dutta and D.~Mondal, ``An h-likelihood method for spatial mixed linear model
  based on intrinsic autoregressions,'' \emph{Journal of the Royal Statistical
  Society: Series B (Statistical Methodology)}, vol.~77, pp. 699--726, 09 2015.

\bibitem{wilkinson1958}
J.~H. Wilkinson, ``The calculation of the eigenvectors of codiagonal
  matrices,'' \emph{Computer Journal}, vol.~1, pp. 90--96, 02 1958.

\bibitem{maitra09}
R.~Maitra, ``Initializing partition-optimization algorithms,'' \emph{IEEE/ACM
  Transactions on Computational Biology and Bioinformatics}, vol.~6, pp.
  144--157, 2009.

\bibitem{maitra13}
------, ``On the expectation-maximization algorithm for {R}ice-{R}ayleigh
  mixtures with application to estimating the noise parameter in magnitude {MR}
  datasets,'' \emph{Sankhy\=a: The Indian Journal of Statistics, Series {B}},
  vol.~75, no.~2, p. 293–318, 2013.

\bibitem{schwarz78}
G.~Schwarz, ``Estimating the dimensions of a model,'' \emph{Annals of
  Statistics}, vol.~6, pp. 461--464, 1978.

\bibitem{argelaguet2018multi}
R.~Argelaguet, B.~Velten, D.~Arnol, S.~Dietrich, T.~Zenz, J.~C. Marioni,
  F.~Buettner, W.~Huber, and O.~Stegle, ``Multi-omics factor analysis--a
  framework for unsupervised integration of multi-omics data sets,''
  \emph{Molecular Systems Biology}, vol.~14, no.~6, p. e8124, 2018.

\bibitem{meng2016mocluster}
C.~Meng, D.~Helm, M.~Frejno, and B.~Kuster, ``{moCluster}: identifying joint
  patterns across multiple omics data sets,'' \emph{Journal of Proteome
  Research}, vol.~15, no.~3, pp. 755--765, 2016.

\bibitem{mo2018fully}
Q.~Mo, R.~Shen, C.~Guo, M.~Vannucci, K.~S. Chan, and S.~G. Hilsenbeck, ``A
  fully bayesian latent variable model for integrative clustering analysis of
  multi-type omics data,'' \emph{Biostatistics}, vol.~19, no.~1, pp. 71--86,
  2018.

\bibitem{shen2009integrative}
R.~Shen, A.~B. Olshen, and M.~Ladanyi, ``Integrative clustering of multiple
  genomic data types using a joint latent variable model with application to
  breast and lung cancer subtype analysis,'' \emph{Bioinformatics}, vol.~25,
  no.~22, pp. 2906--2912, 2009.

\bibitem{grun2015single}
D.~Gr{\"u}n, A.~Lyubimova, L.~Kester, K.~Wiebrands, O.~Basak, N.~Sasaki,
  H.~Clevers, and A.~Van~Oudenaarden, ``Single-cell messenger {RNA} sequencing
  reveals rare intestinal cell types,'' \emph{Nature}, vol. 525, no. 7568, pp.
  251--255, 2015.

\bibitem{saelens2019comparison}
W.~Saelens, R.~Cannoodt, H.~Todorov, and Y.~Saeys, ``A comparison of
  single-cell trajectory inference methods,'' \emph{Nature Biotechnology},
  vol.~37, no.~5, pp. 547--554, 2019.

\bibitem{distefano2009understanding}
C.~DiStefano, M.~Zhu, and D.~Mindrila, ``Understanding and using factor scores:
  Considerations for the applied researcher,'' \emph{Practical Assessment,
  Research, and Evaluation}, vol.~14, no.~1, p.~20, 2009.

\bibitem{patterson2006population}
N.~Patterson, A.~L. Price, and D.~Reich, ``Population structure and
  eigenanalysis,'' \emph{PLoS Genetics}, vol.~2, no.~12, p. e190, 2006.

\bibitem{price2006principal}
A.~L. Price, N.~J. Patterson, R.~M. Plenge, M.~E. Weinblatt, N.~A. Shadick, and
  D.~Reich, ``Principal components analysis corrects for stratification in
  genome-wide association studies,'' \emph{Nature Genetics}, vol.~38, no.~8,
  pp. 904--909, 2006.

\bibitem{novembre2008interpreting}
J.~Novembre and M.~Stephens, ``Interpreting principal component analyses of
  spatial population genetic variation,'' \emph{Nature Genetics}, vol.~40,
  no.~5, pp. 646--649, 2008.

\bibitem{thurstone1935vectors}
L.~L. Thurstone, \emph{The vectors of mind: Multiple-factor analysis for the
  isolation of primary traits.}\hskip 1em plus 0.5em minus 0.4em\relax
  University of Chicago Press, 1935.

\bibitem{louis1982}
T.~A. Louis, ``Finding the observed information matrix when using the {EM}
  algorithm,'' \emph{Journal of the Royal Statistical Society, Series B
  (Methodological)}, vol.~44, no.~2, pp. 226--233, 1982.

\bibitem{barnard63}
G.~A. Barnard, ``Discussion of ``the spectral analysis of point processes" by
  m. s. bartlett,'' \emph{Journal of the Royal Statistical Society: Series B},
  vol.~25, pp. 264--294, 1963.

\bibitem{besagandclifford91}
J.~Besag and P.~Clifford, ``Sequential {M}onte {C}arlo p-values.''
  \emph{Biometrika}, vol.~78, pp. 301--304, 1991.

\bibitem{anandarajanetal18}
M.~Anandarajan, C.~Hill, and T.~Nolan, \emph{Practical text analytics:
  Maximizing the value of text data}, 1st~ed.\hskip 1em plus 0.5em minus
  0.4em\relax Springer Publishing Company, Incorporated, 2018.

\bibitem{lebanon06}
G.~Lebanon, ``Metric learning for text documents,'' \emph{{IEEE} Transactions
  on Pattern Analysis and Machine Intelligence}, vol.~28, pp. 497--508, 05
  2006.

\bibitem{tanandporzecanski18a}
G.~Tan and K.~Porzecanski, ``Wall street rule for the \#metoo era: Avoid women
  at all cost,'' \emph{Bloomberg}, 12 2018.

\bibitem{tanandporzecanski18b}
------, ``{NYC} officials blast {W}all {S}treet's ice-out of women in wake of
  \#metoo,'' \emph{Bloomberg}, 12 2018.

\bibitem{nebeletal14}
M.~B. Nebel, S.~Joel, J.~Muschelli, A.~Barber, B.~Caffo, J.~Pekar, and
  S.~Mostofsky, ``Disruption of functional organization within the primary
  motor cortex in children with autism,'' \emph{Human Brain Mapping}, vol.~35,
  pp. 567--580, 02 2014.

\bibitem{cox96}
R.~W. Cox, ``Afni: software for analysis and visualization of functional
  magnetic resonance neuroimages,'' \emph{Computers and Biomedical research},
  vol.~29, no.~3, pp. 162--173, 1996.

\bibitem{cox12}
------, ``{AFNI}: What a long strange trip it has been,'' \emph{NeuroImage},
  vol.~62, pp. 743--747, 2012.

\bibitem{chenandwang19}
C.-J. Chen and J.-L. Wang, ``A new approach for functional connectivity via
  alignment of blood oxygen level-dependent signals,'' \emph{Brain
  Connectivity}, vol.~9, no.~6, pp. 464--474, 2019.

\bibitem{sukel19}
\BIBentryALTinterwordspacing
K.~Sukel, \emph{Neuroanatomy: The Basics}, 2019. [Online]. Available:
  \url{https://www.dana.org/article/neuroanatomy-the-basics/}
\BIBentrySTDinterwordspacing

\bibitem{weinsteinetal13}
J.~N. Weinstein, E.~A. Collisson, G.~B. Mills, K.~R.~M. Shaw, B.~A. Ozenberger,
  K.~Ellrott, I.~Shmulevich, C.~Sander, and J.~M. Stuart, ``The cancer genome
  atlas pan-cancer analysis project,'' \emph{Nature Genetics}, vol.~45, no.
  1010, p. 1113–1120, 2013.

\bibitem{ramos20}
M.~Ramos, \emph{curatedTCGAData: Curated Data From The Cancer Genome Atlas
  (TCGA) as MultiAssayExperiment Objects}, 2020, r package version 1.10.0.

\bibitem{metzker10}
M.~L. Metzker, ``Sequencing technologies - the next generation.'' \emph{Nature
  Review Genetics}, vol.~11, p. 31–46, 2010.

\bibitem{wang2008rnaseq}
Z.~Wang, M.~Gerstein, and M.~Snyder, ``Rna-seq: A revolutionary tool for
  transcriptomics,'' \emph{Nature reviews. Genetics}, vol.~10, pp. 57--63, 12
  2008.

\bibitem{cox1972}
D.~R. Cox, ``Regression models and life-tables,'' \emph{Journal of the Royal
  Statistical Society, Series B (Methodological)}, vol.~34, no.~2, pp.
  187--220, 1972.

\bibitem{kaplanmeier58}
E.~L. Kaplan and P.~Meier, ``Nonparametric estimation from incomplete
  observations,'' \emph{Journal of the American Statistical Association},
  vol.~53, no. 282, pp. 457--481, 1958.

\bibitem{Mantel1966}
N.~Mantel, ``Evaluation of survival data and two new rank order statistics
  arising in its consideration,'' \emph{Cancer Chemotherapy Reports}, vol.~50,
  no.~3, pp. 163--170, 1966.

\bibitem{subramanian2005}
A.~Subramanian, P.~Tamayo, V.~K. Mootha, S.~Mukherjee, B.~L. Ebert, M.~A.
  Gillette, A.~Paulovich, S.~L. Pomeroy, and T.~R. Golub, ``Gene set enrichment
  analysis: a knowledge-based approach for interpreting genome-wide expression
  profiles.'' \emph{Proceedings of the National Academy of Sciences}, vol. 102,
  pp. 15\,545--15\,550, 10 2005.

\bibitem{Liberzon2015}
A.~Liberzon, C.~Birger, H.~Thorvaldsd{\'o}ttir, M.~Ghandi, J.~P. Mesirov, and
  P.~Tamayo, ``The {Molecular} {Signatures} {Database} ({MSigDB}) hallmark gene
  set collection,'' \emph{Cell Systems}, vol.~1, no.~6, pp. 417--425, 12 2015.

\bibitem{Gatti-Mays2019}
M.~E. Gatti-Mays, J.~M. Balko, S.~R. Gameiro, H.~D. Bear, S.~Prabhakaran,
  J.~Fukui, M.~L. Disis, R.~Nanda, J.~L. Gulley, K.~Kalinsky, H.~A. Sater,
  J.~A. Sparano, D.~Cescon, D.~B. Page, H.~McArthur, S.~Adams, and E.~A.
  Mittendorf, ``If we build it they will come: Targeting the immune response to
  breast cancer,'' \emph{npj Breast Cancer}, vol.~5, no.~1, pp. 1--13, 10 2019.

\bibitem{Heimes2020}
A.-S. Heimes, F.~H{\"a}rtner, K.~Almstedt, S.~Krajnak, A.~Lebrecht, M.~J.
  Battista, K.~Edlund, W.~Brenner, A.~Hasenburg, U.~Sahin, M.~Gehrmann, J.~G.
  Hengstler, and M.~Schmidt, ``Prognostic significance of interferon-$\gamma$
  and its signaling pathway in early breast cancer depends on the molecular
  subtypes,'' \emph{International Journal of Molecular Sciences}, vol.~21, no.
  1919, p. 7178, 1 2020.

\bibitem{Cervantes-Badillo2020}
M.~G. Cervantes-Badillo, A.~Paredes-Villa, V.~G{\'o}mez-Romero,
  R.~Cervantes-Rold{\'a}n, L.~E. Arias-Romero, O.~Villamar-Cruz,
  M.~Gonz{\'a}lez-Montiel, T.~Barrios-Garc{\'i}a, A.~J. Cabrera-Quintero,
  G.~Rodr{\'i}guez-G{\'o}mez, L.~Cancino-Villeda, A.~Zentella-Dehesa, and
  A.~Le{'o}n-Del-Rio, ``{IFI27/ISG12} downregulates estrogen receptor $\alpha$
  transactivation by facilitating its interaction with {CRM1/XPO1} in breast
  cancer cells,'' \emph{Frontiers in Endocrinology}, vol.~11, p. 568375, 2020.

\bibitem{Bektas2008}
N.~Bektas, E.~Noetzel, J.~Veeck, M.~F. Press, G.~Kristiansen, A.~Naami,
  A.~Hartmann, A.~Dimmler, M.~W. Beckmann, R.~Kn{\"u}chel, P.~A. Fasching, and
  E.~Dahl, ``The ubiquitin-like molecule interferon-stimulated gene 15
  ({ISG15}) is a potential prognostic marker in human breast cancer,''
  \emph{Breast Cancer Research}, vol.~10, no.~4, p. R58, Jul 2008.

\bibitem{Kariri2021}
Y.~A. Kariri, M.~Alsaleem, C.~Joseph, S.~Alsaeed, A.~Aljohani, S.~Shiino, O.~J.
  Mohammed, M.~S. Toss, A.~R. Green, and E.~A. Rakha, ``The prognostic
  significance of interferon-stimulated gene 15 ({ISG15}) in invasive breast
  cancer,'' \emph{Breast Cancer Research and Treatment}, vol. 185, no.~2, pp.
  293--305, 1 2021.

\bibitem{Rivera-Serrano2020}
E.~E. Rivera-Serrano, A.~S. Gizzi, J.~J. Arnold, T.~L. Grove, S.~C. Almo, and
  C.~E. Cameron, ``Viperin reveals its true function,'' \emph{Annual Review of
  Virology}, vol.~7, no.~1, pp. 421--446, 9 2020.

\bibitem{Tang2020}
J.~Tang, Q.~Yang, Q.~Cui, D.~Zhang, D.~Kong, X.~Liao, J.~Ren, Y.~Gong, and
  G.~Wu, ``Weighted gene correlation network analysis identifies {RSAD2},
  {HERC5}, and {CCL8} as prognostic candidates for breast cancer,''
  \emph{Journal of Cellular Physiology}, vol. 235, no.~1, pp. 394--407, 2020.

\bibitem{Chen2021}
G.~Chen, M.~Yu, J.~Cao, H.~Zhao, Y.~Dai, Y.~Cong, and G.~Qiao, ``Identification
  of candidate biomarkers correlated with poor prognosis of breast cancer based
  on bioinformatics analysis,'' \emph{Bioengineered}, vol.~12, no.~1, pp.
  5149--5161, 1 2021.

\bibitem{Sun2019}
C.~C. Sun, S.~J. Li, W.~Hu, J.~Zhang, Q.~Zhou, C.~Liu, L.~L. Li, Y.~Y.
  Songyang, F.~Zhang, Z.~L. Chen, G.~Li, Z.~Y. Bi, F.~Y. Gong, T.~Bo, Z.~P.
  Yuan, W.~D. Hu, B.~T. Zhan, Q.~Zhang, Q.~Q. He, and D.~J. Li, ``Comprehensive
  analysis of the expression and prognosis for {E2Fs} in human breast cancer,''
  \emph{Molecular Therapy}, vol.~27, no.~6, pp. 1153--1165, 6 2019.

\bibitem{Oshi2020}
M.~Oshi, H.~Takahashi, Y.~Tokumaru, Y.~Li, O.~M. Rashid, R.~Matsuyama, I.~Endo,
  and K.~Takabe, ``{G2M} cell cycle pathway score as a prognostic biomarker of
  metastasis in estrogen receptor ({ER})-positive breast cancer,''
  \emph{International Journal of Molecular Sciences}, vol.~21, no.~8, p. 2921,
  2020.

\bibitem{Wolter2017}
P.~Wolter, S.~Hanselmann, G.~Pattschull, E.~Schruf, and S.~Gaubatz, ``Central
  spindle proteins and mitotic kinesins are direct transcriptional targets of
  {MuvB}, {B-MYB} and {FOXM1} in breast cancer cell lines and are potential
  targets for therapy,'' \emph{Oncotarget}, vol.~8, no.~7, pp.
  11\,160--11\,172, 1 2017.

\bibitem{Ogino2020}
M.~Ogino, T.~Fujii, Y.~Nakazawa, T.~Higuchi, Y.~Koibuchi, T.~Oyama,
  J.~Horiguchi, and K.~Shirabe, ``Implications of topoisomerase ({TOP1} and
  {TOP2$\alpha$}) expression in patients with breast cancer,'' \emph{In Vivo},
  vol.~34, no.~6, pp. 3483--3487, 11 2020.

\bibitem{Li2021}
J.~Li, P.~Sun, T.~Huang, S.~He, L.~Li, and G.~Xue, ``Individualized
  chemotherapy guided by the expression of {ERCC1}, {RRM1}, {TUBB3}, {TYMS} and
  {TOP2A} genes versus classic chemotherapy in the treatment of breast cancer:
  A comparative effectiveness study,'' \emph{Oncology Letters}, vol.~21, no.~1,
  p.~21, 1 2021.

\bibitem{Seto2012}
B.~Seto, ``Rapamycin and {mTOR}: a serendipitous discovery and implications for
  breast cancer,'' \emph{Clinical and Translational Medicine}, vol.~1, no.~1,
  p. e29, 2012.

\bibitem{Qiu2018}
Z.~Qiu, Y.~Li, B.~Zeng, X.~Guan, and H.~Li, ``Downregulated {CDKN1C}/{p57kip2}
  drives tumorigenesis and associates with poor overall survival in breast
  cancer,'' \emph{Biochemical and Biophysical Research Communications}, vol.
  497, no.~1, pp. 187--193, 2 2018.

\bibitem{Takeshita2019}
T.~Takeshita, M.~Asaoka, E.~Katsuta, S.~J. Photiadis, S.~Narayanan, L.~Yan, and
  K.~Takabe, ``High expression of polo-like kinase 1 is associated with {TP53}
  inactivation, {DNA} repair deficiency, and worse prognosis in {ER} positive
  {Her2} negative breast cancer,'' \emph{American Journal of Translational
  Research}, vol.~11, no.~10, pp. 6507--6521, 10 2019.

\bibitem{Wang2020}
M.~Wang, M.~Dai, Y.-s. Wu, Z.~Yi, Y.~Li, and G.~Ren, ``Immunoglobulin
  superfamily member 10 is a novel prognostic biomarker for breast cancer,''
  \emph{PeerJ}, vol.~8, p. e10128, 10 2020.

\bibitem{Cai2021}
M.~Cai, H.~Li, R.~Chen, and X.~Zhou, ``{MRPL13} promotes tumor cell
  proliferation, migration and {EMT} process in breast cancer through the
  {PI3K}-{AKT}-{mTOR} pathway,'' \emph{Cancer Management and Research},
  vol.~13, pp. 2009--2024, 2 2021.

\bibitem{Brown2016}
D.~N. Brown, I.~Caffa, G.~Cirmena, D.~Piras, A.~Garuti, M.~Gallo, S.~Alberti,
  A.~Nencioni, A.~Ballestrero, and G.~Zoppoli, ``Squalene epoxidase is a bona
  fide oncogene by amplification with clinical relevance in breast cancer,''
  \emph{Scientific Reports}, vol.~6, no.~1, p. 19435, 1 2016.

\bibitem{Kim2021b}
N.~I. Kim, M.~H. Park, S.-S. Kweon, N.~Cho, and J.~S. Lee, ``Squalene epoxidase
  expression is associated with breast tumor progression and with a poor
  prognosis in breast cancer,'' \emph{Oncology Letters}, vol.~21, no.~4, p.
  259, 4 2021.

\bibitem{R}
\BIBentryALTinterwordspacing
{R Development Core Team}, \emph{R: A Language and Environment for Statistical
  Computing}, R Foundation for Statistical Computing, Vienna, Austria, 2018.
  [Online]. Available: \url{http://www.R-project.org}
\BIBentrySTDinterwordspacing

\bibitem{downs03}
T.~D. Downs, ``Spherical regression,'' \emph{Biometrika}, vol.~90, no.~3, pp.
  655--668, 2003.

\bibitem{Anderson2003}
T.~W. Anderson, \emph{An Introduction to multivariate statistical analysis},
  ser. Wiley Series in Probability and Statistics.\hskip 1em plus 0.5em minus
  0.4em\relax Wiley, 2003.

\end{thebibliography}

\setcounter{equation}{0}
\setcounter{figure}{0}
\setcounter{table}{0}
\setcounter{page}{1}
\setcounter{section}{0}

\renewcommand{\thesection}{S\arabic{section}}
\renewcommand{\thesubsection}{\thesection.\arabic{subsection}}
\renewcommand{\theequation}{S\arabic{equation}}
\renewcommand{\thefigure}{S\arabic{figure}}
\renewcommand{\thetable}{S\arabic{table}}
\renewcommand{\thetheorem}{S\arabic{theorem}}

%


\section{Supplement to Section~\ref{sec:meth-ch3}}
\label{sec:supp-ch3}

\subsection{Derivation of the projected normal density function}
\label{sec:supp-pndensity-ch3}
\begin{theorem}\label{theorem:pndensity}
Suppose $\bY\sim\normal_p(\mu,\Sigma).$ with mean vector $\mu$ and covariance matrix $\Sigma$. Then, $\bX = \bY/\|\bY\|$ follows a projected normal distribution $\sPN_{p}(\mu,\Sigma)$ and the density function of $\bX$ at $\mathbf{x} = \mathbf{y}/\|\mathbf{y}\|$ is given by,
\begin{equation}\label{supp-eqn:pdfPN}
 f(\mathbf{x};\mu,\Sigma) =
 (2\pi)^{-\frac{p}{2}}|\Sigma|^{-\frac{1}{2}}\exp{\left\{-\displaystyle{\frac{\mu^\top\Sigma^{-1}\mu}{2}} + \displaystyle{\frac{m^2}{2v}}\right\}}
 {\int_{0}^{\infty}{R^{p-1}\exp{\left\{-\displaystyle{\frac{1}{2v}}\Big(R-m\Big)^2\right\}}}\, dR},
\end{equation}
where $R = \|\mathbf{y}\|$, $m = {\mathbf{x}}^\top\Sigma^{-1}\mu/{\mathbf{x}}^\top\Sigma^{-1}{\mathbf{x}}$, $v = 1/{\mathbf{x}}^\top\Sigma^{-1}{\mathbf{x}}$.
\end{theorem}

\begin{proof}
Let $\mathbf{y} = (y_1,y_2,\ldots,y_p)$ with $R = \|\mathbf{y}\|$, and $\mathbf{x} = (x_1,x_2,\ldots,x_p) = \mathbf{y}/\|\mathbf{y}\|$ so that $\|\mathbf{x}\|=1$. Define 
the representation of $\mathbf{x}$ in the spherical coordinate system as ${\bm\theta} = (\theta_1,\theta_2,\ldots,\theta_{p-1})$ so that
\begin{equation}\label{supp-eqn:coordinate-x}
   x_1 = \cos{\theta_1}\text{; }\quad  x_j = \prod_{k=1}^{j-1}{\sin{\theta_{k}}}\cos{\theta_j}, \text{ for }j = 2,3,\ldots,p-1\text{; }\quad x_p = \prod_{k=1}^{p-1}{\sin{\theta_{k}}},
\end{equation}
with the Jacobian of transformation as $\mathbf{J_x} = \prod_{j=1}^{p-2}{\sin^{(p-1-j)}{\theta_{j}}}.$ Then, we denote the transformation from $R,{\bm\theta}$ to $\mathbf{y}$ as $\mathbf{y} = \mathbf{y}(R,{\bm\theta})$ which is given by,
\begin{equation}\label{supp-eqn:coordinate-y}
   y_1 = R\cos{\theta_1}\text{; }\quad  y_j = R\prod_{k=1}^{j-1}{\sin{\theta_{k}}}\cos{\theta_j}, \text{ for }j = 2,3,\ldots,p-1\text{; }\quad y_p = R\prod_{k=1}^{p-1}{\sin{\theta_{k}}},
\end{equation}
with the Jacobian of transformation as $\mathbf{J_y} =  R^{p-1}\prod_{j=1}^{p-2}{\sin^{(p-1-j)}{\theta_{j}}}.$
Now, by the change of variables, the density function of ${\bm\theta}$ is given by,
\begin{equation}\label{supp-eqn:pdfPN-theta}
\begin{split}
 f({\bm\theta};\mu,\Sigma) = &
 \int_{0}^{\infty}{\mathbf{J_y}(2\pi)^{-\frac{p}{2}}|\Sigma|^{-\frac{1}{2}}\exp{\Big\{-\displaystyle{\frac{1}{2}}(\mathbf{y}(R,{\bm\theta})-\mu)^\top\Sigma^{-1}(\mathbf{y}(R,{\bm\theta})-\mu)\Big\}} }\, dR.
\end{split}
\end{equation}
Then, by applying the change of variables again, the density function of $\mathbf{x}$ is given by,
\begin{equation}\label{supp-eqn:pdfPN-x}
\begin{split}
 f(\mathbf{x};\mu,\Sigma) 
 = &
 \int_{0}^{\infty}{\mathbf{J_y}/\mathbf{J_x}(2\pi)^{-\frac{p}{2}}|\Sigma|^{-\frac{1}{2}}\exp{\Big\{-\displaystyle{\frac{1}{2}}(R\mathbf{x}-\mu)^\top\Sigma^{-1}(R\mathbf{x}-\mu)\Big\}} }\, dR\\
  =&
 (2\pi)^{-\frac{p}{2}}|\Sigma|^{-\frac{1}{2}}\exp{\left\{-\displaystyle{\frac{\mu^\top\Sigma^{-1}\mu}{2}} + \displaystyle{\frac{m^2}{2v}}\right\}}
 {\int_{0}^{\infty}{R^{p-1}\exp{\Big\{-\displaystyle{\frac{1}{2v}}\Big(R-m\Big)^2\Big\}}}\, dR},
\end{split}
\end{equation}
where $m = {\mathbf{x}}^\top\Sigma^{-1}\mu/{\mathbf{x}}^\top\Sigma^{-1}{\mathbf{x}}$, $v = 1/{\mathbf{x}}^\top\Sigma^{-1}{\mathbf{x}}$.

\end{proof}

\subsection{Numerical computation of the MLE of \texorpdfstring{$\mu$}{\muold}}
\label{sec:supp-proposition1&2}
\begin{theorem}\label{theorem:maxmu}
Let $\bf B$ be a $p\times p$ positive definite symmetric matrix, ${\bf b} \in {\bbR}^p$. Then, the unique global minimizer ${\bf x}_0$ that
\begin{equation}\label{eqn:maxmu-minx}
    \text{minimize }{\Tr{\bf B}^{-1}({\bf b-x})({\bf b-x})^\top} \text{ subject to } {\bf x}^\top{\bf x}=1
\end{equation}
is given by
${\bf x}_0 = ({\bf I}+\lambda_0{\bf B})^{-1}{\bf b}$, where $\lambda_0$ lies between
\begin{itemize}
    \item[(1)] $[0,{\bf b}^\top{\bf B}^{-1}{\bf b}/2]$, if ${\bf b}^\top{\bf b}-1>0$,
    \item[(2)] $\lambda = 0$, if ${\bf b}^\top{\bf b}-1=0$,
    \item[(3)] $[(|{\bf u}_m^\top{\bf b}|/2-1)/e_m,0]$, if ${\bf b}^\top{\bf b}-1<0$.
\end{itemize}
Where $e_m$, ${\bf u}_m$ are the largest eigenvalue and corresponding eigenvector of $\bf B$.
\end{theorem}

\begin{proof}
The Lagrangian associated with \eqref{eqn:maxmu-minx} is given by,
\begin{equation}\label{eqn:lagrange}
    L({\bf x},\lambda) = \Tr{\bf B}^{-1}({\bf b-x})({\bf b-x})^\top + \lambda({\bf x}^\top{\bf x}-1).
\end{equation}
with a Lagrange multiplier $\lambda$. Setting $\nabla L = 0$ yields ${\bf x}_0 = ({\bf I}+\lambda{\bf B})^{-1}{\bf b}\text{ and } {\bf b}^\top(\mathbf{I}+\lambda{\bf B})^{-2}{\bf b}=1$.

Let $f(\lambda) = {\bf b}^\top(\mathbf{I}+\lambda{\bf B})^{-2}{\bf b}-1$, ${\bf B} = \mathbf{UMU}^\top$, ${\bf M} = \mathrm{diag}\{e_j\}, j = 1,2,\ldots,p$ and ${\bf U^\top b} =  (c_1,c_2,\ldots,c_p)^\top = \bf{c}$, then, 
\[f(\lambda) = {\bf c}^\top(\mathbf{I}+\lambda{\bf M})^{-2}{\bf c}-1 = \sum_{j=1}^{p}{\displaystyle{\frac{c_j^2}{(1+\lambda e_j)^2}}}-1\]
and
\[f'(\lambda) = -2\sum_{j=1}^{p}{\displaystyle{\frac{c_j^2e_j}{(1+\lambda e_j)^3}}}\text{ , }f''(\lambda) = 6\sum_{j=1}^{p}{\displaystyle{\frac{c_j^2e_j^2}{(1+\lambda e_j)^4}}}.\]\par
For $\lambda\geq0$, $f(\lambda)$ is a continuous function with $f'(\lambda)<0$ and
\[f(\lambda) =  \sum_{j=1}^{p}{\displaystyle{\frac{c_j^2}{(1+\lambda e_j)^2}}}-1\leq\sum_{j=1}^{p}{\displaystyle{\frac{c_j^2}{2\lambda e_j}}}-1 = \displaystyle{\frac{1}{2\lambda}}{\bf b}^\top{\bf B}^{-1}{\bf b}-1.\]\par
So,
\[\lambda>\displaystyle{\frac{\bf{b}^\top{\bf B}^{-1}\bf{b}}{2}}\Rightarrow f(\lambda)<0\].\par
If $f(0) = {\bf b}^\top{\bf b}-1>0$, we can then find the root of $f(\lambda)=0$ in $[0,{\bf b}^\top{\bf B}^{-1}{\bf b}/2].$\par
\bigskip\medskip

If ${\bf b}^\top{\bf b}-1<0$, no positive root exists. Then let $\lambda<0$, we have
\[f(\lambda) =  \sum_{j=1}^{p}{\displaystyle{\frac{c_j^2}{(1+\lambda e_j)^2}}}-1\geq \displaystyle{\frac{c^2_j}{(1+\lambda e_j)^2}}-1.\]\par
For each $j = 1,2,\ldots,p$. So,
\[\displaystyle{\frac{-\sqrt{c^2_j}-1}{e_j}}<\lambda<\displaystyle{\frac{\sqrt{c^2_j}-1}{e_j}}\Rightarrow f(\lambda)>0.\]

According to the Hessian matrix ${\bf H}_{{\bf x}} = {\bf B}^{-1} + \lambda{\bf I}$, the unique global minimizer is obtained when $\lambda>-\displaystyle{\frac{1}{e_{m}}}$, where $e_{m} = \max\{e_1,e_2,\ldots,e_p\}$. Let ${\bf u}_m$ denote the corresponding eigenvector and $c_m  = {\bf u}_m^\top{\bf b}$ in $\bf c$, then $-\displaystyle{\frac{1}{e_j}}\leq-\displaystyle{\frac{1}{e_{m}}}<\displaystyle{\frac{0.5\sqrt{c^2_m}-1}{e_{m}}}$, for $j = 1,2,\cdots,p$, and when $\lambda\geq\displaystyle{\frac{0.5\sqrt{c^2_m}-1}{e_{m}}}$, $f(\lambda)$ is continuous. Hence, we can find an negative root of $f(\lambda)=0$ in $\left[\displaystyle{\frac{0.5\sqrt{c^2_m}-1}{e_{m}}},0\right]$.
\end{proof}

\subsection{Proof of Proposition \ref{proposition:profileout-ch3}} \label{sec:supp-proposition3-ch3}
From Section \ref{sec:meth-profile-aecm}, for fixed $\mu_{t+1},$ $Q_P$ is equivalent to 
\begin{equation} \label{eq:exp-llk2-eq}
        Q_{P}(\mu_{t+1},\Lambda,\Psi;\mu_{t+1},\Lambda_t,\Psi_t) = c -\frac n2\{\log\det (\Lambda\Lambda^\top+\Psi) +\Tr (\Lambda\Lambda^\top+\Psi)^{-1}\widetilde\bS_t\},
\end{equation}
where $\widetilde\bS_t = \frac1n\sum_{i=1}^n\bbE_{R_i|\bX_i} \left\{(R_i\bX_i-\bmu_{t+1})(R_i\bX_i-\bmu_{t+1})^\top\right\}$.

From \eqref{eq:exp-llk2-eq}, the ML estimators of $\Lambda$ and $\Psi$ are obtained by solving the score equations
\begin{equation} \label{eq:system_eq-ch3}
\begin{split}
\begin{cases}
    & \Lambda(\mathbf{I}_q + \Lambda^\top\Psi^{-1}\Lambda) = 
    \widetilde\bS_t\Psi^{-1}\Lambda\\
    & \Psi = \mathrm{diag}(\widetilde\bS_t-\Lambda\Lambda^\top)
    \end{cases}
    \end{split}
\end{equation}
From $\Lambda(\mathbf{I}_q + \Lambda^\top\Psi^{-1}\Lambda) = 
   \widetilde\bS_t\Psi^{-1}\Lambda$, we have
    \begin{equation}\label{eq:profileout_Lambda-ch3}
       \Psi^{-1/2}\Lambda(\mathbf{I}_q + (\Psi^{-1/2}\Lambda)^\top\Psi^{-1/2}\Lambda) = \Psi^{-1/2}\widetilde\bS_t\Psi^{-1/2}\Psi^{-1/2}\Lambda.
    \end{equation}
Suppose that $\Psi^{-1/2}\widetilde\bS_t\Psi^{-1/2} = \mathbf{V}\mathbf{D}\mathbf{V}^\top$ and that the diagonal elements in $\mathbf{D}$ are in decreasing order with $\theta_1\geq\theta_2\geq,\cdots,\geq\theta_p$. Let $\mathbf{D} = \begin{bmatrix} 
\mathbf{D}_q & 0 \\
0 & \mathbf{D}_m 
\end{bmatrix}$ with $m=p-q$ and $\mathbf{D}_q$ containing the largest $q$ eigenvalues  $\theta_1\geq\theta_2\geq,\cdots,\geq\theta_q$. The corresponding $q$ eigenvectors form columns of the matrix $\mathbf{V}_q$ so that $\mathbf{V} = [\mathbf{V}_q,\mathbf{V}_m]$. Then, if $\mathbf{D}_q > \mathbf{I}_q$, \eqref{eq:profileout_Lambda-ch3} shows that
\begin{equation}
     \Lambda =  \Psi^{1/2}\mathbf{V}_q(\mathbf{D}_q-\mathbf{I}_q)^{1/2}.
\end{equation}
Hence, conditional on $\Psi$, $\Lambda$ is maximized at $\hat{\Lambda} = \Psi^{1/2}\mathbf{V}_q\bm{\Delta}$, where $\bm{\Delta}$ is a diagonal matrix with elements $\mathrm{max}(\theta_i-1,0)^{1/2}, i = 1,\cdots,q$.

From the construction of $\mathbf{V}_q$ and $\mathbf{V}_m$, we have $\mathbf{V}^\top_q\mathbf{V}_q = \mathbf{I}_q\text{, }
\mathbf{V}^\top_m\mathbf{V}_m = \mathbf{I}_m\text{, }
\mathbf{V}_q\mathbf{V}^\top_q + \mathbf{V}_m\mathbf{V}^\top_m = \mathbf{I}_p\text{, }\mathbf{V}^\top_q\mathbf{V}_m = \boldsymbol{0}$ and hence, $(\mathbf{V}_q\mathbf{D}_q\mathbf{V}^\top_q +  \mathbf{V}_m\mathbf{V}^\top_m)(\mathbf{V}_q\mathbf{D}^{-1}_q\mathbf{V}^\top_q +  \mathbf{V}_m\mathbf{V}^\top_m) = \mathbf{I}_p$.

Let $\mathbf{A} = \mathbf{V}_q\bm{\Delta}^2\mathbf{V}^\top_q$. Then $\mathbf{A}\mathbf{A} = \mathbf{V}_q\bm{\Delta}^4\mathbf{V}^\top_q$ and
\begin{equation}\label{eq:matrix_results1-ch3}
    \begin{split}
       |\mathbf{A}+\mathbf{I}_p| 
       = & |(\mathbf{A}+\mathbf{I}_p)\mathbf{A}|/|\mathbf{A}| \\
       = & |\mathbf{V}_q(\bm{\Delta}^4+\bm{\Delta}^2)\mathbf{V}^\top_q|/|\mathbf{V}_q\bm{\Delta}^2\mathbf{V}^\top_q| \\
       = & |\bm{\Delta}^2+\mathbf{I}_q| = \prod_{j=1}^{q}{\theta_j}
       \end{split}
\end{equation}
and
\begin{equation}\label{eq:matrix_results2-ch3}
      \begin{split}
      (\mathbf{A}+\mathbf{I}_p)^{-1} 
      = & (\mathbf{V}_q\bm{\Delta}^2\mathbf{V}^\top_q + \mathbf{V}_q\mathbf{V}^\top_q + \mathbf{V}_m\mathbf{V}^\top_m)^{-1}\\
      = & (\mathbf{V}_q(\bm{\Delta}^2+\mathbf{I}_q)\mathbf{V}^\top_q +  \mathbf{V}_m\mathbf{V}^\top_m)^{-1}\\
      = & (\mathbf{V}_q\mathbf{D}_q\mathbf{V}^\top_q +  \mathbf{V}_m\mathbf{V}^\top_m)^{-1}\\
      = & \mathbf{V}_q\mathbf{D}^{-1}_q\mathbf{V}^\top_q +  \mathbf{V}_m\mathbf{V}^\top_{m}.
    \end{split}
\end{equation}
Based on \eqref{eq:matrix_results1-ch3} and \eqref{eq:matrix_results2-ch3} and \eqref{eqn:profilelikelihood-ch3}, the profile log-likelihood is given by,
\begin{equation}\label{eq:proof_proposition1-ch3}
    \begin{split}
        Q_{p}(\Psi) 
        = c' & - \frac{n}{2}\log{|\hat{\Lambda}\hat{\Lambda}^\top + \Psi|} -\frac{n}{2}\Tr(\hat{\Lambda}\hat{\Lambda}^\top + \Psi)^{-1}\widetilde\bS_t\\
        = c' & - \frac{n}{2}\Big\{\log|\Psi^{1/2}(\mathbf{V}_q\bm{\Delta}^2\mathbf{V}^\top_q + \mathbf{I}_p)\Psi^{1/2}| + \Tr(\Psi^{1/2}(\mathbf{V}_q\bm{\Delta}^2\mathbf{V}^\top_q + \mathbf{I}_p)\Psi^{1/2})^{-1}\widetilde\bS_t\Big\}\\
        = c' & - \frac{n}{2}\Big\{ \log\det\Psi + \log{|\mathbf{V}_q\bm{\Delta}^2\mathbf{V}^\top_q + \mathbf{I}_p|} + \Tr(\mathbf{V}_q\mathbf{D}^{-1}_q\mathbf{V}^\top_q +  \mathbf{V}_m\mathbf{V}^\top_m)\Psi^{-1/2}\widetilde\bS_t\Psi^{-1/2}
        \Big\}\\
        = c' & - \frac{n}{2}\Big\{  \log\det\Psi + \sum_{j=1}^{q}{\log{\theta_j}} + \Tr\mathbf{D}^{-1}_q\mathbf{V}^\top_q\mathbf{V}\mathbf{D}\mathbf{V}^\top\mathbf{V}_q + \Tr\mathbf{V}^\top_m\mathbf{V}\mathbf{D}\mathbf{V}^\top\mathbf{V}_m\Big\}\\
        = c' & - \frac{n}{2}\Big\{  \log\det\Psi + \sum_{j=1}^{q}{\log{\theta_j}} + \Tr\mathbf{D}^{-1}_q\mathbf{D}_q + \Tr\mathbf{D}_m\Big\}\\
        = c' &- \frac{n}{2}\Big\{ \log\det\Psi + \sum_{j=1}^{q}{\log{\theta_j}} + q + \Tr\Psi^{-1}\widetilde\bS_t- \sum_{j=1}^{q}{\theta_j}\Big\},
    \end{split}
\end{equation}
where $c'$ is a constant that does not depend on $\Psi$.

\subsection{Recursive formulas}
\label{sec:supp-int}
Let $\mI_k = \int_{0}^{\infty}{R^{k}\exp\left\{-\displaystyle{\frac{1}{2v}}(R-m)^2\right\}\, dR}$ with $m\in \mathbb R$, $v>0$, $k\in\mathbb N$, then,
\begin{equation*}
    \begin{split}
 \mI_k &= \left.\frac{R^{k+1}}{k+1}\exp\left\{-\frac{1}{2v}(R-m)^2\right\}\right\vert^{\infty}_{0} -  \int_{0}^{\infty}{\frac{R^{k+1}}{k+1}\exp\left\{-\frac{1}{2v}(R-m)^2\right\}\left(-\frac{R}{v}+\frac{m}{v}\right)}\, dR\\
 &= \frac{1}{v(k+1)}\mI_{k+2} - \frac{m}{v(k+1)}\mI_{k+1}.
    \end{split}
\end{equation*}\par
That is, $\mI_{k+2} = m\mI_{k+1} + (k+1)v\mI_k$. Then, $\displaystyle{\frac{\mI_{k+2}}{\mI_{k+1}}} = m + (k+1)v
\displaystyle{\frac{\mI_{k}}{\mI_{k+1}}}
\text{ , }
\displaystyle{\frac{\mI_{k+2}}{\mI_k}} = (k+1)v + m\displaystyle{\frac{\mI_{k+1}}{\mI_k}}$.

\subsection{FADS-D: E- and CM-step computation}
\label{sec:supp-fads-d}
The conditional
expectations of $\bZ_i$, $\bZ_i\bZ_i^\top$ and $R_i, \bZ_i$ given $\bX$ are
\[\bbE_{\bZ_i|\bX_i}(\bZ_i) = \bbE_{R_i|\bX_i}\{\bbE_{\bZ_i|R_i,\bX_i}(\bZ_i)\},\quad 
\bbE_{\bZ_i|\bX_i}(\bZ_i\bZ_i^\top) = \bbE_{R_i|\bX_i}\{\bbE_{\bZ_i|R_i,\bX_i}(\bZ_i\bZ_i^\top)\},\]
\[\bbE_{\bZ_i,R_i|\bX_i}(R_i\bZ_i^\top) = \bbE_{R_i|\bX_i}\{R_i\bbE_{\bZ_i|R_i,\bX_i}(\bZ_i^\top)\}.
\]
The joint density function of $R_i,\bX_i,\bZ_i$ is given by,
\begin{equation} \label{eq:z-llk}
\begin{split}
        f_{r,x,z}(R_i,\bX_i,\bZ_i) 
        = & f_{y|z}({\bf Y}_i|\bZ_i)|{\bf J}_i|\cdot f_z(\bZ_i)\\
        \propto & \exp\{-\frac{1}{2}(R_i\bX_i-\mu-\Lambda\bZ_i)^\top\Psi^{-1}(R_i\bX_i-\mu-\Lambda\bZ_i)-\frac{1}{2}\bZ_i^\top\bZ_i\}.
\end{split}
\end{equation}
Then, by the conditional normal distribution, we have
\[\bbE_{\bZ_i|\bX_i}(\bZ_i) = \Lambda^\top\Sigma^{-1}\{\bbE_{R_i|\bX_i}(R_i)\bX_i-\mu\},
\]
\[\bbE_{\bZ_i,R_i|\bX_i}(R_i\bZ_i^\top) = 
\{\bbE_{R_i|\bX_i}(R_i^2)\bX_i^\top-\bbE_{R_i|\bX_i}(R_i)\mu^\top\}\Sigma^{-1}\Lambda,
\]
\[\bbE_{\bZ_i|\bX_i}(\bZ_i\bZ_i^\top) = 
({\bf I}+\Lambda^\top\Psi^{-1}\Lambda)^{-1} + \Lambda^\top\Sigma^{-1}\{\bbE_{R_i|\bX_i}(R_i^2)\bX_i\bX_i^\top
-\bbE_{R_i|\bX_i}(R_i)(\mu\bX_i^\top+\bX_i\mu^\top)+\mu\mu^\top\}\Sigma^{-1}\Lambda.
\]
Then, conditional on $\hat{\mu}_{t+1}$ and $\Psi_t$, the CM-step for $\Lambda_{t+1}$ is given by,
\begin{equation}
\begin{split}
\hat{\Lambda}_{t+1} = 
& \left\{\sum^{n}_{i=1}{\bX_i\bbE_{\bZ_i,R_i|\bX_i}(R_i\bZ_i^\top)}-\hat{\mu}_{t+1}\sum^{n}_{i=1}{\bbE_{\bZ_i|\bX_i}(\bZ_i^\top)}\right\}\left\{\sum^{n}_{i=1}{\bbE_{\bZ_i|\bX_i}(\bZ_i\bZ_i^\top)}\right\}^{-1}.
\end{split}
\end{equation}

Conditional on $\hat{\mu}_{t+1}$ and $\hat{\Lambda}_{t+1}$, the CM-step for $\Psi_{t+1}$ is given by,
\begin{equation} 
    \begin{split}
    \hat{\Psi}_{t+1}
 = & \frac{1}{n}\mathrm{diag} \left\{
 \sum^{n}_{i=1}{\bbE_{R_i|\bX_i}(R_i^2)\bX_i\bX_i^\top} - 2\hat{\mu}_{t+1}\sum^{n}_{i=1}{\bbE_{R_i|\bX_i}(R_i)\bX_i^\top} \right.\\ &\left.-2\hat{\Lambda}_{t+1}\sum^{n}_{i=1}{\bbE_{\bZ_i,R_i|\bX_i}(R_i\bZ_i)\bX_i^\top} + 2\hat{\mu}_{t+1}\sum^{n}_{i=1}{\bbE_{\bZ_i|\bX_i}(\bZ_i^\top)}\hat{\Lambda}_{t+1}^\top \right.\\
&+\left.\hat{\Lambda}_{t+1}\sum^{n}_{i=1}{\bbE_{\bZ_i|\bX_i}(\bZ_i\bZ_i^\top)}\hat{\Lambda}_{t+1}^\top + n\hat{\mu}_{t+1}\hat{\mu}_{t+1}^\top\right\}.
\end{split}
\end{equation}
The non-positive values in $\hat{\Psi}_{t+1}$ are truncated to a small positive number \citep{Anderson2003}.

%

\section{Supplement to Section~\ref{sec:sim-ch3}}
\subsection{Average iterations and convergence rates}
\label{sec:supp-sim-avgiter}
\begin{table}[H]
\caption{Average iterations (in converged cases for both FADS-P and FADS-D) and convergence rates (in \%) of FADS-P and FADS-D algorithms applied to randomly simulated datasets where true $(n,p,q) = (300,30,3)$.} 
\label{table:sim-c1q3-iter}
\centering
\resizebox{.6\textwidth}{!}{%
\begin{tabular}{cc|cccccc}
  \toprule 
  & & 1 & 2 & 3 & 4 & 5 & 6 \\
 \midrule 
 \midrule\multirow{2}{*}{FADS-P} & iteration & 109 &67  &48  &83 &107 &115\\ 
   & convergence rate & 100&100 &100 &100 &100 &100\\ 
 \midrule\multirow{2}{*}{FADS-D} & iteration & 257  &341  &595  &885 &1438 &2732\\ 
   & convergence rate &100 &100 &100 &92 &82 &70\\  
   \bottomrule 
\end{tabular}%
}
\end{table}
\begin{table}[H]
\caption{Average iterations (in converged cases for both FADS-P and FADS-D) and convergence rates (in \%) of FADS-P and FADS-D algorithms applied to randomly simulated datasets where true $(n,p,q) = (300,30,5)$.} 
\label{table:sim-c1q5-iter}
\centering
\resizebox{.9\textwidth}{!}{%
\begin{tabular}{cc|cccccccccc}
  \toprule 
  & & 1 & 2 & 3 & 4 & 5 & 6 & 7 & 8 & 9 & 10 \\
 \midrule 
 \midrule\multirow{2}{*}{FADS-P} & iteration & 127  &88  &70  &62  &33  &76  &89 &102 &111 &147\\
   & convergence rate &100 &100 &100 &100 &100 &100 &100 &100 &100 &100\\ 
 \midrule\multirow{2}{*}{FADS-D} & iteration &378  &467  &570  &820 &1304 &1952 &2312 &3225 &3296 &4091\\ 
   & convergence rate &91  &94  &98 &100 &100 &82  &55  &38  &28  &24\\  
   \bottomrule 
\end{tabular}%
}
\end{table}
\begin{table}[H]
\caption{Average iterations (in converged cases for both FADS-P and FADS-D) and convergence rates (in \%) of FADS-P and FADS-D algorithms applied to randomly simulated datasets where true $(n,p,q) = (1000,100,3)$.} 
\label{table:sim-c2q3-iter}
\centering
\resizebox{.65\textwidth}{!}{%
\begin{tabular}{cc|cccccc}
  \toprule 
  & & 1 & 2 & 3 & 4 & 5 & 6 \\
 \midrule 
 \midrule\multirow{2}{*}{FADS-P} & iteration & 683 &208 &123 &141 &275 &402\\ 
   & convergence rate &100 &100 &100 &100 &100 &100\\ 
 \midrule\multirow{2}{*}{FADS-D} & iteration &1411 &2075 &3156 &3120 &3864 &4042\\ 
   & convergence rate &100  &100  &100  &100  &100  &100\\  
   \bottomrule 
\end{tabular}%
}
\end{table}
\begin{table}[H]
\caption{Average iterations (in converged cases for both FADS-P and FADS-D) and convergence rates (in \%) of FADS-P and FADS-D algorithms applied to randomly simulated datasets where true $(n,p,q) = (1000,100,5)$.} 
\label{table:sim-c2q5-iter}
\centering
\resizebox{.9\textwidth}{!}{%
\begin{tabular}{cc|cccccccccc}
  \toprule 
  & & 1 & 2 & 3 & 4 & 5 & 6 & 7 & 8 & 9 & 10 \\
 \midrule 
 \midrule\multirow{2}{*}{FADS-P} & iteration & 789 &231 &275 &233 &141 &204 &307 &456 &475 &601\\ 
   & convergence rate &100 &100 &100 &100 &100 &100 &100 &100 &100 &100\\ 
 \midrule\multirow{2}{*}{FADS-D} & iteration & 2281 &2411 &3581 &4498 &5818 &6185 &6971 &7269 &7206 &7189\\ 
   & convergence rate &100 &100 &100 &100  &96  &79  &72  &65  &64  &59\\  
   \bottomrule 
\end{tabular}%
}
\end{table}
\subsection{Average CPU time}
\label{sec:supp-sim-avgtime-ch3}
\begin{table}[H]
\caption{Average CPU time (in seconds) of FADS-P and FADS-D algorithms applied to randomly simulated datasets (Converged cases).} 
\label{table:sim-time}
\centering
\resizebox{\textwidth}{!}{%
\begin{tabular}{cc|cccccccccc}
  \toprule 
  & & 1 & 2 & 3 & 4 & 5 & 6 & 7 & 8 & 9 & 10\\
 \midrule 
 \midrule\multirow{2}{*}{\bm{$(n,p,q)=(300,30,3)$}} 
 & FADS-P &7.570  &5.640  &3.451  &6.826  &9.241 &10.239& -& -& -& -\\
 & FADS-D &10.411  &13.965  &24.603  &36.619  &59.461 &113.415& -& -& -& -\\ 
  \midrule\multirow{2}{*}{\bm{$(n,p,q)=(300,30,5)$}} 
   & FADS-P &7.983  &6.248  &5.901  &4.594  &2.171  &5.698  &7.014  &8.718 &10.044 &13.119\\ 
  & FADS-D &15.229  &18.929  &23.348  &33.641  &54.090  &81.192  &97.207 &134.969 &138.628 &168.619\\  
  \midrule\multirow{2}{*}{\bm{$(n,p,q)=(1000,100,3)$}} 
   & FADS-P &134.633  &47.134  &26.230  &36.970  &76.216 &111.779 &- &- &- &-\\ 
  & FADS-D &219.185 &331.254 &503.038 &498.516 &620.895 &644.078 &- &- & -&-\\ 
  \midrule\multirow{2}{*}{\bm{$(n,p,q)=(1000,100,5)$}} 
   & FADS-P &145.243  &46.902  &60.470  &49.188  &27.735  &49.616  &78.007 &121.077 &131.134 &172.879\\ 
   & FADS-D &354.947  &381.501  &570.623  &710.967  &927.920 &1007.009 &1147.169 &1189.592 &1184.188 &1178.299\\
   \bottomrule 
\end{tabular}%
}
\end{table}

\subsection{Data-driven cases: Relative Frobenius distances of FADS-P estimates}
\label{sec:supp-sim-tcga-ch3}
\begin{figure}[H]
\centering
  \resizebox{.7\linewidth}{!}{
  \input{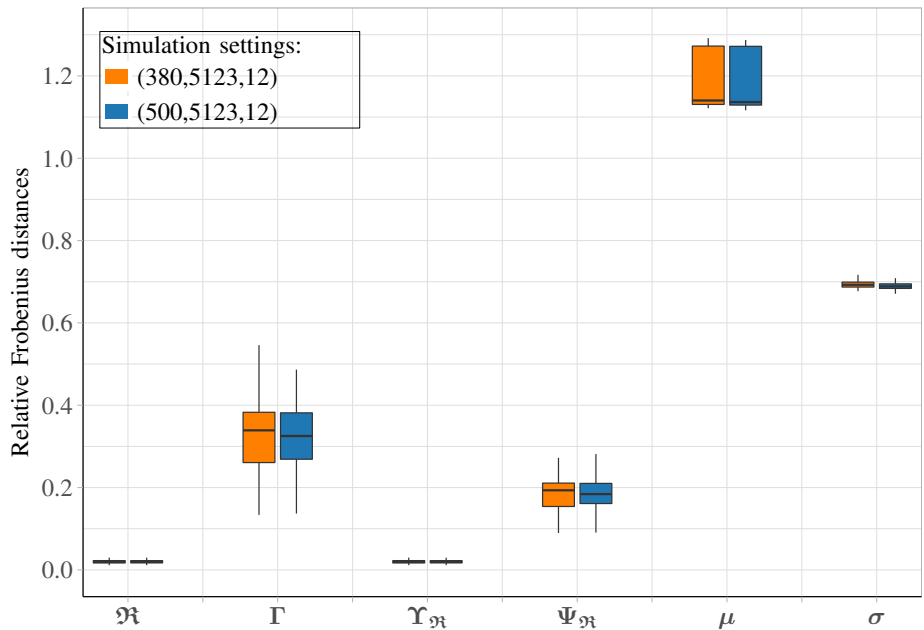}
  }
\caption{Data-driven cases: Relative Frobenius distances of FADS-P for $\corr$, $\bm\Gamma$, ${\Upsilon}_{\corr}$, ${\Psi}_{\corr}$, $\bm\sigma$ and $\mu$.}
\label{fig-est-tcga}
\end{figure}
\section{Supplement to Section~\ref{sec:app-ch3}}
\subsection{List of words in  \#MeToo tweets}
Figures~\ref{fig-wlist-pic1}--\ref{fig-wlist-pic4} together present the
721 words (in alphabetical order along columns) for the final data
matrix of the \#MeToo tweet data that were collected and analyzed in Section~\ref{sec:app-ch3-metoo}. 
\label{sec:supp-wordlist}
\begin{figure}[H]
\centering
\includegraphics[width = 1\textwidth]{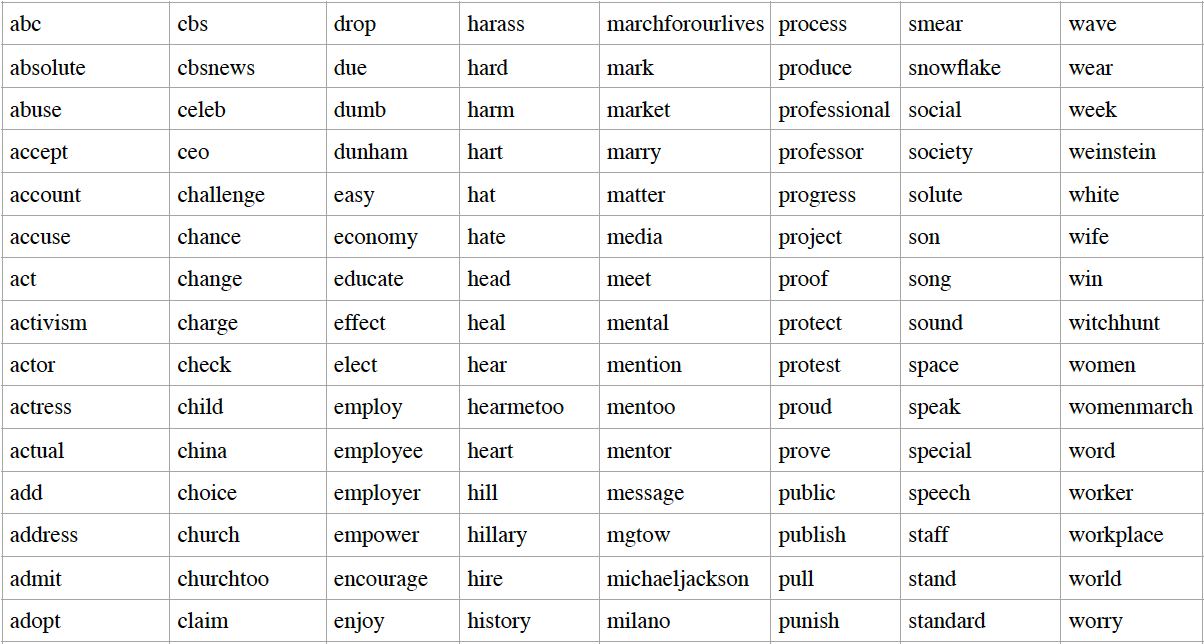}\\
\vspace{-0.05in}
\includegraphics[width = 1\textwidth]{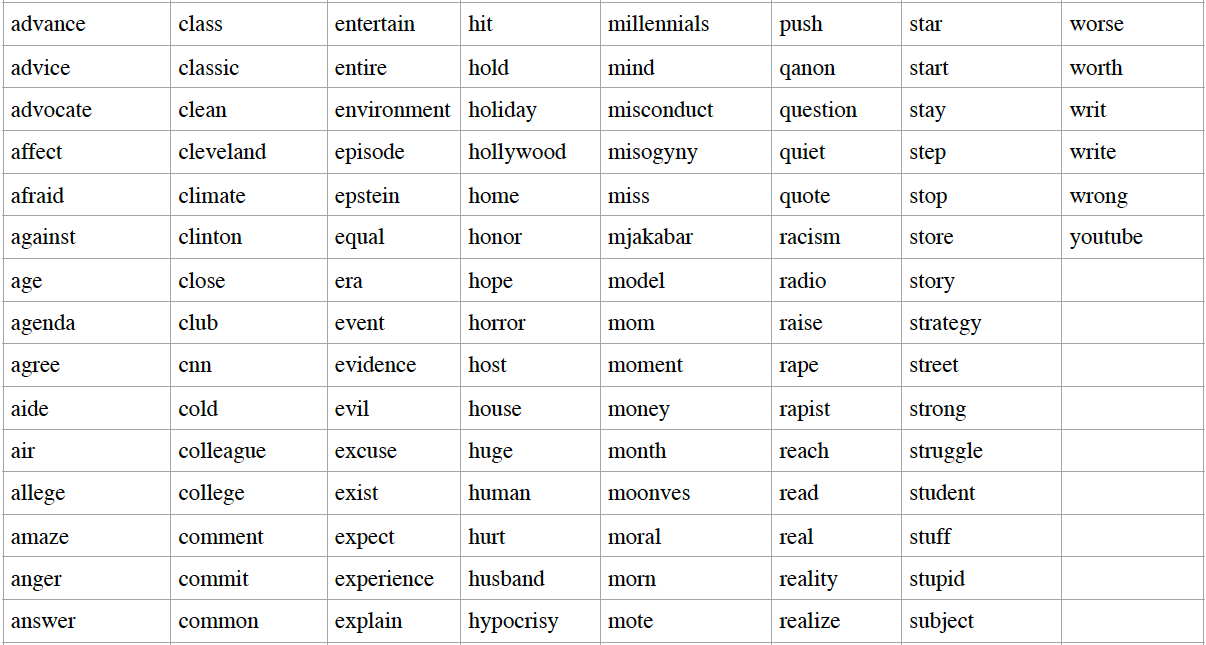}\\
\caption{Word list for \#MeToo tweet data, part 1}
\label{fig-wlist-pic1}
\end{figure}

\begin{figure}[H]
\includegraphics[width = 0.875\textwidth]{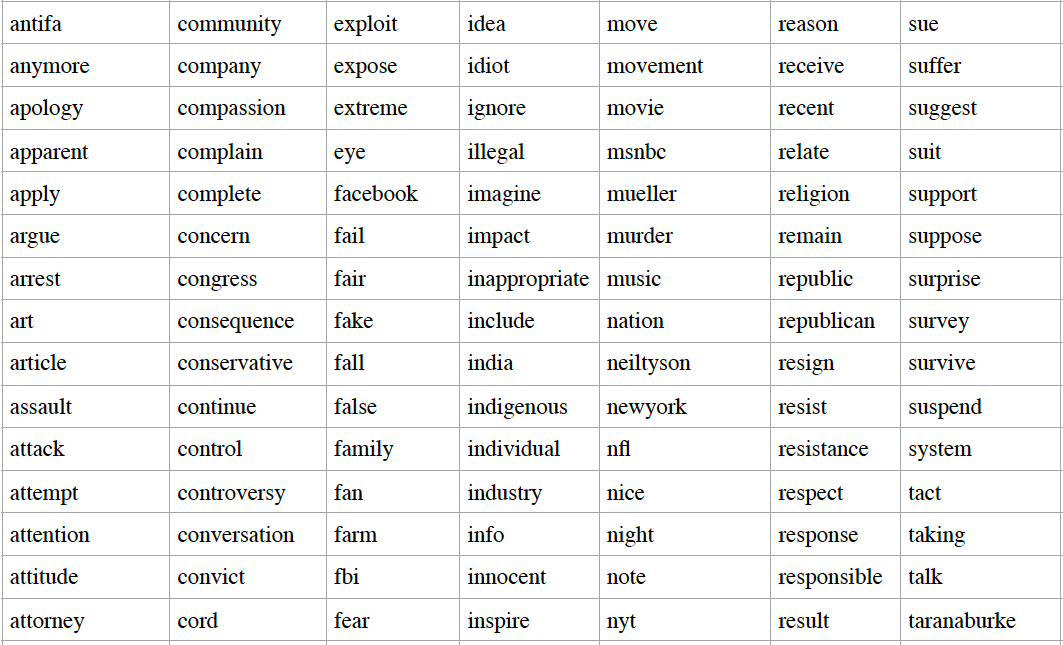}\\
\vspace{-0.05in}
\hspace{-0.08in}
\includegraphics[width = 0.875\textwidth]{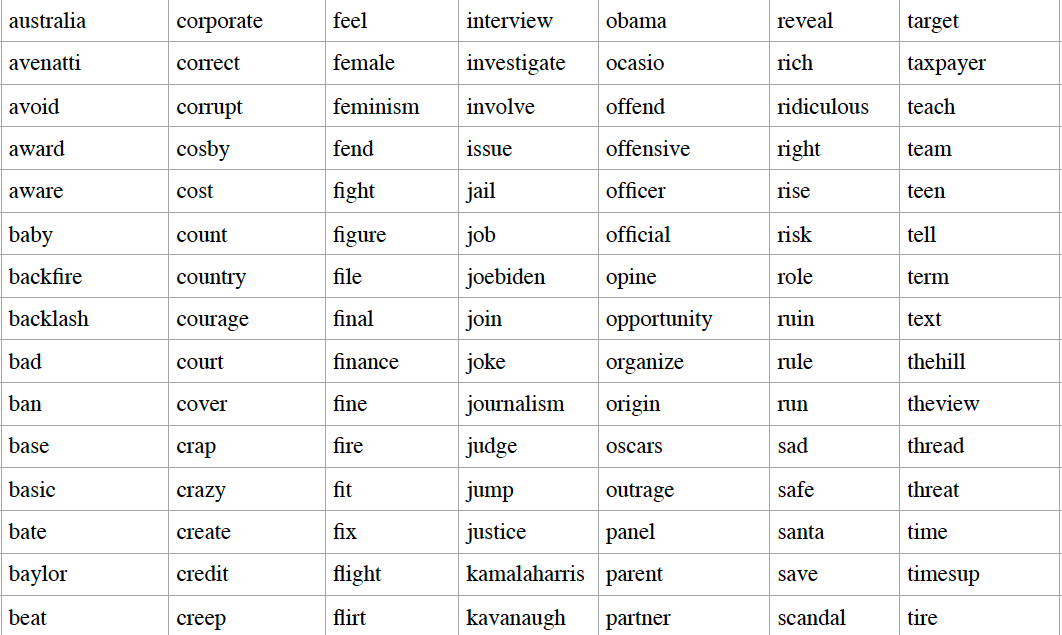}
\caption{Word list for \#MeToo tweet data, part 2}
\label{fig-wlist-pic2}
\end{figure}

\begin{figure}[H]
\includegraphics[width = 0.875\textwidth]{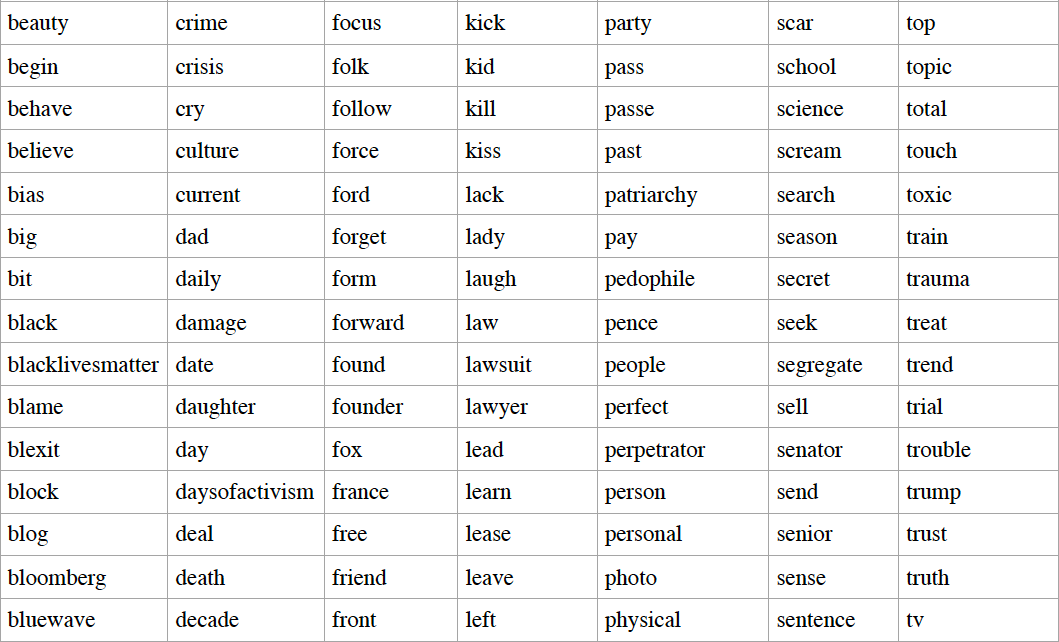}\\
\hspace{-0.15in}
\includegraphics[width = 0.877\textwidth]{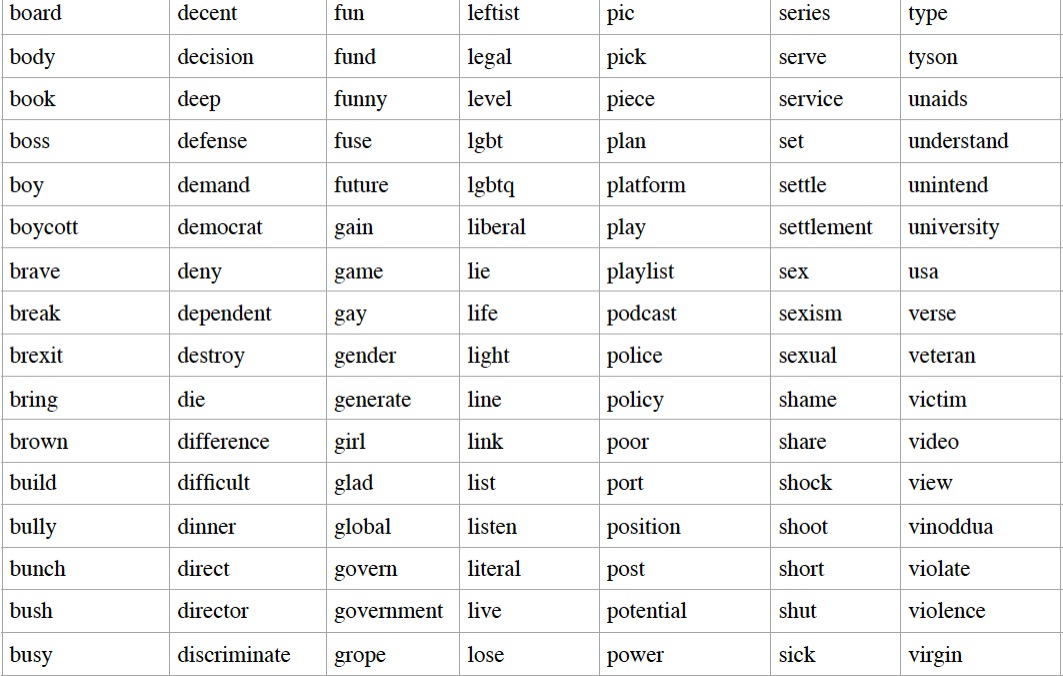}
\vspace{-0.05in}
\caption{Word list for \#MeToo tweet data, part 3}
\label{fig-wlist-pic3}
\end{figure}

\begin{figure}[H]
\includegraphics[width = 0.875\textwidth]{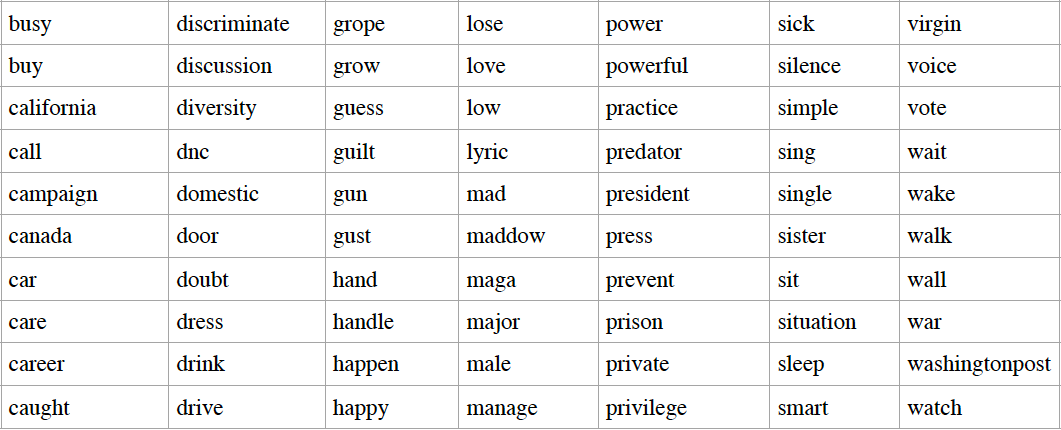}
\vspace{-0.05in}
\caption{Word list for \#MeToo tweet data, part 4}
\label{fig-wlist-pic4}
\end{figure}

\subsection{Handwritten digits}
\begin{figure}[H]
    \centering
    \mbox{
      \subfloat[original sample]{\includegraphics[width =
        0.5\textwidth]{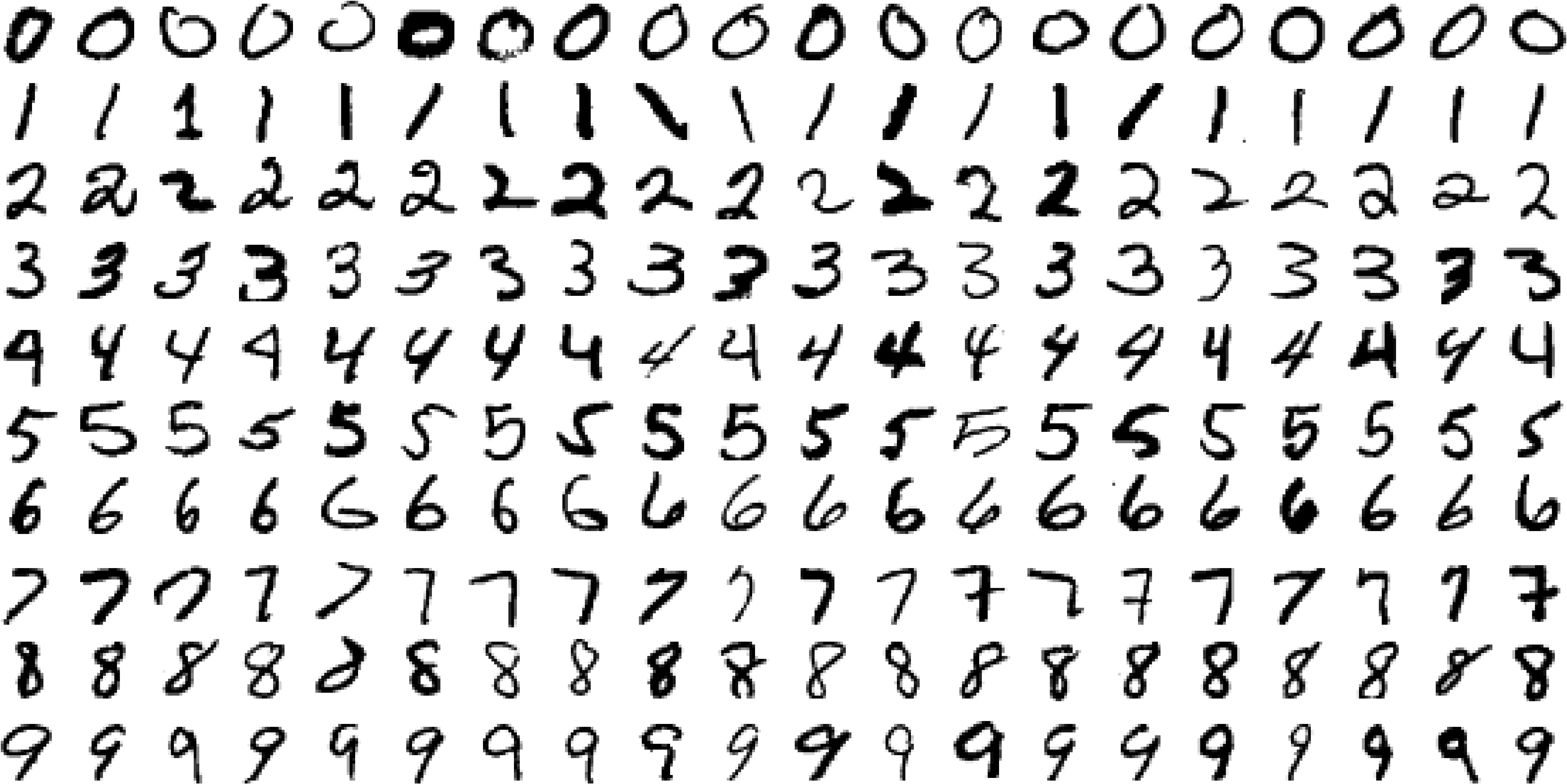}}
      \vrule
      \subfloat[jittered sample]{\includegraphics[width = 0.5\textwidth]{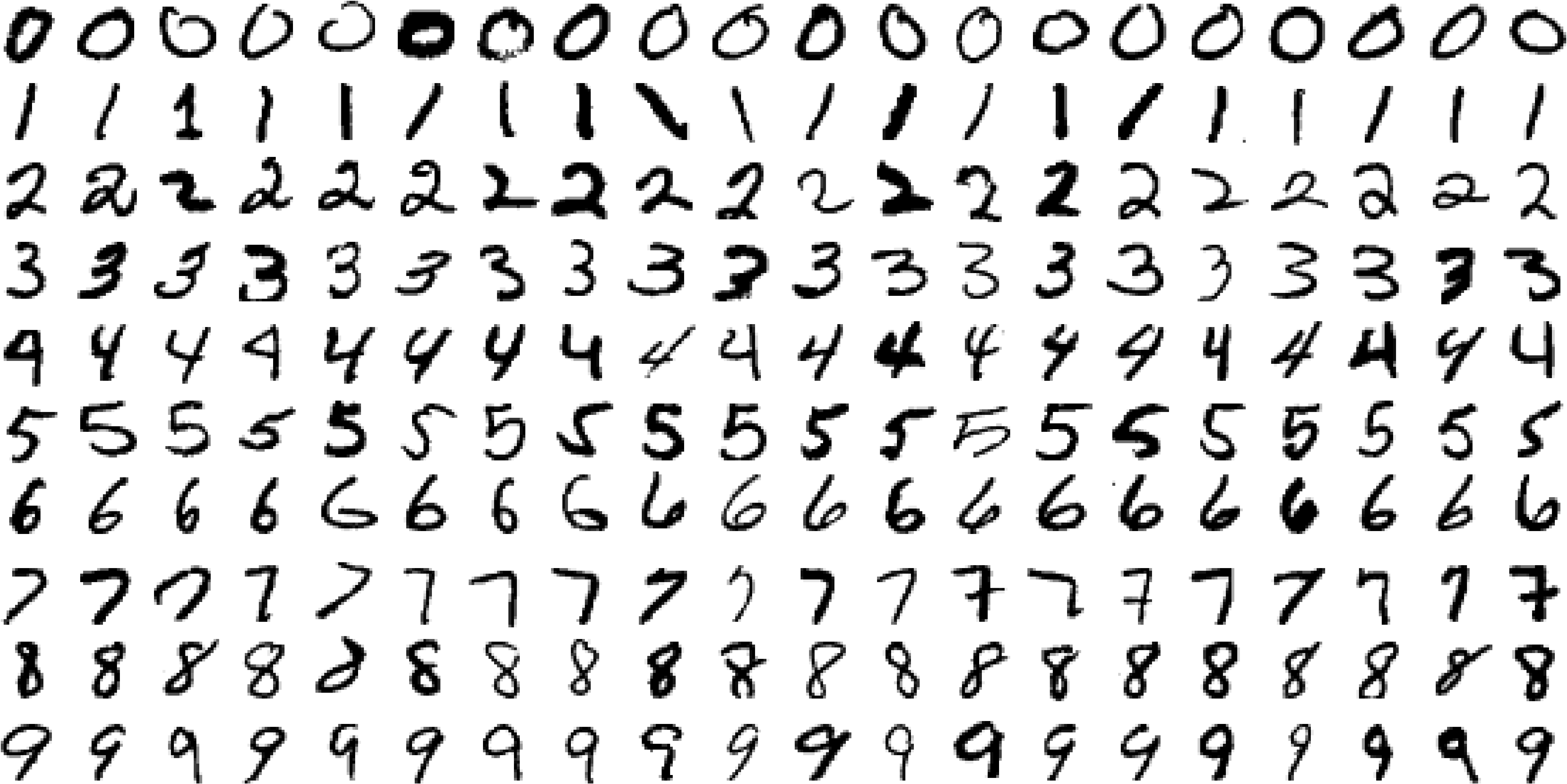}}
    }
    \caption{(a) Original and (b) Jittered images of 20 handwritten digits. The two sets are essentially indistinguishable.} 
    \label{fig:jittered-digits}
\end{figure}
\end{document}